\documentclass{CSML}

\def\dOi{12(3:1)2016}
\lmcsheading%
{\dOi}
{1--55}
{}
{}
{Mar.~\phantom02, 2015}
{Aug.~\phantom01, 2016}
{}

\ACMCCS{[{\bf Theory of computation}]: Logic; Formal languages and
  automata theory; Theory and algorithms for application domains}
\keywords{data-word, counter systems, LTL}
\usepackage{amsthm,amsfonts,amssymb}
\usepackage{url}
\usepackage[bookmarks=false,hidelinks=true]{hyperref}
\usepackage{graphicx}
\usepackage{verbatim}
\usepackage{cleveref}
\usepackage{booktabs}
\usepackage{wrapfig}
\usepackage{float}

 \usepackage{tikz}
 \usetikzlibrary{automata}
 \usetikzlibrary{positioning}
 \usetikzlibrary{arrows}
\usepackage{pgf,pgfarrows,pgfnodes}
\usepackage{pgflibraryshapes}
\usetikzlibrary{automata}
\usetikzlibrary{backgrounds}
\usetikzlibrary{fit}
\usetikzlibrary{calc}
\usetikzlibrary{decorations}
\usetikzlibrary{decorations.pathmorphing}
\usepgflibrary{decorations.pathreplacing}

\theoremstyle{definition}

\newcommand{\altparagraph}[1]{\smallskip\paragraph*{#1}}

\newif\ifLONG\LONGtrue

\floatstyle{plain}
\newfloat{chainsystem}{h}{csm}[section]
\floatname{chainsystem}{Chain system}

\makeatletter{}\newcommand{\set}[1]{\{ #1 \}}

\newcommand{\pair}[2]{\langle #1,#2 \rangle}
\newcommand{\triple}[3]{\langle #1,#2,#3 \rangle}
\newcommand{\tuple}[2]{\langle #1, \ldots, #2 \rangle}

\newcommand{\Nat}{\ensuremath{\mathbb{N}}}

\newcommand{\Zed}{\ensuremath{\mathbb{Z}}}

 \newcommand{\powerset}[1]{2^{#1}}
\newcommand{\nepowerset}[1]{\mathcal{P}^{+}({#1})}

\newcommand{\card}[1]{{\rm card}(#1)}
\newcommand{\amap}{f}

\newcommand{\aset}{X}

\newcommand{\asetter}{Z}
\newcommand {\length}[1] {\ensuremath{|#1|}}

\newcommand{\aconstraint}{\varphi}

 \newcommand{\aformula}{\phi} \newcommand{\aformulabis}{\psi} \newcommand{\aformulater}{\varphi}

\newcommand{\mynext}{{\sf X}}
\newcommand{\previous}{{\sf X}^{-1}}
\newcommand{\until}{{\sf U}}
\newcommand{\sometimes}{{\sf F}}
\newcommand{\pastsometimes}{{\sf F}^{-1}}
\newcommand{\always}{{\sf G}}
\newcommand{\since}{{\sf S}}

\newcommand{\dcltl}{\mainlogic}
\newcommand{\ddcltl}{\pmainlogic}
\newcommand{\mainlogic}{\textup{LRV}}
\newcommand{\pmainlogic}{\textup{PLRV}}

\newcommand{\amodel}{\sigma}

\newcommand{\aletter}{a}
\newcommand{\aletterbis}{b}
\newcommand{\aalphabet}{\Sigma}

\newcommand{\aautomaton}{{\mathcal A}}
\newcommand{\aautomatonbis}{{\mathcal B}}
\newcommand{\astate}{q}
\newcommand{\states}{Q}

\newcommand{\counters}{C}
\newcommand{\ainitst}{\astate_{0}}
\newcommand{\afinst}{\astate_{f}}

\newcommand{\aupfunc}{\vect{u}}
\newcommand{\incer}{\mathit{incer}}

\newcommand{\egdef}{\stackrel{\mbox{\begin{tiny}def\end{tiny}}}{=}} \newcommand{\eqdef}{\stackrel{\mbox{\begin{tiny}def\end{tiny}}}{\equiv}} \newcommand{\equivdef}{\stackrel{\mbox{\begin{tiny}def\end{tiny}}}{\equivaut}} \newcommand{\equivaut}{\;\Leftrightarrow\;}

\newcommand{\step}[1]{\xrightarrow{\!\!#1\!\!}}
\newcommand{\stepinc}[1]{\xrightarrow{\!\!#1\!\!}_{\rm \small  gainy}}
\newcommand{\stepdec}[1]{\xrightarrow{\!\!#1\!\!}_{\rm \small lossy}}

\newcommand {\pspace} {\textsc{pspace}}

\newcommand {\expspace} {\textsc{expspace}}
\newcommand {\twoexpspace} {\textsc{2expspace}}

\newcommand{\symbval}{\mathit{fr}}

\newcommand{\symbmodels}{\models_{{\rm symb}}}

\newcommand{\FFrame}{\mathtt{FFrame}}
\newcommand{\numvars}{k}

\newcommand{\runlen}{n}

\newcommand{\Var}{{\rm VAR}}

\newcommand{\vect}[1]{\mathbf{#1}}

\newcommand{\Aphi}{\aautomaton_\aformula}
\newcommand{\Asymb}{\aautomaton_{\rm symb}}
\newcommand{\Asat}{\aautomaton_{\rm real}}
\newcommand{\Aosc}{\aautomaton_{\rm 1sc}}
\newcommand{\Ainc}{\aautomaton_ {\rm inc}}
\newcommand{\Adec}{\aautomaton_{\rm dec}}

\newcommand{\union} {\ensuremath{\cup}}
\newcommand{\intersection} {\ensuremath{\cap}}

\newcommand{\acountseq}{\beta}
\newcommand{\acountval}{\vect{v}}

\newcommand{\fo}[2]{{\rm FO}^{#1}(#2)}

\newcommand{\avariable}{\mathtt{x}}
\newcommand{\avariablebis}{\mathtt{y}}
\newcommand{\avariableter}{\mathtt{z}}

\newcommand{\defstyle}[1]{\emph{#1}}
\newcommand{\macstyle}[1]{\textbf{#1}}
\newcommand{\commstyle}[1]{\textsl{\scriptsize /* #1 */}}

\newcommand{\tup}[1]{\langle #1 \rangle}
\newcommand{\tupz}[1]{\left\langle #1 \right\rangle}
\newcommand{\ptup}[1]{\langle #1 \rangle^{-1}}
\newcommand{\dcltltup}{\dcltl_{\textit{vec}}}
\newcommand{\oblieq}[3]{#1 \approx \tup{#3 ?}#2}
\newcommand{\oblieqz}[3]{#1 \approx \tupz{#3 ?}#2}
\newcommand{\oblineq}[3]{#1 \not\approx \tup{#3 ?}#2}
\newcommand{\poblieq}[3]{#1 \approx \ptup{#3 ?} #2}
\newcommand{\poblineq}[3]{#1 \not\approx \ptup{#3 ?} #2}

\newcommand{\oblieqlocal}{\approx}
\newcommand{\oblineqlocal}{\not\approx}
\newcommand{\noblieqlocal}{\not\approx}

\newcommand{\N}{\Nat}

\newcommand{\inc}[1]{{\rm inc}(#1)}
\newcommand{\dec}[1]{{\rm dec}(#1)}

\newcommand{\cnext}[1]{{\rm next}(#1)}
\newcommand{\cprevious}[1]{{\rm prev}(#1)}
\newcommand{\cisfirst}[1]{{\rm first(#1)?}}
\newcommand{\cisnotfirst}[1]{\overline{{\rm first}(#1)}{\rm ?}}
\newcommand{\cislast}[1]{{\rm last(#1)?}}
\newcommand{\cisnotlast}[1]{\overline{{\rm last}(#1)}{\rm ?}}

\newcommand{\locations}{Q}
\newcommand{\alocation}{q}
\newcommand{\alocationbis}{r}
\newcommand{\transitions}{\delta} 
\newcommand{\instructions}{I}
\newcommand{\ainstruction}{{\rm instr}}
\newcommand{\achain}{\alpha}
\newcommand{\arun}{\rho}
\newcommand{\atransition}{t}

\newcommand{\interval}[2]{[#1,#2]}
\newcommand{\acounter}{\mathtt{c}}

\newcommand{\vertuple}[4]{
  \left(
  \begin{array}{c}
    #1 \\ #2 \\ #3 \\ #4
  \end{array}
  \right)
}
\newcommand{\odd}{O}
\newcommand{\even}{E}
\newcommand{\oddin}{O_{\mathit{in}}}
\newcommand{\evenen}{E_{\mathit{fi}}}
\newcommand{\avariablebp}{\avariableter_{b}^{1}}
\newcommand{\avariableep}{\avariableter_{b}^{2}}

\newcommand{\reverse}[1]{\widetilde{#1}}
\newcommand{\aword}{u}

\newcommand{\avass}{\mathcal{A}}

\newcommand{\apcpinst}{\mathit{pcp}}

\newcounter{openproblem}
 \def\theopenproblem{\arabic{openproblem}}

\newcommand{\emsot}{\textup{EMSO}\ensuremath{{}^2(+1,<,\sim)}}

\newcommand{\emsotraneq}{\textup{EMSO}\ensuremath{{}^2(<,\sim)}}
\newcommand{\emsotk}{\textup{EMSO}\ensuremath{{}^2(+1, \dotsc, +k,<,\sim)}}
\newcommand{\emsotw}{\textup{EMSO}\ensuremath{{}^2(+1,<)}}
\newcommand{\femsot}{\textit{forward-}\emsot}
\newcommand{\femsotk}{\textit{forward-}\emsotk}
\newcommand{\coloneqq}{::=}

\newcommand{\X}{\ensuremath{{\sf X}}}
\newcommand{\F}{\ensuremath{{\sf F}}}
\newcommand{\U}{\ensuremath{{\sf U}}}

\newcommand{\ltlrepeat}{\textsf{LTL}(\U, \F_=, \F_{\neq}, \set{\X^k_=, \X^k_{\neq}}_{k \in \Nat})}
\newcommand{\G}{\ensuremath{{\sf G}}}

\newcommand{\cut}[1]{}

\newcommand{\D}{\mathbb{D}}

\newcommand{\pcpdir}{MPCP$^{\textit{dir}}$}
\newcommand{\minupfunc}{\mathit{minup}}
\newcommand{\SAT}[1]{{\rm SAT(}#1{\rm )}}

\newcommand{\size}[1]{| #1 |}

\newcommand{\variabs}{\textit{vars}}
\newcommand{\subft}[1]{\textit{sub}_{\tup{ }}\!(#1)}

\newcommand{\mainlogiceq}{\mainlogic^\approx}
\newcommand{\pmainlogiceq}{\pmainlogic^\approx}

\newcommand{\adatum}{\mathfrak{d}}
\newcommand{\adatumbis}{\mathfrak{e}}

\newcommand{\maxupfunc}{\mathit{maxup}}

\newcommand{\counternb}{D} 
\makeatletter{}\usetikzlibrary{shapes.multipart}
\usetikzlibrary{shapes.geometric}

\newlength{\ml}
\tikzstyle{atom}=[rectangle split, rectangle split draw splits =
false, rectangle split ignore empty parts, draw=black, rounded
corners=0.05\ml, minimum width=0.4\ml, inner sep=0.05\ml]

\tikzstyle{state}=[rectangle, draw=black, rounded corners=0.05\ml,
minimum width=0.4\ml, minimum height=0.2\ml]

\tikzstyle{flchop}=[rectangle, draw=black, rounded corners=0.05\ml,
minimum width=0.4\ml, minimum height=0.2\ml]

\tikzstyle{flchcnd}=[diamond, draw=black, aspect=2]

\tikzstyle{traceop}=[rectangle, draw=black,
minimum width=0.4\ml, minimum height=0.2\ml]
 
\usetikzlibrary{fit}

\newif\ifLONG\LONGtrue 
\begin{document}

\title[Reasoning about Data Repetitions with Counter Systems]
      {Reasoning about Data Repetitions \\with Counter Systems\rsuper*}

\author[S.~Demri]{St{\'e}phane Demri\rsuper a}	
\address{{\lsuper a}LSV, ENS Cachan, CNRS, Universit\'e Paris-Saclay,
  94235 Cachan, France}
	
\author[D.~Figueira]{Diego Figueira\rsuper b}
\address{{\lsuper b}LaBRI, CNRS, 33405 Talence, France}	

\author[M.~Praveen]{M.~Praveen\rsuper c}	
\address{{\lsuper c}CMI, Chennai, India}	
\titlecomment{{\lsuper*}Complete version of~\cite{Demri&Figueira&Praveen13}.}
\thanks{
Work supported by project
REACHARD ANR-11-BS02-001, and FET-Open Project FoX, grant agreement
233599. M. Praveen was supported by the ERCIM ``Alain Bensoussan''
Fellowship Programme.}
	
\begin{abstract}\makeatletter{}We study linear-time temporal logics interpreted over data words with
multiple attributes. We restrict the atomic formulas to equalities of
attribute values in successive positions and to repetitions of
attribute values in the future or past.  We demonstrate
correspondences between satisfiability problems for logics and
reachability-like decision problems for counter systems. We show that
allowing/disallowing atomic formulas expressing repetitions of values
in the past corresponds to the reachability/coverability problem in
Petri nets. This gives us \twoexpspace{} upper bounds for several
satisfiability problems.  We prove matching lower bounds by
reduction from a reachability problem for a newly introduced class of
counter systems. This new class is a succinct version of vector
addition systems with states in which counters are accessed via pointers, a
potentially useful feature in other contexts.  We strengthen further
the correspondences between data logics and counter systems by
characterizing the complexity of fragments, extensions and variants of
the logic. For instance, we precisely characterize the relationship
between the number of attributes allowed in the logic and the number
of counters needed in the counter system.
\cut{
We study linear-time temporal logics interpreted over data words 
with multiple attributes in which atomic formulas are restricted to  equalities between
values of successive positions and to repetitions of values in the future or past. 
We show correspondences between logical satisfiability problems and reachability-like
decision problems for counter systems, as the standard problems for Petri nets. 
We demonstrate that allowing/disallowing atomic formulas 
expressing past repetitions of values corresponds to the reachability/coverability problem in Petri nets. 
This gives us an 2EXPSPACE upper bounds for several logical satisfiability problems. 
We prove matching lower bounds by reduction from a reachability problem for a newly introduced class of 
counter systems (a succinct version of VASS) 
for which counters are accessed via pointers, a potentially useful feature in other contexts.
We strenghen further the correspondences between data logics and counter systems by
characterizing the complexity of fragments, extensions and variants.
For instance, we prove that reasoning about pairs of repeating values leads to 
undecidability and we characterize the complexity of 
several fragments of first-order data logics.
}

\cut{
We study the satisfiability problem of temporal logics that reason
about repetitions of values in words over infinite alphabets. We
reduce the satisfiability problem to different decision problems for
Petri nets, depending on the expressivity of the logic.

We focus on a linear temporal logic $\mainlogic$ (``Logic of
Repeating Values'') interpreted over words with multiple data in which
atomic formulae can state either local (in)equalities or
(non)repetitions of values in the future (future obligations).  The
logic $\mainlogic$ extends known logics for which satisfiability is
shown decidable by reducing them to the reachability problem for Petri
nets.  We show that adding the ability to reason about past
obligations makes the satisfiability problem equivalent to the
reachability problem for Petri nets. Without past obligations, the
satisfiability problem is equivalent to the covering problem for Petri
nets. In the reduction to the covering problem, there is an
exponential blow up. We show that this is unavoidable by formalising
the power of $\mainlogic$ to reason about subsets of variables. We
show that allowing the logic to reason about repetitions of two
variables simultaneously will lead to undecidability.
\ifLONG
We also show \pspace-completeness for several $\mainlogic$ fragments
and undecidability when repetitions of values can be done by pairs.
All the decidability and complexity results can be lifted to the case
with infinite words with multiple data or to the case with a finite
amount of MSO-definable temporal operators.  Finally, we take
advantage of these results to establish complexity results for
fragments of first-order logic over data words and to fragments of
freeze LTL with a single register. 
\else
Other results for fragments, extensions or variants are presented in
the paper. 
\fi
}
 
\end{abstract}\maketitle
\makeatletter{}\section{Introduction}
\label{section-introduction}
\paragraph{\bf Words with multiple data}
Finite data words~\cite{Bouyer02} are ubiquitous structures that include  timed words, 
runs of counter automata or runs of concurrent programs with an unbounded number
of processes. These are finite words in which every position carries a label from a finite alphabet and a 
data value from some infinite alphabet.
More generally, structures over an infinite alphabet
provide an adequate abstraction for objects from several domains: 
for example, infinite runs of counter automata
can be viewed as infinite data words, 
finite arrays are finite data words~\cite{Alur&Cerny&Weinstein12},
finite data trees model XML documents with 
attribute values~\cite{Figueira10} and so on.
A wealth of specification formalisms for data words (or slight variants)
has been introduced stemming from automata, see e.g.~\cite{Neven&Schwentick&Vianu04,Segoufin06},
to adequate logical languages such as first-order logic~\cite{BDMSS06:journal,David09}
or temporal 
logics~\cite{Lisitsa&Potapov05,Kupferman&Vardi06,Laz06,Figueira10,Kara&Schwentick&Zeume10,Figueira11}
(see also a related formalism in~\cite{Fitting02}).
Depending on the type of structures, other formalisms have been considered such as
XPath~\cite{Figueira10} or monadic second-order logic~\cite{Bolligetal12}.  
In full generality, most formalisms lead to undecidable decision problems and a well-known
research trend 
consists of finding a good trade-off
between expressiveness and decidability. Restrictions to regain decidability are protean:
bounding the models (from trees to words for instance), 
restricting the number of variables, see e.g.~\cite{BDMSS06:journal},
limiting the set of the temporal operators or the use of the data manipulating operator, 
see e.g.~\cite{Figueira&Segoufin09,Demri&DSouza&Gascon12}. 
As far as classes of automata for data languages  are concerned, other questions arise
related to closure properties or to logical characterisations,
see e.g.~\cite{BDMSS06:journal,Bojanczyk&Lasota10,Kara&Schwentick&Tan12}. 
Moreover, 
interesting and surprising results have been exhibited about
relationships between
logics for data words and counter  
automata (including vector addition systems with 
states)~\cite{BDMSS06:journal,DL-tocl08,Bojanczyk&Lasota10},
leading to a first classification
of automata on data words~\cite{Bojanczyk&Lasota10,Bollig11}. 
This is why logics for 
data words are not only interesting for their own sake but also for their deep
relationships with data automata or with counter automata. Herein, we  pursue
further this line of work and we work with words in which every
position contains a \emph{vector} of data values.\paragraph{\bf Motivations}
In~\cite{Demri&DSouza&Gascon12}, a decidable linear-time temporal  logic interpreted over 
(finite or infinite) sequences of variable valuations (understood as words with multiple data)
is introduced in which the atomic formulae
are of the form either $\avariable \oblieqlocal \mynext^i \avariablebis$ or 
$\oblieq{\avariable}{\avariablebis}{\top}$. 
The formula $\avariable \oblieqlocal \mynext^i \avariablebis$ 
states that the current value of variable $\avariable$ is the same as the value of $\avariablebis$
$i$ steps ahead (local constraint) whereas 
$\oblieq{\avariable}{\avariablebis}{\top}$ states that the current 
value of $\avariable$
is repeated in a future value of $\avariablebis$ (future obligation). 
Such atomic properties can be naturally expressed with a \emph{freeze}  
operator that stores a data value for later comparison, and in~\cite{Demri&DSouza&Gascon12}, it is shown that
the satisfiability problem is decidable with the  temporal operators 
in $\set{\mynext, \previous, \until, \since}$. 
The freeze operator allows to store a data value in a register and then to test later 
equality between the value in the register and a data value at some other position.
This is a powerful mechanism but the logic in \cite{Demri&DSouza&Gascon12} uses it in a limited way: 
only repetitions of data values can be 
expressed and 
it restricts very naturally the use of the freeze operator. 
The decidability result is  robust
since it holds for finite or infinite 
sequences, for any set of MSO-definable temporal
operators and 
with the addition of atomic formulas of the form 
$\poblieq{\avariable}{\avariablebis}{\top}$ stating that 
 the current value of $\avariable$
is repeated in a past value of $\avariablebis$ (past obligation). 
Decidability can be shown either by reduction into $\fo{2}{\sim,<,+\omega}$, a first-order logic
over data words introduced in~\cite{BDMSS06:journal} or by reduction into the verification
of fairness properties in Petri nets, 
shown decidable in~\cite{Jancar95}. In both cases, an essential use  
of the decidability of the reachability problem for Petri nets is made, for which no primitive
recursive algorithm is known, see e.g.~\cite{Leroux11} (see also a first upper bound 
established recently in~\cite{Leroux&Schmitz15}).
 Hence, even though the logics shown 
decidable 
in~\cite{Demri&DSouza&Gascon12} poorly use the freeze 
operator 
(or equivalently, the only
properties about data are related to controlled repetitions), 
the complexity of their
satisfiability problems is unknown. Moreover, it is unclear whether the reductions into
the reachability problem for Petri nets are really 
needed; this would be the case if
reductions in the other direction exist.  
Note that in~\cite{Kara&Schwentick&Zeume10}, a richer logic BD-LTL has been introduced
and it has been shown that satisfiability is equivalent to the reachability problem for 
Petri nets. Moreover, in~\cite{Deckeretal14}, two fragments BD-LTL$^{-}$ and BD-LTL$^{+}$
of that richer logic
have been introduced and shown to admit \twoexpspace{}-complete satisfiability problems.
Forthcoming logic $\mainlogic^{\top}$ is shown in~\cite{Deckeretal14} to be strictly
less expressive than BD-LTL$^{+}$.

Our main motivation is to investigate logics that  express
repetitions of values, revealing the correspondence between
expressivity of the logic and reachability problems for counter machines, including
well-known problems for Petri
nets. This work can be seen as a study of the precision with which
counting needs to be done as a consequence of having a mechanism for
demanding ``\emph{the current data value is repeated in the
future/past}''  in a logic. 
Hence, this is not the study of yet another logic, but of a natural feature 
shared by most studied logics on data words~\cite{Demri&DSouza&Gascon12,BDMSS06:journal,DL-tocl08,Figueira&Segoufin09,Kara&Schwentick&Zeume10,Figueira10,Figueira11}: the property of demanding that a data value be repeated. 
We consider different ways in which one can demand the repetition of a value, and study the repercussion in terms of the ``precision'' with which we
need to count in order to solve the satisfiability problem. Our
measurement of precision here  distinguishes the reachability
versus the coverability problem for Petri nets and the number of
counters needed as a function of the number of variables used in the
logic.

\paragraph{\bf Our contribution}
 We introduce the  linear-time temporal logic $\mainlogic$ 
(``Logic of Repeating Values'')
interpreted over finite words with multiple data, equipped with atomic formulas of the form either
$\avariable \oblieqlocal \mynext^i \avariablebis$ or 
$\oblieq{\avariable}{\avariablebis}{\aformula}$ [resp.
$\oblineq{\avariable}{\avariablebis}{\aformula}$],
where  
$\oblieq{\avariable}{\avariablebis}{\aformula}$ [resp.\ $\oblineq{\avariable}{\avariablebis}{\aformula}$] 
states that the current value of $\avariable$
is repeated [resp.\ is not repeated] in some future value of $\avariablebis$ in a position where 
$\aformula$ holds true. 
When we impose $\aformula = \top$, the logic introduced
in~\cite{Demri&DSouza&Gascon12} is obtained and it is denoted by  $\mainlogic^{\top}$
(a different name is used in~\cite{Demri&DSouza&Gascon12}, namely CLTL$^{\mathtt{XF}}$). 
Note that the  syntax
 for future obligations is freely inspired from 
 propositional dynamic logic PDL with its test operator `?'.
Even though $\mainlogic$ contains the past-time temporal operators $\previous$ and $\since$,
it  has no past obligations. We write $\pmainlogic$ to denote the extension
of $\mainlogic$  with past obligations of the form 
$\poblieq{\avariable}{\avariablebis}{\aformula}$ or 
$\poblineq{\avariable}{\avariablebis}{\aformula}$. 
Figure~\ref{figure-logics} illustrates how $\mainlogic$ and variants are compared
to existing data logics. 
\begin{figure}
\begin{center}
\scalebox{0.8}{
\begin{tikzpicture}[->,>=stealth',shorten >=1pt, node distance=1.5cm, thick,auto]

  \tikzstyle{every state}=[fill=white,draw=black,text=black,minimum size=6mm]
  \tikzstyle{stlab}=[node distance=6mm]

  \node (1)    {BD-LTL~\cite{Kara&Schwentick&Zeume10} ($\equiv$ Reach(VASS))};
  \node (2)  [right=1cm of 1] {LTL$_1^{\downarrow}(\mynext, \mynext^{-1},\until,\since)$~\cite{DL-tocl08} (undec.)};
  \node (3)  [below right=.3cm and .5cm of 1] {$\pmainlogic$};  
  \node (4)  [below right=.3cm and .5cm of 3] {$\mainlogic$};  
  \node (5) [left=1.3cm of 4] {$\pmainlogic^{\top}$ =  CLTL$^{\mathtt{XF},\mathtt{XF}^{-1}}$~\cite{Demri&DSouza&Gascon12}}; 
  \node (6) [below=1.3cm of 3] {$\mainlogic^{\top}$ = CLTL$^{\mathtt{XF}}$~\cite{Demri&DSouza&Gascon12}};
  
   \path[-] (1) edge             node {} (3);
   \path[-] (3) edge             node {} (4);
   \path[-] (3) edge             node {} (5);
   \path[-] (5) edge             node {} (6);
   \path[-] (4) edge             node {} (6);
   \path[-] (2) edge node {(1 attribute)} (3);
                       
\end{tikzpicture}}
\end{center}
\caption{Placing $\mainlogic$ and variants in the family of data logics}
\label{figure-logics}
\end{figure}

Our main results are listed below.
\begin{enumerate}[label=\roman*.]
\item We begin where~\cite{Demri&DSouza&Gascon12} stopped: the
  reachability problem for Petri nets is reduced to the satisfiability
  problem of $\pmainlogic$ (i.e., the logic with past obligations).
\item Without past obligations, the satisfiability problem is much
  easier: we reduce the satisfiability problem of
  $\mainlogic^{\top}$ and $\mainlogic$ to the control-state reachability problem for
  VASS, via a detour to a reachability problem on gainy VASS. But the number of counters in the VASS is
  exponential in the number of variables used in the formula. This
  gives us a \twoexpspace{} upper bound.
\item\label{it:chainAutomata} The exponential blow up mentioned above
  is unavoidable: we show a polynomial-time reduction in the
  converse direction, starting from a linear-sized counter machine
  (without zero tests) that can access exponentially many counters.
  This gives us a matching \twoexpspace{} lower bound.
\item Several augmentations to the logic do not alter the complexity: we
  show that complexity is preserved when MSO-definable temporal
  operators are added or when infinite words with multiple data are
  considered.
\item The power of nested testing formulas: we show that the
  complexity of the satisfiability problem for $\mainlogic^{\top}$
  reduces to \pspace{}-complete when the number of variables in the
  logic is bounded by a constant, while the complexity of the
  satisfiability of $\mainlogic$ does not reduce even when only one
  variable is allowed.  Recall that the difference between
  $\mainlogic^{\top}$ and $\mainlogic$ is that the later allows any
  $\aformula$ in $\oblieq{\avariable}{\avariablebis}{\aformula}$ while
  the former restricts $\aformula$ to just $\top$.
\item The power of pairs of repeating values: we show that the
  satisfiability problem of $\mainlogic^{\top}$ augmented with
  $\oblieq{\langle \avariable, \avariablebis \rangle}{\langle
  \avariable', \avariablebis' \rangle}{\top}$ (repetitions of pairs of
  data values) is undecidable, even when $\poblieq{\langle \avariable,
  \avariablebis \rangle}{\langle \avariable', \avariablebis'
  \rangle}{\top}$ is not allowed (i.e., even when past obligations are
  not allowed).
\item Implications for classical logics: we show a
  3\expspace{} upper bound for the satisfiability problem for
  \femsot{} over data words, using  results on $\mainlogic$.
\end{enumerate}
For proving the result mentioned in point~\ref{it:chainAutomata}
above, we introduce a new class of counter machines that we call 
\defstyle{chain systems} and show a key hardness result for them. This class is
interesting for its own sake and could be used in situations where the
power of binary encoding needs to be used. We prove the
$(k+1)$\expspace{}-completeness of the control state reachability
problem for chain systems of level $k$ (we only use $k=1$ in this
paper but the proof for arbitrary $k$ is no more complex than the
proof for the particular case of $k=1$). In chain systems, the number
of counters is equal to an exponential tower of height $k$ but we cannot 
access the counters directly in the transitions. Instead, we have
a pointer that we can move along a chain of counters. The
$(k+1)$\expspace{} lower bound is obtained by a non-trivial extension
of the \expspace{}-hardness result from \cite{Lipton76,Esparza98}. Then we show
that the control state reachability problem for the class of chain
systems with $k=1$ can be reduced to the satisfiability problem for
$\mainlogic$ (see Section~\ref{section-simulating}).
It was known that data logics are strongly related to classes of counter automata, 
      see e.g.~\cite{BDMSS06:journal,DL-tocl08,Bojanczyk&Lasota10} but 
      herein, we show how varying the expressive power of logics leads to
      correspondence with different reachability  problems for counter machines. 
\makeatletter{}\section{Preliminaries}\label{section-preliminaries}\ifLONG
We write $\Nat$ [resp.\ $\Zed$] to denote  the set of non-negative integers [resp.\ integers]
and $\interval{i}{j}$ to denote  the set $\set{k \in \Zed: i \leq k \ {\rm and} \ k \leq j}$.
For every $\vect{v}\in\Zed^n$, $\vect{v}(i)$ 
denotes the $i^{th}$ element of $\vect{v}$ for every $i\in \interval{1}{n}$.
We write $\vect{v} \preceq \vect{v'}$ whenever for every $i \in \interval{1}{n}$, we have
$\vect{v}(i) \leq \vect{v'}(i)$. 
For a (possibly infinite) alphabet $\aalphabet$, $\aalphabet^*$ represents the set of finite words over $\aalphabet$, 
$\aalphabet^+$ the set of finite non-empty words over $\aalphabet$.  
For a finite word $\aword=\aletter_1\ldots \aletter_k$ over $\aalphabet$, 
we write $\length{\aword}$ to denote its \defstyle{length} $k$. 
For every $0 \leq i < \length{\aword}$, $\aword(i)$ represents the $(i+1)$-th letter of the word, here $\aletter_{i+1}$. 
We use $\card{\aset}$ to denote the number of elements of a finite set $\aset$.

\subsection{Logics of Repeating Values}
Let $\Var = \set{\avariable_1, \avariable_2, \dots}$ be a countably infinite set of
\defstyle{variables}. We denote by $\dcltl$ the logic whose formulas are defined
as follows: 
\[
\aformula ::= \avariable \oblieqlocal \mynext^i \avariablebis \mid
\oblieq{\avariable}{\avariablebis}{\aformula} \mid
\oblineq{\avariable}{\avariablebis}{\aformula} 
\mid \aformula \wedge \aformula \mid \neg \aformula \mid \mynext
\aformula \mid \aformula \until \aformula \mid \mynext^{-1} \aformula
\mid \aformula \since \aformula
\]
where $\avariable, \avariablebis \in \Var$ and $i \in \Nat$.
Formulas of one of the forms
$\avariable \oblieqlocal \mynext^i \avariablebis$,
$\oblieq{\avariable}{\avariablebis}{\aformula}$ or 
$\oblineq{\avariable}{\avariablebis}{\aformula}$
are said to be \defstyle{atomic} and an expression
of the form $\mynext^i \avariable$ (abbreviation for $i$ next symbols followed by
a variable) is called a \defstyle{term}. 
 
A \defstyle{valuation} is  a map from $\Var$ to $\Nat$, and a \defstyle{model} is a
finite non-empty sequence $\sigma$ of valuations. 
 All the
subsequent developments can be equivalently done with the domain
$\Nat$ replaced by an infinite set $\D$ since only equality tests are
performed in the logics.  
 
We write $\length{\sigma}$ to denote the length of
$\sigma$. For every
model $\sigma$ and $0 \leq i < \length{\sigma}$, the satisfaction
relation $\models$ is defined inductively as follows.
Note that the temporal operators next ($\mynext$), 
previous ($\previous$), until ($\until$) and since ($\since$) and Boolean connectives are defined in the 
usual way.

\cut{
\begin{itemize}
\itemsep 0 cm 
\item $\sigma, i \models \avariable \oblieqlocal \mynext^{j} \avariablebis$ iff $i+j <
  \length{\sigma}$ and $\sigma(i)(\avariable) = \sigma(i+j)(\avariablebis)$,
\item $\sigma, i \models \oblieq{\avariable}{\avariablebis}{\aformula}$ 
  iff there exists $j$ such that $i < j <
  \length{\sigma}$, $\sigma(i)(\avariable) = \sigma(j)(\avariablebis)$
  and $\sigma, j \models \aformula$,
\item $\sigma, i \models \oblineq{\avariable}{\avariablebis}{\aformula}$ 
  iff there exists $j$ such that $i < j <
  \length{\sigma}$, $\sigma(i)(\avariable) \neq \sigma(j)(\avariablebis)$
  and $\sigma, j \models \aformula$,
\item $\sigma,i \models \aformula \wedge \aformula'$ iff $\sigma,i
  \models \aformula$ and $\sigma,i \models \aformula'$,
\item $\sigma,i \models \neg \aformula$ iff $\sigma,i \not \models
  \aformula$,
\item $\sigma,i \models \mynext \aformula$ iff $i+1 <
  \length{\sigma}$ and $\sigma,i+1 \models \aformula$,
\item $\sigma,i \models \mynext^{-1} \aformula$ iff $i>0$ and 
  $\sigma,i-1 \models \aformula$, 
\item $\sigma,i \models \aformula \until \aformula'$ iff there is $i
  \leq j < \length{\sigma}$ such that $\sigma,j \models
  \aformula'$ and for every $i \leq l <j,$ we have $\sigma,l \models
  \aformula$.
\item $\sigma,i \models \aformula \since \aformula'$ iff there is $0
  \leq j \leq i$ such that $\sigma,j \models \aformula'$ and for every 
  $j < l \leq i$ we have $\sigma,l \models \aformula$.
\end{itemize}
}
\begin{center}
\begin{tabular}{lcl}
$\sigma, i \models \avariable \oblieqlocal \mynext^{j} \avariablebis$ & iff & $i+j <
  \length{\sigma}$ and $\sigma(i)(\avariable) = \sigma(i+j)(\avariablebis)$ \\

$\sigma, i \models \oblieq{\avariable}{\avariablebis}{\aformula}$  & iff &
   there exists $j$ such that $i < j <
  \length{\sigma}$, \\
& & $\sigma(i)(\avariable) = \sigma(j)(\avariablebis)$
  and $\sigma, j \models \aformula$\\

$\sigma, i \models \oblineq{\avariable}{\avariablebis}{\aformula}$  & iff &
there exists $j$ such that $i < j <
  \length{\sigma}$, \\
& & $\sigma(i)(\avariable) \neq \sigma(j)(\avariablebis)$
  and $\sigma, j \models \aformula$ \\

$\sigma,i \models \aformula \wedge \aformula'$ & iff & $\sigma,i
  \models \aformula$ and $\sigma,i \models \aformula'$ \\

$\sigma,i \models \neg \aformula$ & iff & $\sigma,i \not \models
  \aformula$ \\

$\sigma,i \models \mynext \aformula$ & iff & $i+1 <
  \length{\sigma}$ and $\sigma,i+1 \models \aformula$ \\

$\sigma,i \models \mynext^{-1} \aformula$ & iff & $i>0$ and
  $\sigma,i-1 \models \aformula$ \\

$\sigma,i \models \aformula \until \aformula'$ & iff & there is $i
  \leq j < \length{\sigma}$ such that $\sigma,j \models
  \aformula'$ and \\
& &  for every $i \leq l <j,$ 
      we have $\sigma,l \models
  \aformula$ \\ 

$\sigma,i \models \aformula \since \aformula'$ & iff & there is $0
  \leq j \leq i$ such that $\sigma,j \models \aformula'$ and  \\
& & for every 
  $j < l \leq i$ we have $\sigma,l \models \aformula$.
\end{tabular}
\end{center}

\noindent We write $\sigma \models \aformula$ if $\sigma, 0 \models \aformula$.
We  use the standard derived temporal operators ($\always$,
$\sometimes$, $\pastsometimes$, \ldots), and derived Boolean operators
($\vee$, $\Rightarrow$, \ldots) and constants $\top$, $\bot$. 

We also use the notation $\mynext^i \avariable \oblieqlocal
\mynext^j \avariablebis$ as an abbreviation for the formula $\mynext^i(\avariable \oblieqlocal \mynext^{j-i}
\avariablebis)$ (assuming without any loss of generality that $i \leq j$).
Similarly, $\mynext^{j} \avariablebis \oblieqlocal  \avariable$ is an abbreviation
for  $ \avariable  \oblieqlocal \mynext^{j} \avariablebis $.

Given a set of temporal operators $\mathcal{O}$ definable from those in
$\set{\mynext,\previous, \since,\until}$ and a natural number $k \geq
0$, we write $\dcltl_k(\mathcal{O})$ to denote the fragment of
$\dcltl$ restricted to formulas with temporal operators from
$\mathcal{O}$ and with at most $k$ variables.
The \defstyle{satisfiability problem} for
$\mainlogic$ (written \SAT{$\mainlogic$}) is to
check for a given $\mainlogic$ formula $\aformula$, whether there exists a 
model $\sigma$ such that $\sigma \models
\aformula$.
Note that there is a logarithmic-space reduction from the satisfiability problem for
$\mainlogic$ to its restriction where atomic formulas of the form 
$\avariable \oblieqlocal \mynext^i \avariablebis$ satisfy $i \in \set{0,1}$ (at the cost of introducing new variables).

Let $\pmainlogic$ be the extension of $\mainlogic$ with
additional atomic formulas of the form $\poblieq{\avariable}{\avariablebis}{\aformula}$ and $\poblineq{\avariable}{\avariablebis}{\aformula}$.
The satisfaction relation is extended as follows:
\begin{center}
\begin{tabular}{lcl}
$\sigma, i \models \poblieq{\avariable}{\avariablebis}{\aformula}$ & iff &
       there is $0 \leq j < i$ such that
       $\sigma(i)(\avariable) = \sigma(j)(\avariablebis)$ and $\sigma, j \models \aformula$ \\
$\sigma, i \models \poblineq{\avariable}{\avariablebis}{\aformula}$ & iff &
       there is $0 \leq j < i$ such that
       $\sigma(i)(\avariable) \neq \sigma(j)(\avariablebis)$ and $\sigma, j \models \aformula$.
\end{tabular}
\end{center}

\noindent We write $\mainlogic^{\top}$ [resp.\ $\pmainlogic^{\top}$]  to denote the fragment of $\mainlogic$ 
[resp.\ $\pmainlogic$] in which atomic formulas are restricted to
$\avariable \oblieqlocal \mynext^i \avariablebis$ and
$\oblieq{\avariable}{\avariablebis}{\top}$ [resp.\ 
$\avariable \oblieqlocal \mynext^i \avariablebis$, $\oblieq{\avariable}{\avariablebis}{\top}$ and 
$\poblieq{\avariable}{\avariablebis}{\top}$].  These are precisely the fragments considered 
in~\cite{Demri&DSouza&Gascon12} and shown decidable by reduction into 
the reachability problem
for Petri nets. 

\begin{prop} \label{proposition-diamond}
\cite{Demri&DSouza&Gascon12}
\begin{enumerate}[label=(\Roman*)]
\item Satisfiability problem for $\mainlogic^{\top}$ is decidable (by reduction to the reachability problem
for Petri nets).
\item Satisfiability problem for $\mainlogic^{\top}$ restricted to a single variable
is \pspace-complete.  
\item Satisfiability problem for $\pmainlogic^{\top}$ is decidable (by reduction to the reachability problem
for Petri nets). 
\end{enumerate}
\end{prop}\smallskip

\noindent In~\cite{Demri&DSouza&Gascon12}, there are no reductions in the directions opposite to (I) and (III). 
The characterisation of the computational complexity for the satisfiability problems for  $\dcltl$ and  $\ddcltl$ 
remained unknown so far and this will be a contribution of the paper.

\subsection{Properties} 
In the table below, we justify
our choices for atomic formulae by presenting several abbreviations
(with their obvious semantics).
By contrast, we include in $\mainlogic$ both $\oblieq{\avariable}{\avariablebis}{\aformula}$
and $\oblineq{\avariable}{\avariablebis}{\aformula}$ when $\aformula$ is an arbitrary formula
since there is no obvious way to express one with the other.

\bigskip

\begin{center}
\begin{tabular}{|c|c|} \hline 
Abbreviation & Definition \\ \hline
$\avariable \oblineqlocal \mynext^i \avariablebis$  & 
$\neg (\avariable \oblieqlocal \mynext^i \avariablebis) \wedge 
\overbrace{\mynext \cdots \mynext}^{i \ {\rm times}} \top$ 
\\ \hline 
$\avariable \oblieqlocal \mynext^{-i} \avariablebis$ & 
$\overbrace{\mynext^{-1} \cdots \mynext^{-1}}^{i \ {\rm times}}(\avariablebis  \oblieqlocal \mynext^i \avariable )$ \\
\hline
$\avariable \oblineqlocal \mynext^{-i} \avariablebis$ & 
$\neg (\avariable \oblieqlocal \mynext^{-i} \avariablebis) \wedge 
\overbrace{\mynext^{-1} \cdots \mynext^{-1}}^{i \ {\rm times}} \top$ 
\\ \hline 
$\oblineq{\avariable}{\avariablebis}{\top}$ &  
$(\avariable \oblineqlocal \mynext \avariablebis) \vee 
       \mynext((\avariablebis \oblieqlocal \mynext \avariablebis) \until 
               (\avariablebis \oblineqlocal \mynext \avariablebis))$ \\ \hline 
 $\poblineq{\avariable}{\avariablebis}{\top}$ & 
$(\avariable \oblineqlocal \mynext^{-1} \avariablebis) \vee 
       \mynext^{-1}(
(\avariablebis \oblieqlocal \mynext^{-1} \avariablebis)
 \since 
 (\avariablebis \oblineqlocal  \mynext^{-1} \avariablebis))$ \\ \hline
\end{tabular}
\end{center}

\bigskip 
\cut{
First, let us come back to our choices for atomic formulae.
In $\mainlogic$, there is no need for 
atomic formulae stating that
`$\avariable$ is not equal to $\avariablebis$ $i$ steps ahead'
since this can be expressed  with 
$\neg (\avariable \oblieqlocal \mynext^i \avariablebis)$.
Similarly,  `$\avariable$ is not equal to $\avariablebis$ $i$ steps backward'
(which could correspond to  $\avariable \oblieqlocal \mynext^{-i} \avariablebis$)
can be expressed
with $\overbrace{\mynext^{-1} \cdots \mynext^{-1}}^{i \ {\rm times}}(\mynext^i \avariable  
       \oblieqlocal \avariablebis)$.
Moreover, in  $\mainlogic^{\top}$, there is no need for atomic formulae of the
form  $\oblineq{\avariable}{\avariablebis}{\top}$
(`$\avariable$ is not repeated in a future value of $\avariablebis$'), since they can be expressed
by  $(\avariable \oblineqlocal \mynext \avariablebis) \vee 
       \mynext((\avariablebis \oblieqlocal \mynext \avariablebis) \until 
               (\avariablebis \oblineqlocal \mynext \avariablebis))$.
The same holds for  $\poblineq{\avariable}{\avariablebis}{\top}$ in  $\pmainlogic^{\top}$.
}

\noindent Models for $\mainlogic$ can be  viewed
as finite data words  in $(\Sigma \times \D)^*$, where $\Sigma$ is a finite alphabet and $\D$ is an 
infinite domain. E.g., equalities between dedicated variables can simulate that a position
is labelled by a letter from $\Sigma$; moreover, we may assume that the data values are encoded with the
variable $\avariable$. Let us express that
whenever there are $i < j$ such that $i$ and $j$ [resp. $i+1$ and $j+1$, $i+2$ and $j+2$]  are 
labelled by $\aletter$ [resp. $\aletter'$, $\aletter''$], 
$\sigma(i+1)(\avariable) \neq \sigma(j+1)(\avariable)$.
This can be stated in $\mainlogic$ by:
$$
\always(
\aletter \wedge \mynext(\aletter' \wedge \mynext \aletter'')
\Rightarrow
\mynext \neg 
(
\oblieq{\avariable}{\avariable}{\mynext^{-1}\aletter \wedge \aletter' \wedge \mynext \aletter''}
)
).
$$
This is an example of key constraints, see e.g.~\cite[Definition 2.1]{Niewerth&Schwentick11}
and the current paper contains also numerous examples of properties that can be captured by $\mainlogic$. 

\subsection{Basics on VASS} 
A \defstyle{vector addition system with states} is a tuple 
$\avass = \triple{\locations}{\counters}{\transitions}$ where
$\locations$ is a finite set of \defstyle{control states},
$\counters$ is  a finite set of \defstyle{counters} and $\transitions$ is a finite
set of \defstyle{transitions} in $\locations \times \Zed^{\counters} \times \locations$. 
A \defstyle{configuration} of $\avass$ is a pair
$\pair{\alocation}{\vect{v}} \in \locations \times
\Nat^{\counters}$. We write $\pair{\alocation}{\vect{v}} \step{} \pair{\alocation'}{\acountval'}$ if
there is \ifLONG a transition \fi $(\alocation, \vect{u}, \alocation')\in \transitions$ such that
$\vect{v'} = \vect{v} + \vect{u}$. Let $\step{*}$ be the reflexive and transitive closure of $\step{}$. 
The reachability problem for VASS (written Reach(VASS)) 
consists of checking whether 
 $\pair{\alocation_0}{\vect{v_0}} \step{*} \pair{\alocation_f}{\vect{v_f}}$, given two configurations
$\pair{\alocation_0}{\vect{v_0}}$ and $\pair{\alocation_f}{\vect{v_f}}$. 
The reachability
problem for VASS is decidable but all known 
algorithms~\cite{Mayr84,Kosaraju82,Lambert92,Leroux11} 
take non-primitive
recursive space in the worst case. The best known lower bound is
\expspace{} \cite{Lipton76,Esparza98} whereas a first upper bound
has been recently established in ~\cite{Leroux&Schmitz15}.
The control state reachability problem consists in checking
whether $\pair{\alocation_0}{\vect{v_0}} \step{*} \pair{\alocation_f}{\vect{v}}$
for some $\vect{v} \in \Nat^{\counters}$, given a configuration  $\pair{\alocation_0}{\vect{v_0}}$
and a control state $\alocation_f$. This problem is known to be 
\expspace-complete~\cite{Lipton76,Rackoff78}. 
The relation $\step{}$ denotes the one-step transition in a perfect computation.
In the paper, we need to introduce computations with gains or with losses. We define below
the variant relations $\stepinc{}$ and $\stepdec{}$. 
We write $\pair{\alocation}{\vect{v}} \stepinc{} \pair{\alocation'}{\vect{v'}}$ if
there is a transition $(\alocation, \vect{u}, \alocation') \in \transitions$ 
and $\vect{w}, \vect{w'} \in \Nat^{\counters}$ such that
$\vect{v} \preceq \vect{w}$, $\vect{w'} = \vect{w} + \vect{u}$ and $\vect{w'} \preceq \vect{v'}$.
Let $\stepinc{*}$ be the reflexive and transitive closure of $\stepinc{}$. 
Similarly, we write $\pair{\alocation}{\vect{v}} \stepdec{} \pair{\alocation'}{\vect{v'}}$ if
there is a transition $(\alocation, \vect{u}, \alocation') \in \transitions$ 
and $\vect{w}, \vect{w'} \in \Nat^{\counters}$ such that
$\vect{w} \preceq \vect{v}$, $\vect{w'} = \vect{w} + \vect{u}$ and $\vect{v'} \preceq \vect{w'}$.
Let $\stepdec{*}$ be the reflexive and transitive closure of $\stepdec{}$. 
Counter automata with imperfect computations such as  
lossy channel systems~\cite{Abdulla&Jonsson96,Finkel&Schnoebelen01},
lossy counter automata~\cite{Mayr03} or gainy counter automata~\cite{Schnoebelen10b}
have been intensively studied (see also~\cite{Schnoebelen10}). In the paper, imperfect
computations are used with VASS in Section~\ref{section-upper-bound}.

\else\section{Preliminaries}
\label{section-preliminaries}We write $\Nat$ [resp.\ $\Zed$] to denote  the set of non-negative integers [resp.\ integers]
and $\interval{i}{j}$ to denote  $\set{k \in \Zed: i \leq k \ {\rm and} \ k \leq j}$.
For $\vect{v}\in\Zed^n$, $\vect{v}(i)$ 
denotes the $i^{th}$ element of $\vect{v}$ for every $i\in \interval{1}{n}$.
We write $\vect{v} \preceq \vect{v'}$ whenever for every $i \in \interval{1}{n}$, we have
$\vect{v}(i) \leq \vect{v'}(i)$. 
For a (possibly infinite) alphabet $\aalphabet$, $\aalphabet^*$ represents the set of finite words over $\aalphabet$, 
$\aalphabet^+$ the set of finite non-empty words over $\aalphabet$.  
For a finite word or sequence $\aword=\aletter_1\ldots \aletter_k$ over $\aalphabet$, 
we write $\length{\aword}$ to denote its \defstyle{length} $k$. 
For $0 \leq i < \length{\aword}$, $\aword(i)$ represents the $(i+1)$-th letter of the word, here $\aletter_{i+1}$. 
\altparagraph{Logics of Repeating Values}
Let $\Var = \set{\avariable_1, \avariable_2, \dots}$ be a countably infinite set of
\defstyle{variables}. We denote by $\dcltl$ the logic whose formulas are defined
as follows, where $\avariable, \avariablebis \in \Var$, $i \in \Nat$.
\begin{align*}
  \aformula &::= \avariable \oblieqlocal \mynext^i \avariablebis \mid
      \oblieq{\avariable}{\avariablebis}{\aformula} \mid
  \oblineq{\avariable}{\avariablebis}{\aformula} \mid \aformula \wedge
  \aformula \mid \neg \aformula \mid \\
&~~~~~\,\mynext \aformula \mid \aformula
  \until \aformula \mid \mynext^{-1} \aformula \mid \aformula \since
  \aformula
\end{align*}
A \defstyle{valuation} is  a map from $\Var$ to $\Nat$, and a \defstyle{model} is a
finite non-empty sequence $\sigma$ of valuations. 
For every model $\sigma$ and $0 \leq i < \length{\sigma}$, the satisfaction
relation $\models$ is defined:
\begin{itemize}
\itemsep 0 cm
\item $\sigma, i \models \avariable \oblieqlocal \mynext^{j} \avariablebis$ if{f} $i+j <
  \length{\sigma}$ and $\sigma(i)(\avariable) = \sigma(i+j)(\avariablebis)$,
\item $\sigma, i \models \oblieq{\avariable}{\avariablebis}{\aformula}$ 
  if{f} there exists $j$ such that $i < j <
  \length{\sigma}$, $\sigma(i)(\avariable) = \sigma(j)(\avariablebis)$
  and $\sigma, j \models \aformula$.
\end{itemize}
The semantics 
$\oblineq{\avariable}{\avariablebis}{\aformula}$ 
is defined similarly but asking for a different data value. 
The temporal operators next ($\mynext$), 
previous ($\previous$), until ($\until$) and since ($\since$) and Boolean connectives are defined in the 
usual way. We also use the standard derived temporal operators ($\always$,
$\sometimes$, $\pastsometimes$, \ldots) and constants $\top$, $\bot$. 
We write $\sigma \models \aformula$ if $\sigma, 0 \models \aformula$.
\ifLONG
We also use the notation $\mynext^i \avariable \oblieqlocal
\mynext^j \avariablebis$ as an abbreviation for the formula $\mynext^i(\avariable \oblieqlocal \mynext^{j-i}
\avariablebis)$ (assuming without any loss of generality that $i \leq j$).
\fi 
Given a set of temporal operators $\mathcal{O}$,
\ifLONG definable from those in $\set{\mynext,\previous, \since,\until}$ \fi  
we write $\mainlogic(\mathcal{O})$ to denote the fragment of
$\dcltl$ restricted to formulas with  operators from
$\mathcal{O}$.
The \defstyle{satisfiability problem} for
$\mainlogic$  (written \SAT{$\mainlogic$}) is to
check for an $\mainlogic$ formula $\aformula$, whether there is  
$\sigma$ such that $\sigma \models
\aformula$.
\ifLONG  
Note that there is a logarithmic-space reduction from the satisfiability problem for
$\mainlogic$ to its restriction where atomic formulas of the form 
$\avariable \oblieqlocal \mynext^i \avariablebis$ satisfy $i \in \set{0,1}$ (at the cost of introducing new variables).
\fi 
Let $\pmainlogic$ be the extension of $\mainlogic$ with
additional atomic formulas of the form $\poblieq{\avariable}{\avariablebis}{\aformula}$ and $\poblineq{\avariable}{\avariablebis}{\aformula}$.
The satisfaction relation is extended as expected:
$\sigma, i \models \poblieq{\avariable}{\avariablebis}{\aformula}$ if{f}
       there is $0 \leq j < i$ such that
       $\sigma(i)(\avariable) = \sigma(j)(\avariablebis)$ and $\sigma, j \models \aformula$,
and similarly for $\poblineq{\avariable}{\avariablebis}{\aformula}$.
We write $\mainlogic^{\top}$ [resp.\ $\pmainlogic^{\top}$]  to denote the fragment of $\mainlogic$ 
[resp.\ $\pmainlogic$] in which atomic formulas are restricted to  
$\avariable \oblieqlocal \mynext^i \avariablebis$ and
$\oblieq{\avariable}{\avariablebis}{\top}$ [resp. to $\avariable \oblieqlocal \mynext^i \avariablebis$, 
$\oblieq{\avariable}{\avariablebis}{\top}$ and
$\poblieq{\avariable}{\avariablebis}{\top}$].  In~\cite{Demri&DSouza&Gascon12}, $\mainlogic^{\top}$ and $\pmainlogic^{\top}$ were shown to be decidable by reduction into 
the reachability problem
for Petri nets. However, the characterisation of their complexity remained open. 
\altparagraph{Properties} In the table below, we justify
our choices for atomic formulae by presenting several abbreviations
(with their obvious semantics).
By contrast, we include in $\mainlogic$ both $\oblieq{\avariable}{\avariablebis}{\aformula}$
and $\oblineq{\avariable}{\avariablebis}{\aformula}$ when $\aformula$ is an arbitrary formula
since there is no obvious way to express one with the other.
{\small 
\begin{center}
\begin{tabular}{|c|c|} \hline 
Abbreviation & Definition \\ \hline
$\avariable \oblineqlocal \mynext^i \avariablebis$  & 
$\neg (\avariable \oblieqlocal \mynext^i \avariablebis) \wedge 
\overbrace{\mynext \cdots \mynext}^{i \ {\rm times}} \top$ 
\\ \hline 
$\avariable \oblieqlocal \mynext^{-i} \avariablebis$ & 
$\overbrace{\mynext^{-1} \cdots \mynext^{-1}}^{i \ {\rm times}}(\mynext^i \avariable   \oblieqlocal \avariablebis)$ \\
\hline
$\avariable \oblineqlocal \mynext^{-i} \avariablebis$ & 
$\neg (\avariable \oblieqlocal \mynext^{-i} \avariablebis) \wedge 
\overbrace{\mynext^{-1} \cdots \mynext^{-1}}^{i \ {\rm times}} \top$ 
\\ \hline 
$\oblineq{\avariable}{\avariablebis}{\top}$ &  
$(\avariable \oblineqlocal \mynext \avariablebis) \vee 
       \mynext((\avariablebis \oblieqlocal \mynext \avariablebis) \until 
               (\avariablebis \oblineqlocal \mynext \avariablebis))$ \\ \hline 
 $\poblineq{\avariable}{\avariablebis}{\top}$ & 
$(\avariable \oblineqlocal \mynext^{-1} \avariablebis) \vee 
       \mynext^{-1}(
(\avariablebis \oblieqlocal \mynext^{-1} \avariablebis)
 \since 
 (\avariablebis \oblineqlocal  \mynext^{-1} \avariablebis))$ \\ \hline
\end{tabular}
\end{center}
}
\cut{
First, let us come back to our choices for atomic formulae.
In $\mainlogic$, there is no need for 
atomic formulae stating that
`$\avariable$ is not equal to $\avariablebis$ $i$ steps ahead'
since this can be expressed  with 
$\neg (\avariable \oblieqlocal \mynext^i \avariablebis)$.
Similarly,  `$\avariable$ is not equal to $\avariablebis$ $i$ steps backward'
(which could correspond to  $\avariable \oblieqlocal \mynext^{-i} \avariablebis$)
can be expressed
with $\overbrace{\mynext^{-1} \cdots \mynext^{-1}}^{i \ {\rm times}}(\mynext^i \avariable  
       \oblieqlocal \avariablebis)$.
Moreover, in  $\mainlogic^{\top}$, there is no need for atomic formulae of the
form  $\oblineq{\avariable}{\avariablebis}{\top}$
(`$\avariable$ is not repeated in a future value of $\avariablebis$'), since they can be expressed
by  $(\avariable \oblineqlocal \mynext \avariablebis) \vee 
       \mynext((\avariablebis \oblieqlocal \mynext \avariablebis) \until 
               (\avariablebis \oblineqlocal \mynext \avariablebis))$.
The same holds for  $\poblineq{\avariable}{\avariablebis}{\top}$ in  $\pmainlogic^{\top}$.
}

Models for $\mainlogic$ can be  viewed
as finite data words  in $(\Sigma \times \D)^*$, where $\Sigma$ is a finite alphabet and $\D$ is an 
infinite domain. E.g., equalities between dedicated variables can simulate that a position
is labelled by a letter from $\Sigma$; moreover, we may assume that the data values are encoded with the
variable $\avariable$. Let us express that
whenever there are $i < j$ such that $i$ and $j$ [resp. $i+1$ and $j+1$, $i+2$ and $j+2$]  are 
labelled by $\aletter$ [resp. $\aletter'$, $\aletter''$], 
$\sigma(i+1)(\avariable) \neq \sigma(j+1)(\avariable)$.
This can be stated in $\mainlogic$ by:
$
\always(
\aletter \wedge \mynext(\aletter' \wedge \mynext \aletter'')
\Rightarrow
\mynext \neg 
(
\oblieq{\avariable}{\avariable}{\mynext^{-1}\aletter \wedge \aletter' \wedge \mynext \aletter''}
)
)
$.
This is an example of key constraints, see e.g.~\cite[Definition 2.1]{Niewerth&Schwentick11}
and the paper contains also numerous examples of properties that can be captured by $\mainlogic$. 
\altparagraph{Basics on VASS}
A \defstyle{vector addition system with states} (VASS) is a tuple 
$\avass = \triple{\locations}{\counters}{\transitions}$ where
$\locations$ is a finite set of \defstyle{control states},
$\counters$ is  a finite set of \defstyle{counters} and $\transitions$ is a finite
set of \defstyle{transitions} in $\locations \times \Zed^{\counters} \times \locations$. 
A \defstyle{configuration} of $\avass$ is a pair
$\pair{\alocation}{\vect{v}} \in \locations \times
\Nat^{\counters}$. We write $\pair{\alocation}{\vect{v}} \step{} \pair{\alocation'}{\acountval'}$ if
there is \ifLONG a transition \fi $(\alocation, \vect{u}, \alocation')\in \transitions$ such that
$\vect{v'} = \vect{v} + \vect{u}$. 
The \defstyle{reachability problem} for VASS (written Reach(VASS)) consists in checking whether 
 $\pair{\alocation_0}{\vect{v_0}} \step{*} \pair{\alocation_f}{\vect{v_f}}$, given two configurations
$\pair{\alocation_0}{\vect{v_0}}$ and $\pair{\alocation_f}{\vect{v_f}}$. 
The reachability
problem for VASS is decidable with non-primitive recursive complexity
\ifLONG algorithms~\cite{Mayr84,Kosaraju82,Lambert92,Leroux11}\else  algorithms~\cite{Mayr84,Kosaraju82,Leroux11}\fi, but the best known lower bound is only
\expspace{} \cite{Lipton76,Esparza98}.
The \defstyle{control state reachability problem} is the following \expspace-complete
problem~\cite{Lipton76,Rackoff78}: given a configuration  $\pair{\alocation_0}{\vect{v_0}}$
and a control state $\alocation_f$, check
whether $\pair{\alocation_0}{\vect{v_0}} \step{*} \pair{\alocation_f}{\vect{v}}$
for some $\vect{v} \in \Nat^{\counters}$.
We also use two other kinds of computations: with gains ($\stepinc{}$) or losses ($\stepdec{}$).
We write $\pair{\alocation}{\vect{v}} \stepinc{} \pair{\alocation'}{\vect{v'}}$ [resp.\ $\pair{\alocation}{\vect{v}} \stepdec{} \pair{\alocation'}{\vect{v'}}$] if
there is a transition $(\alocation, \vect{u}, \alocation') \in \transitions$ 
and $\vect{w} \in \Nat^{\counters}$ such that
$\vect{v} \preceq \vect{w}$ and $\vect{w} + \vect{u} \preceq
\vect{v'}$ [resp.\ $\vect{w} \preceq \vect{v}$ and $\vect{v'} \preceq \vect{w} + \vect{u}$].
\fi
 
\makeatletter{}\section{The Power of Past: From 
Reach(VASS)
to
\SAT{$\pmainlogic$}}
\label{section-past}
While~\cite{Demri&DSouza&Gascon12} concentrated on decidability
results, here we begin with a hardness result. When past obligations are allowed
as in $\pmainlogic$, \SAT{$\pmainlogic$}  is equivalent to the very
difficult problem of reachability in VASS (see recent developments in~\cite{Leroux&Schmitz15}). Combined with the result of
the next section where we prove that removing past obligations leads
to a reduction into the control state reachability problem for VASS,
this means that reasoning with past obligations is probably much more
complicated.

\begin{thm} 
\label{theorem-reach-PLRV}
There is a polynomial-space reduction 
from Reach(VASS)
into \SAT{$\pmainlogic$}. 
\end{thm}

The proof of Theorem~\ref{theorem-reach-PLRV} 
is analogous to the proof of~\cite[Theorem 16]{BDMSS06:journal}
except that properties are expressed in $\pmainlogic$ instead of being expressed in $\fo{2}{\sim,<,+1}$.

\makeatletter{}\begin{proof} First, the reachability problem for VASS can be reduced in polynomial space to its restriction such that
the initial and final configurations have all the counters equal to zero and each transition can only increment or 
decrement a unique counter. In the sequel, we consider an instance of this subproblem: 
$\avass = \triple{\locations}{\counters}{\transitions}$ is a VASS, the initial configuration
 is $\pair{\alocation_i}{\vec{0}}$ and the final configuration is $\pair{\alocation_f}{\vec{0}}$. 

Now, we build a formula  $\aformula$ in $\pmainlogic$ such that $\pair{\alocation_f}{\vec{0}}$ is reachable
from $\pair{\alocation_i}{\vec{0}}$ iff $\aformula$ is satisfiable. To do so, we encode runs of $\avass$ by data words
following exactly the proof of~\cite[Theorem 16]{BDMSS06:journal} except that the properties are expressed in $\pmainlogic$
instead of  FO$^2(\sim,<,+1)$. Letters from the finite alphabet are also encoded by equalities
of the form $\avariable_0 = \avariable_i$.

The objective is to encode a word $\arun \in \transitions^*$ that represents an accepting run
from  $\pair{\alocation_i}{\vec{0}}$ to $\pair{\alocation_f}{\vec{0}}$. 
We use the alphabet $\transitions$ of transitions, that we code using a logarithmic number of variables. 
One can simulate $m$ different labels in $\pmainlogic$, by using  $\lceil\log(m)+1\rceil$ variables and 
its equivalence classes. In order to simulate the alphabet $\transitions$, we use the variables 
$\avariable_0, \dotsc, \avariable_N$, with $N = \lceil \log( \card{\transitions} ) \rceil$. 
For any $\atransition \in \transitions$, let $\tup{\atransition} \in \pmainlogic$ be the 
formula that tests for label $\atransition$ at the current position.  
More precisely, for any fixed injective function $\lambda : \transitions \to 
\powerset{\interval{1}{N}}$  we define 
\[
\tup{\atransition} ~~=~~ \bigwedge_{i \in \lambda(\atransition)} \avariable_0 = \avariable_i 
~~\land~~ \bigwedge_{1 \leq i \leq N, i \not\in \lambda(\atransition)} \avariable_0 \neq \avariable_i.
\]
Note that $\tup{\atransition}$ uses exclusively variables $\avariable_0, \dotsc, \avariable_N$, 
that is of size logarithmic in $\card{\transitions}$, and that $\tup{\atransition}$ 
holds at a position for at most one $\atransition \in \transitions$.
 We build a $\pmainlogic$ formula $\aformula$ so that any word from $\transitions^*$ corresponding to a 
model of $\aformula$ is an accepting run for $\avass$ from  $\pair{\alocation_i}{\vec{0}}$ to 
$\pair{\alocation_f}{\vec{0}}$. And conversely, for any accepting run of $\avass$ from  $\pair{\alocation_i}{\vec{0}}$ to 
$\pair{\alocation_f}{\vec{0}}$ there is a model of $\aformula$ corresponding to the run.
The following are standard counter-blind conditions to check.
\begin{enumerate}
  \item  Every position satisfies $\tup{\atransition}$ for some $\atransition \in \transitions$.
  \item  The first position satisfies $\tup{\triple{\alocation_i}{\ainstruction}{\alocation}}$
  for some $\alocation \in \locations$.
  \item  The last position satisfies $\tup{\triple{\alocation}{\ainstruction}{\alocation_f}}$ 
   for some $\alocation \in \locations$.
  \item There are no two consecutive positions $i$ and 
   $i+1$ satisfying $\tup{\triple{\alocation_1}{\ainstruction}{\alocation_2}}$ and 
   $\tup{\triple{\alocation_1'}{\ainstruction'}{\alocation_2'}}$ respectively, with $\alocation_2 \neq \alocation_1'$.
\end{enumerate}\medskip

\noindent In the formula $\aformula$, we use the distinguished 
variable $\avariable$ to relate increments and decrements. Here are the main
properties to satisfy.

\begin{enumerate}
\item For every counter $\acounter \in \counters$, there are no two positions labelled by a transition
      with instruction $\inc{\acounter}$ having the
      same value for $\avariable$:
     $$
     \always(
     \inc{\acounter}
     \Rightarrow
     \neg (\oblieq{\avariable}{\avariable}{\inc{\acounter}})
     )
     $$
     where $\inc{\acounter}$ is a shortcut for 
     $\bigvee_{\triple{\alocation}{\inc{\acounter}}{\alocation'} \in \transitions}  \tup{\atransition}$. 
     A similar constraint can be expressed with $\dec{\acounter}$.
\item For every counter $\acounter \in \counters$, for every position labelled by a transition
      with instruction  $\inc{\acounter}$, there is a future position labelled by a transition
      with instruction  $\dec{\acounter}$ with the same value for $\avariable$: 
      $$
     \always(
     \inc{\acounter}
     \Rightarrow
     \oblieq{\avariable}{\avariable}{\dec{\acounter}}
     )
     $$
     where $\dec{\acounter}$ is a shortcut for 
     $\bigvee_{\triple{\alocation}{\dec{\acounter}}{\alocation'} \in \transitions}  \tup{\atransition}$. 
     This guarantees that the final configuration ends with all counters equal to zero. 
\item Similarly,  for every counter $\acounter \in \counters$, for every position labelled by a transition
      with instruction  $\dec{\acounter}$, there is a past position labelled by a transition
      with instruction  $\inc{\acounter}$ with the same value for $\avariable$: 
      $$
     \always(
     \dec{\acounter}
     \Rightarrow
     \poblieq{\avariable}{\avariable}{\inc{\acounter}}
     ).
     $$
     This guarantees that every decrement follows a corresponding increment, satisfying that counter values
     are never negative.
\end{enumerate}\medskip

\noindent Let $\aformula$ be the conjunction of all the formulas defined above. Since all the properties
considered herein are those used in the proof of~\cite[Theorem 16]{BDMSS06:journal} (but herein they are expressed
in $\pmainlogic$ instead of FO2$(\sim,<,+1)$), it follows that 
 $\aformula$ is satisfiable in $\pmainlogic$ iff there is an accepting run of $\avass$ 
from  $\pair{\alocation_i}{\vec{0}}$ to 
$\pair{\alocation_f}{\vec{0}}$.
\end{proof}

\makeatletter{}\section{Leaving the Past Behind Simplifies Things: From \SAT{$\mainlogic$}  to Control State Reachability}
\label{section-upper-bound}
In this section, we show the reduction from \SAT{$\mainlogic$}
to the control state reachability problem in VASS. We obtain
a \twoexpspace{} upper bound for  \SAT{$\mainlogic$}
as a
consequence. This is done in two steps: 
\begin{enumerate}
\itemsep 0 cm
\item simplifying formulas of
the form $\oblieq{\avariable}{\avariablebis}{\aformula}$ to remove the
test formula $\aformula$ (i.e., a reduction from  \SAT{$\mainlogic$}
into \SAT{$\mainlogic^\top$}); and 
\item  reducing from
\SAT{$\mainlogic^\top$}  into the control state reachability 
\ifLONG problem \else pb. \fi 
in VASS.
\end{enumerate}
\makeatletter{}\subsection{Elimination of Test Formulas}
\label{section-elimination}
We give a 
polynomial-time algorithm such that given  $\varphi \in \mainlogic$, it computes a formula $\varphi' \in 
\mainlogic^{\top}$ that preserves satisfiability: there is a model
$\sigma$  such that $\sigma \models \varphi$ iff 
there is a model $\sigma'$ such that $\sigma' \models \varphi'$.
We give the reduction in two steps. First, we eliminate formulas
with inequality tests of the form
$\oblineq{\avariable}{\avariablebis}{\aformulabis}$ using only positive
tests of the form $\oblieq{\avariable}{\avariablebis}{\aformulabis}$. We
then eliminate formulas of the form
$\oblieq{\avariable}{\avariablebis}{\aformulabis}$, using only formulas 
of the form $\oblieq{\avariable}{\avariablebis}{\top}$. 
Although both
reductions share some common structure, they use independent coding
strategies, and exploit different features of the logic; we therefore
present them separately.

We first show how to eliminate all formulas with inequality tests of the form $\oblineq{\avariable}{\avariablebis}{\aformulabis}$.

Let $\mainlogiceq$  be the logic $\mainlogic$ where there are no appearances of formulas of the form $\oblineq{\avariable}{\avariablebis}{\aformulabis}$; 
and let $\pmainlogiceq$ be $\pmainlogic$ without $\oblineq{\avariable}{\avariablebis}{\aformulabis}$ or $\poblineq{\avariable}{\avariablebis}{\aformulabis}$.

\makeatletter{}

Henceforward, $\variabs(\varphi)$ denotes the set of all variables in $\varphi$, and $\subft{\varphi}$ the set of all subformulas $\psi$ such that $\oblieq{\avariable}{\avariablebis}{\psi}$ or $\oblineq{\avariable}{\avariablebis}{\psi}$ appears in $\varphi$ for some $\avariable,\avariablebis \in \variabs(\varphi)$.

In both reductions we make use of the following easy lemma.
\begin{lem}\label{lem:no-nesting-reduction}
There is a polynomial-time satisfiability-preserving translation $t : \mainlogic \to \mainlogic$ [resp.\ $t: \mainlogic^\approx \to \mainlogic^\approx$] such that for every $\varphi$ and $\psi\in\subft{t(\varphi)}$, we have 
 $\subft{\psi} = \emptyset$.
\end{lem}
\begin{proof}
This is a standard reduction.
Indeed, given a formula $\varphi$ with subformula $\oblieq{\avariable}{\avariablebis}{\psi}$,
$\varphi$ is satisfiable if{f} $\always(\psi \Leftrightarrow \avariable_{new} \oblieqlocal \avariablebis_{new})
\wedge \varphi[\oblieq{\avariable}{\avariablebis}{\psi} \leftarrow 
\oblieq{\avariable}{\avariablebis}{\avariable_{new} \oblieqlocal \avariablebis_{new}}]$
is satisfiable, where $\avariable_{new}, \avariablebis_{new} \not\in \variabs(\varphi)$,
and $\varphi[\oblieq{\avariable}{\avariablebis}{\psi} \leftarrow 
\oblieq{\avariable}{\avariablebis}{ \avariable_{new} \oblieqlocal \avariablebis_{new}}]$ is the result of replacing every occurrence of  $\oblieq{\avariable}{\avariablebis}{\psi}$
by $\oblieq{\avariable}{\avariablebis}{ \avariable_{new} \oblieqlocal \avariablebis_{new}}$ in $\varphi$.
Similarly, given a formula $\varphi$ with subformula $\oblineq{\avariable}{\avariablebis}{\psi}$,
$\varphi$ is satisfiable if{f} $\always(\psi \Leftrightarrow \avariable_{new} \oblieqlocal \avariablebis_{new})
\wedge \varphi[\oblineq{\avariable}{\avariablebis}{\psi} \leftarrow 
\oblineq{\avariable}{\avariablebis}{\avariable_{new} \oblieqlocal \avariablebis_{new}}]$
is satisfiable.
We need to apply these replacements repeatedly, at most a polynomial number of times if we apply
it to the innermost occurrences. 
\end{proof}

  \begin{prop}[from $\mainlogic$ to $\mainlogiceq$]\label{prop:LRV2LRVapprox}
    There is a polynomial-time reduction from \SAT{$\mainlogic$} into \SAT{$\mainlogiceq$}; 
    and from \SAT{$\pmainlogic$}  into \SAT{$\pmainlogiceq$}.
  \end{prop}
\makeatletter{}\ifLONG \begin{proof}

  For every $\varphi \in \mainlogic$, we compute $\varphi' \in
  \mainlogiceq$ in polynomial time, which preserves satisfiability.
The variables of $\varphi'$ consist of all the variables of $\varphi$,  plus: a distinguished
  variable $\mathtt{k}$, and variables
  $\mathtt{v}^\approx_{\avariablebis,\psi}$,
$\mathtt{v}_{\oblineq{\avariable}{\avariablebis}{\psi}}$ for every
  subformula $\psi$ of $\varphi$ and variables $\avariable, \avariablebis$ of $\varphi$.
The variables
$\mathtt{v}_{\oblineq{\avariable}{\avariablebis}{\psi}}$'s
will be used to get rid of $\not\approx$ in formulas of the form
$\oblineq{\avariable}{\avariablebis}{\psi}$, and the variables
 $\mathtt{v}^\approx_{\avariablebis,\psi}$'s to treat formulas of the form
$\lnot(\oblineq{\avariable}{\avariablebis}{\psi})$.
Finally, $\mathtt k$ is a special variable, which has a constant value, different from all the values of the variables of $\varphi$.

Assume $\varphi$ is in negation normal form.
Note that each positive occurrence of $\oblineq{\avariable}{\avariablebis}{\psi}$ can be safely replaced with $\avariable \not\approx  \mathtt{v}_{\oblineq{\avariable}{\avariablebis}{\psi}} \land  \oblieq{\mathtt{v}_{\oblineq{\avariable}{\avariablebis}{\psi}}}{\avariablebis}{\psi}$. Indeed, the latter formula implies the former, and it is not difficult to see that whenever there is a model for the former formula, there is also one for the latter.
    On the other hand, translating formulas of the form $\lnot (\oblineq{\avariable}{\avariablebis}{\psi})$
is more involved as these implicate some form of universal quantification. For treating these formulas, we use the variables $\mathtt{v}^\approx_{\avariablebis,\psi}$  and $\mathtt{k}$ as explained next.
Let $i$ be the first position of the model so that all future positions $j>i$ verifying $\psi$ have the same value on variable $\avariablebis$, say  value $n$.
As we will see, with a formula of $\mainlogiceq$, one can ensure that
$\mathtt{v}^\approx_{\avariablebis,\psi}$ has the same value as
$\mathtt k$ for all $j\leq i$ and value $n$ for all other positions.
The enforced values are illustrated below, with an initial prefix where variable $\mathtt{v}^\approx_{\avariablebis,\psi}$ is equal to $\mathtt k$ until we reach position $i$, from which point all values of $\mathtt{v}^\approx_{\avariablebis,\psi}$ concide with value $n$ ---represented as a dashed area. 
  \begin{center}
    \includegraphics{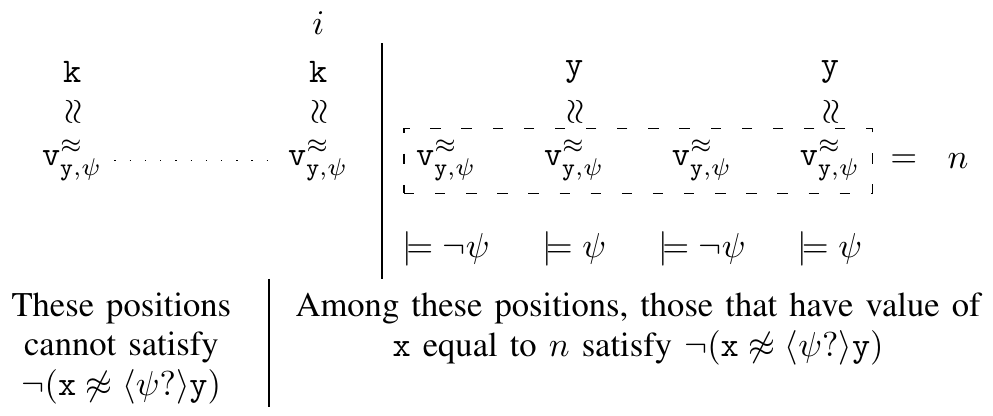}
  \end{center}
The  positions satisfying $\lnot (
  \oblineq{\avariable}{\avariablebis}{\psi})$ are of two types: the
  first type are those positions such that no future position
  satisfies $\psi$. The second type are those such that all future
  positions satisfying $\psi$ have the same value $n$ on variable
  $\avariablebis$, and the variable $\avariable$ takes the value $n$. The first type of
  positions are captured by the formula $\lnot \mynext\sometimes
  \psi$. As can be seen from the illustration above, the second type
  of positions can be captured by the formula $\avariable \approx
  \mynext \mathtt{v}^\approx_{\avariablebis,\psi}$. Thus, $\lnot (
  \oblineq{\avariable}{\avariablebis}{\psi})$ can be replaced with $\lnot \mynext\sometimes \psi ~\lor~ \avariable \approx \mynext \mathtt{v}^\approx_{\avariablebis,\psi}$.
                    Past
  obligations are treated in a symmetrical way.  We now formalise
  these ideas, showing that $ \sigma  \models \varphi$ implies that $ \sigma_\varphi \models \varphi' $, where $\varphi'$ is the translation of $\varphi$ and 
$\sigma_\varphi$ an extension of $\sigma$ with the new variables of the translation and values corresponding to the intuition above. 
On the other hand, we will also show that $\sigma  \models \varphi'$ implies $ \sigma  \models \varphi$. Next, we formally define $\sigma_\varphi$ and the translation $\varphi'$, and then we show these two facts.

\cut{
We show how to compute in polynomial time, for every $\varphi \in \mainlogic$,  a formula $\varphi' \in \mainlogiceq$ that preserves satisfiability. The formula $\varphi'$ uses, besides all the variables from $\varphi$, a distinguished variable $\mathtt{k} \not\in \variabs(\varphi)$, and variables $\mathtt{v}^\approx_{\avariable,\psi}$, $\mathtt{v}_{\oblineq{\avariable}{\avariablebis}{\psi}}$ for every subformula $\psi$ of $\varphi$ and variables $\avariable, \avariablebis \in\variabs(\varphi)$.

The idea behind the coding is the following. Any positive test $\oblineq{\avariable}{\avariablebis}{\psi}$ can be always converted to a test $\avariable \not\approx \mathtt{v}_{\oblineq{\avariable}{\avariablebis}{\psi}} \land 
\oblieq{\mathtt{v}_{\oblineq{\avariable}{\avariablebis}{\psi}}}{\avariablebis}{\psi}$, where 
$\mathtt{v}_{\oblineq{\avariable}{\avariablebis}{\psi}}$ is a fresh variable. 
At any given position, we only need at most as many fresh variables as there are subformulas $\psi$ and source and target variables $\avariable,\avariablebis \in \variabs(\varphi)$, and we hence index these fresh variables by $\avariable$, $\avariablebis$ and $\psi$. 
On the other hand, when a negative test $\lnot (\oblineq{\avariable}{\avariablebis}{\psi})$ holds at a position $i$, it means that all positions satisfying $\psi$ to the right of $i$ have the same data value in the variable $\avariablebis$. Since this data value is determined for every $\psi,\avariablebis$, there is only boundedly many data values to remember, one for every subformula and variable. We use separate variables $\mathtt{v}^\approx_{\avariable,\psi}$ to remember this data value for every $\avariable, \psi$. This variable is used in conjunction with $\mathtt{k}$. Whenever $\mathtt{v}^\approx_{\avariable,\psi} \not\approx \mathtt{k}$ at a position $i$, it means that $\lnot ( \oblineq{\mathtt{v}^\approx_{\avariable,\psi}}{\avariable}{\psi})$, and we can make sure of this without using $\not\approx$. Indeed, we can do this by letting $\mathtt{v}^\approx_{\avariable,\psi}$ maintain the same value from position $i$ onwards ---until the last element verifying $\psi$---, we can test that whenever $\psi$ holds, we have $\avariable \approx \mathtt{v}^\approx_{\avariable,\psi}$.
}
\medskip

Given a model $\sigma$, let us define the model $\sigma_\varphi$ as follows:
\begin{enumerate}[label=\({\alph*}]
  \item\label{it:sigmaphi:red1:1}
 $|\sigma| = |\sigma_\varphi|$;
  \item\label{it:sigmaphi:red1:2}
 for every $0 \leq i < |\sigma|$ and $\avariable \in \variabs(\varphi)$, $\sigma(i)(\avariable) = \sigma_\varphi(i)(\avariable)$;
  \item\label{it:sigmaphi:red1:3}
 there is some data value $\adatum \not\in \set{\sigma(i)(\avariable) \mid \avariable \in \variabs(\varphi), 0 \leq i < |\sigma|}$ such that for every $0 \leq i < |\sigma_\varphi|$, $\sigma_\varphi(i)(\mathtt{k}) = \adatum$;
\item\label{it:sigmaphi:red1:4}
 for every $0 \leq i < |\sigma|$, $\avariable \in \variabs(\varphi)$, and $\psi \in \subft{\varphi}$,
  \begin{itemize}
  \item if for some $j \geq i$, $\sigma,j \models \psi$ and for every
    $i \leq j'<|\sigma|$ such that $\sigma,
    j' \models \psi$  we have $\sigma(j)(\avariable) = \sigma(j')(\avariable)$, then
    $\sigma_\varphi(i)(\mathtt{v}^\approx_{\avariable,\psi}) =
\sigma(j)(\avariable)$,
  \item otherwise,
    $\sigma_\varphi(i)(\mathtt{v}^\approx_{\avariable,\psi}) =
    \sigma_\varphi(i)(\mathtt{k})$ (i.e., equal to $\adatum$);
  \end{itemize}
\item  \label{it:sigmaphi:red1:5}
for every $0 \leq i < |\sigma|$, $\avariable, \avariablebis \in \variabs(\varphi)$, and $\psi \in \subft{\varphi}$,
  \begin{itemize}
  \item if for some $j>i$,
    $\sigma(j)(\avariablebis) \neq \sigma(i)(\avariable)$ and $\sigma,j
    \models \psi$, then let $j_0$ be the first such $j$, and
$\sigma_\varphi(i)(\mathtt{v}_{\oblineq{\avariable}{\avariablebis}{\psi}}) =
    \sigma(j_0)(\avariablebis)$,
  \item otherwise,
    $\sigma_\varphi(i)(\mathtt{v}_{\oblineq{\avariable}{\avariablebis}{\psi}}) =
    \sigma_\varphi(i)(\mathtt{k})$ (i.e., equal to $\adatum$).
  \end{itemize}
\end{enumerate}
It is evident that for every $\varphi$ and $\sigma$, a model $\sigma_\varphi$ with the aforementioned properties exists. 
Next, we define $\varphi'\in \mainlogiceq$ so that 
\begin{align}
  \sigma \models \varphi \text{ if, and only if, } \sigma_\varphi \models
  \varphi'.\label{eq:sigma,phi-iff-sigmaphi,phi'}
\end{align}

\noindent We assume that for every $\psi\in\subft{\varphi}$, we have 
that $\subft{\psi} = \emptyset$; this is without any loss of generality due to Lemma~\ref{lem:no-nesting-reduction}. 
We will also assume that $\varphi$ is in negation normal form, that is, negation appears only in subformulas of the type $\lnot(\avariable \approx \mynext^\ell \avariablebis)$ or $\lnot 
( \oblieq{\avariable}{\avariablebis}{\psi})$ [resp.\ $\neg (\oblineq{\avariable}{\avariablebis}{\psi})$]. 
In particular, this means
that we introduce a dual operator for each temporal operator. 

\begin{itemize}
\item First, the variable $\mathtt{k}$ will act as a
  constant; it will always have the same data value at any position of
  the model, which must be different from those of all variables of
  $\varphi$.
  \[
  \textit{const} ~~=~~ \always \left((\mathtt{k}
    \approx \mynext \mathtt{k} \lor \lnot \mynext
    \top) ~\land~ \bigwedge_{\avariable \in \variabs(\varphi)}
    (\mathtt{k} \not\approx \avariable)\right).
  \]
\item   Second, for every position we ensure that if
  $\mathtt{v}^\approx_{\avariable,\psi}$ is different from
  $\mathtt{k}$, then it preserves its value until the
  last element verifying $\psi$; and if $\psi$ holds at any of these positions then $\mathtt{v}^\approx_{\avariable,\psi}$ contains the value of $\avariable$.
  \begin{align*}
    \textit{val-}\mathtt{v}^\approx_{\avariable,\psi} ~~=~~ \always 
    \left(\begin{tabular}{rcl}\normalsize
      $\mathtt{k} \not \approx
      \mathtt{v}^\approx_{\avariable,\psi}$&$\Rightarrow$&$ \mathtt{v}^\approx_{\avariable,\psi}
      \approx \mynext \mathtt{v}^\approx_{\avariable,\psi} \lor \lnot \mynext\sometimes \psi
      ~~~~\land$\\
      $\mathtt{k} \not \approx
      \mathtt{v}^\approx_{\avariable,\psi} \land \psi $&$\Rightarrow$&$
      \mathtt{v}^\approx_{\avariable,\psi} \approx \avariable$
    \end{tabular}\right).
  \end{align*}

\item Finally, let $\varphi^\approx$ be the result of replacing
\begin{enumerate}
\item[(f)] every appearance of   $ \lnot (\oblineq{\avariable}{\avariablebis}{\psi})$ by
$\lnot \mynext\sometimes \psi ~\lor~ \avariable \approx \mynext \mathtt{v}^\approx_{\avariablebis,\psi}$, and
\item[(g)] every positive appearance of   $ \oblineq{\avariable}{\avariablebis}{\psi}$ by
  $\avariable \not \approx \mathtt{v}_{\oblineq{\avariable}{\avariablebis}{\psi}} \land \oblieq{\mathtt{v}_{\oblineq{\avariable}{\avariablebis}{\psi}}}{\avariablebis}{\psi}$
\end{enumerate}
in $\varphi$.
\end{itemize}
The formula $\varphi'$ is defined as follows:
\[
\varphi' ~~=~~ 
\varphi^\approx
~~\land~~
\textit{const}
~~\land~~
\bigwedge_{\substack{\psi  \in \subft{\aformula},\\ \avariable \in \variabs(\varphi)}}
\textit{val-}\mathtt{v}^\approx_{\avariable,\aformulabis}.
\]
Notice that $\varphi'$ can be computed from $\varphi$ in polynomial time in the size of $\varphi$.

\smallskip

\begin{clm}
  $\varphi$ is satisfiable if, and only if, $\varphi'$ is satisfiable.
\end{clm}
\begin{proof}

[$\Rightarrow$]
Note that one direction would follow 
from~\eqref{eq:sigma,phi-iff-sigmaphi,phi'}: 
if $\varphi$ is satisfiable by some $\sigma$, then $\varphi'$ is satisfiable in  $\sigma_\varphi$.
In order to establish~\eqref{eq:sigma,phi-iff-sigmaphi,phi'}, we show that 
\begin{align}
  \text{ for every subformula $\gamma$ of $\varphi$ and $0 \leq i < |\sigma_\varphi|$, we have
    $\sigma, i \models \gamma$ if{f} $\sigma_{\varphi}, i
    \models \gamma^\approx$.}
  \label{eq:psi-equiv-psiapprox}
\end{align}
We further assume that $\gamma$ is not an atomic formula that is dominated by a negation in $\varphi$. 

Since $\sigma_\varphi \models \textit{const}$ by 
Condition~\eqref{it:sigmaphi:3}, and $\sigma_\varphi \models
\textit{val-}\mathtt{v}^\approx_{\avariable,\psi}$ for every $\psi$ due to
\eqref{it:sigmaphi:4}, this is sufficient to conclude that $\sigma
\models \varphi$ if{f} $\sigma_\varphi \models \varphi'$.

We show \eqref{eq:psi-equiv-psiapprox} by structural induction.  If
$\gamma = \oblineq{\avariable}{\avariablebis}{\psi}$, then $\sigma, i
\models \gamma$ if there is some $j > i$ where
$\sigma(j)(\avariablebis) \neq \sigma(i)(\avariable)$ and $\sigma, j
\models \psi$. Let $j_0$ be the first such $j$. Then, by 
Condition~\eqref{it:sigmaphi:red1:5}, 
$\sigma(i)(\avariable)
\neq\sigma_\varphi(i)( \mathtt{v}_{\oblineq{\avariable}{\avariablebis}{\psi}}) =
\sigma(j_0)(\avariablebis)$. 
Hence, 
$\sigma_\varphi,i \models
\oblieq{\mathtt{v}_{\oblineq{\avariable}{\avariablebis}{\psi}}}{\avariablebis}{\psi}$
and 
$\sigma_\varphi,i \models \avariable \not\approx \mathtt{v}_{\oblineq{\avariable}{\avariablebis}{\psi}}$ 
and thus
$\sigma_\varphi,i \models \gamma^\approx$.

If, on the other hand, $\gamma = \lnot(\oblineq{\avariable}{\avariablebis}{\psi})$ then either
\begin{itemize}
\item there is no $j>i$ so that $\sigma,j \models \psi$, or,
  otherwise,
\item for every $j' \geq i+1$ so that $\sigma,j' \models \psi$ we have
  $\sigma(j')(\avariablebis) = \sigma(i)(\avariable)$.
\end{itemize}
In the first case, we have that $\sigma_\varphi,i \models \lnot
\mynext \sometimes \psi$, and in the second case we have that, by 
Condition~\eqref{it:sigmaphi:red1:4},
$\sigma_\varphi(i')(\mathtt{v}^\approx_{\avariablebis,\psi}) =
\sigma_\varphi(i)(\avariable)$ for every $i' > i$, and in particular
for $i' = i+1$. Hence, $\sigma_\varphi, i \models \lnot \mynext \sometimes
\psi ~\lor~ \avariable \approx \mynext \mathtt{v}^\approx_{\avariablebis,\psi}$
and thus $\sigma_\varphi, i \models \gamma^\approx$.

Finally, the proof for the base case of the form $\gamma =
\avariable \oblieqlocal \mynext^\ell \avariablebis$ [resp.\ $\avariable \noblieqlocal \mynext^\ell \avariablebis$] 
and for all Boolean and temporal operators are by an easy 
verification, since $(\cdot)^\approx$ is homomorphic for these. Hence,
\eqref{eq:psi-equiv-psiapprox} holds.

\smallskip

[$\Leftarrow$] Now suppose that $\sigma \models \varphi'$. Since
$\sigma \models \textit{const}$, we have that $\sigma$ verifies
Condition~\eqref{it:sigmaphi:red1:3}. And since $\sigma \models
\textit{val-}\mathtt{v}^\approx_{\avariable,\psi}$ for every $\psi \in \subft{\varphi}$,
$\avariable \in \variabs(\varphi)$, we have that $\sigma$ verifies
Condition~\eqref{it:sigmaphi:red1:4}.
We prove by structural induction that for every subformula $\gamma$ of
$\varphi$, if $\sigma, i \models \gamma^\approx$ then $\sigma, i
\models \gamma$ ($\gamma$ is not an atomic formula that is dominated by a negation in $\varphi$).

\begin{itemize}

\item If $\gamma = \avariable \approx \avariablebis$ [resp.\
  $\avariable \not\approx \avariablebis$, $\avariable \approx
  \mynext^\ell \avariablebis$, $\avariable \not\approx \mynext^\ell
  \avariablebis$] it is immediate since $\gamma^\approx = \gamma$.

\item If $\gamma = \oblineq{\avariable}{\avariablebis}{\psi}$ and thus
  $\gamma^\approx = \avariable \not \approx
  \mathtt{v}_{\oblineq{\avariable}{\avariablebis}{\psi}} \land
  \oblieq{\avariable_{\avariablebis,\psi}^{\not
      \approx}}{\avariablebis}{\psi}$.  Then,
  $\sigma(i)(\avariable)\neq
  \sigma(i)(\mathtt{v}_{\oblineq{\avariable}{\avariablebis}{\psi}})=\sigma(j)(\avariablebis)$ for some $j>i$ where $\sigma,j
  \models \psi$. Hence, $\sigma, i \models
  \oblineq{\avariable}{\avariablebis}{\psi}$.

\item If $\gamma = \lnot (\oblineq{\avariable}{\avariablebis}{\psi})$,
  and thus $\gamma^\approx = \lnot \mynext\sometimes \psi ~\lor~ \avariable
  \approx \mynext \mathtt{v}^\approx_{\avariablebis,\psi}$. This means that
  either
  \begin{itemize}
  \itemsep 0 cm 
  \item there is no $j>i$ so that $\sigma, j \models \psi$ and hence
    $\sigma, i \models \lnot
    (\oblineq{\avariable}{\avariablebis}{\psi})$, or
  \item $\sigma, j \models \psi$ for some $j>i$ and
    $\sigma(i)(\avariable)=\sigma(i+1)(\mathtt{v}^\approx_{\avariablebis,\psi})$. By
    Condition~\eqref{it:sigmaphi:red1:4}, we have that for all $i <
    j'$, so that $\sigma, j' \models \psi$, we have
    $\sigma(j')(\mathtt{v}^\approx_{\avariablebis,\psi}) =
    \sigma(j')(\avariablebis)$. Then, $\sigma, i \models \lnot
    (\oblineq{\avariable}{\avariablebis}{\psi})$.
  \end{itemize}
\item If $\gamma = \sometimes \psi$ and $\gamma^\approx = \sometimes
  \psi^\approx$, there must be some position $i'\geq i$ so that $\sigma,
  i' \models \psi^\approx$. By inductive hypothesis, $\sigma, i'
  \models \psi$ and hence $\sigma, i \models \sometimes \psi$. We
  proceed similarly for all temporal operators and their dual, as well
  as for the Boolean operators $\land$, $\lor$. This is because
  $(\cdot)^\approx$ is homomorphic for all temporal and positive
  Boolean operators.\qedhere
\end{itemize}
\end{proof}\medskip

\noindent We can easily extend this coding allowing for past obligations. We only need to use some extra variables $\mathtt{v}^{-1,\approx}_{\avariable,\psi}$, $\mathtt{v}^{-1}_{\poblineq{\avariable}{\avariablebis}{\psi}}$ that behave in the same way as the previously defined, but with past obligations. That is, we also define a $\textit{val-}\mathtt{v}^{-1,\approx}_{\avariable,\psi}$ as $\textit{val-}\mathtt{v}^\approx_{\avariable,\psi}$, but making use of $\mathtt{v}^{-1,\approx}_{\avariable,\psi}$, $\mynext^{-1}$ and $\sometimes^{-1}$.
And finally, we have to further replace 
\begin{enumerate}
\itemsep 0 cm 
\item[(c)] every appearance of   $ \lnot (\poblineq{\avariable}{\avariablebis}{\psi})$ by
$\lnot \mynext^{-1}\sometimes^{-1} \psi ~\lor~ \avariable \approx \mynext \mathtt{v}^{-1,\approx}_{\avariablebis,\psi}$, and
\item[(d)] every positive appearance of   $\poblineq{\avariable}{\avariablebis}{\psi}$ by
  $\avariable \not \approx \mathtt{v}_{\poblineq{\avariable}{\avariablebis}{\psi}}^{-1} \land \poblieq{\mathtt{v}_{\poblineq{\avariable}{\avariablebis}{\psi}}^{-1}}{\avariablebis}{\psi}$
\end{enumerate}
in $\varphi$ to obtain $\varphi^\approx$.
\end{proof}

\else \begin{proof}[Proof sketch]
  For every $\varphi \in \mainlogic$, we compute $\varphi' \in
  \mainlogiceq$ in polynomial time, which preserves satisfiability.
  Besides all the variables from $\varphi$, $\varphi'$ uses a distinguished
  variable $\mathtt{k}$ (which will have a constant value, different
  from all the values of the variables of $\varphi$) and variables
  $\mathtt{v}^\approx_{\avariablebis,\psi}$,
  $\mathtt{v}_{\oblineq{\avariable}{\avariablebis}{\psi}}$ for every
  subformula $\psi$ of $\varphi$ and variables $\avariable,
  \avariablebis$ of $\varphi$.

  Each subformula $\oblineq{\avariable}{\avariablebis}{\psi}$ is
  replaced by $\avariable \not\approx
  \mathtt{v}_{\oblineq{\avariable}{\avariablebis}{\psi}} \land \oblieq{\mathtt{v}_{\oblineq{\avariable}{\avariablebis}{\psi}}}{\avariablebis}{\psi}$,
  where $\mathtt{v}_{\oblineq{\avariable}{\avariablebis}{\psi}}$ is a fresh
  variable.
On the other hand, for each subformula $\lnot
  (\oblineq{\avariable}{\avariablebis}{\psi})$, we use
  the variable $\mathtt{v}^\approx_{\avariablebis,\psi}$  in  conjunction with $\mathtt{k}$ as shown below.
  \begin{center}
    \includegraphics{EliminateNegTest}
  \end{center}
 In the beginning,
  $\mathtt{v}^\approx_{\avariablebis,\psi} \oblieqlocal \mathtt{k}$,
  which is broken at the first position where all future positions
  where $\psi$ holds have the same value for $\avariablebis$. At this
  position (say $i+1$), $\mathtt{v}^\approx_{\avariablebis,\psi}
  \oblineqlocal \mathtt{k}$ and 
  $\mathtt{v}^\approx_{\avariablebis,\psi}$ maintains the same value
  (say $n$)
  at all future positions, illustrated as a dashed box above.
  In all positions after $i$ that satisfy $\psi$, $\avariablebis
  \oblieqlocal \mathtt{v}^\approx_{\avariablebis,\psi}$. These
  conditions can be enforced without using $\lnot (
  \oblineq{\avariable}{\avariablebis}{\psi})$. Suppose
  $\avariable$ also has value $n$ at position $i$ as shown above. Now
  in any position after $i$ that satisfies $\psi$, $\avariablebis$ has
  the value $n$, which is the value of $\avariable$ in position
  $i$. This is exactly the semantics of $\lnot (
  \oblineq{\avariable}{\avariablebis}{\psi})$. Hence, $\lnot (
  \oblineq{\avariable}{\avariablebis}{\psi})$ can be replaced by $\lnot \mynext\sometimes \psi ~\lor~ \avariable \approx \mynext \mathtt{v}^\approx_{\avariablebis,\psi}$. Past
  obligations are treated in a symmetrical way.  
\end{proof}
\fi

  \begin{prop}[from $\mainlogiceq$ to $\mainlogic^\top$]
\label{prop:LRVapprox-to-LRVtop}
    There is a polynomial-time reduction from 
    \SAT{$\mainlogiceq$}  to \SAT{$\mainlogic^{\top}$}; and from \SAT{$\pmainlogiceq$}  into 
\SAT{$\pmainlogic^{\top}$}.
  \end{prop}
\makeatletter{}\ifLONG
  \begin{proof}

In a nutshell, for every $\varphi \in
  \mainlogiceq$, we compute in polynomial time a formula $\varphi' \in \mainlogic^{\top}$ that
  preserves satisfiability. Besides all
  the variables from $\varphi$, $\varphi'$ uses a new distinguished
  variable $\mathtt{k}$, and a variable
  $\mathtt{v}_{\avariablebis,\psi}$ for every subformula $\psi$ of
  $\varphi$ and every variable $\avariablebis$ of $\varphi$.
  We enforce $\mathtt{k}$ to have a constant value different from all
  values of variables of $\varphi$. At every position, we 
  enforce
  $\psi$ to hold if $\mathtt{v}_{\avariablebis,\psi} \approx \avariablebis$,
  and $\psi$ not to  hold if $\mathtt{v}_{\avariablebis,\psi} \approx
  \mathtt{k}$ as shown above. 
  Then
  $\oblieq{\avariable}{\avariablebis}{\psi}$ is replaced by
  $\oblieq{\avariable}{\mathtt{v}_{\avariablebis,\psi}}{\top}$. 
  \begin{center}
    \includegraphics{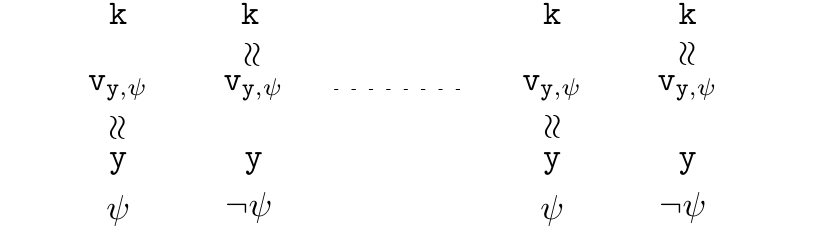}
  \end{center}
  Next,
  we formalise these ideas.

\cut{
Let us be more precise. 
We show how to compute in polynomial time, for every $\varphi \in \mainlogiceq$,  a formula $\varphi' \in \mainlogic^{\top}$ that preserves satisfiability. The formula $\varphi'$ uses, besides all the variables from $\varphi$, a distinguished variable $\mathtt{k}$ that does not appear in $\varphi$, and a variable $\mathtt{v}_{\avariable,\psi}$ for every subformula $\psi$ of $\varphi$ and every variable $\avariable$ appearing in $\varphi$. 
}

\smallskip

Let $\varphi \in \mainlogiceq$. For any model $\sigma$, let $\sigma_\varphi$ be so that
\begin{enumerate}[label=\({\alph*}]
  \item\label{it:sigmaphi:1}
 $|\sigma| = |\sigma_\varphi|$,
  \item\label{it:sigmaphi:2}
 for every $0 \leq i < |\sigma|$ and $\avariable \in \variabs(\varphi)$, $\sigma(i)(\avariable) = \sigma_\varphi(i)(\avariable)$,
  \item\label{it:sigmaphi:3}
 there is some data value $\adatum \not\in \set{\sigma_\varphi(i)(\avariable) \mid \avariable \in \variabs(\varphi), 0 \leq i < |\sigma|}$ such that for every $0 \leq i < |\sigma|$, $\sigma_\varphi(i)(\mathtt{k}) = \adatum$, and
  \item\label{it:sigmaphi:4}
 for every $0 \leq i < |\sigma|$, $\sigma_\varphi(i)(\mathtt{v}_{\avariable,\psi}) = \sigma_\varphi(i)(\avariable)$ if $\sigma,i \models \psi$, and $\sigma_\varphi(i)(\mathtt{v}_{\avariable,\psi}) = \sigma_\varphi(i)(\mathtt{k})$ otherwise.
\end{enumerate}
For every other unmentioned variable, $\sigma$ and $\sigma_\varphi$ coincide. It is evident that for every $\varphi$ and $\sigma$, a model $\sigma_\varphi$ with the aforementioned properties exists. 
Next, we define $\varphi'\in \mainlogic^{\top}$ so that 
\begin{align}
  \sigma \models \varphi ~~\text{if, and only if,}~~ \sigma_\varphi \models \varphi'. \label{eq:lrvapprox-to-lrvtop:1}
\end{align}\smallskip

\noindent We assume that for every $\psi\in\subft{\varphi}$, we have 
that $\subft{\psi} = \emptyset$. This is without any loss of generality by Lemma~\ref{lem:no-nesting-reduction}.

\begin{itemize}
\item First, the variable $\mathtt{k}$ will act as a
  constant; it will always have the same data value at any position of
  the model, which must be different from those of all variables of
  $\varphi$.
  \[
  \textit{const} ~~=~~ \always \left((\mathtt{k}
    \approx \mynext \mathtt{k} \lor \lnot \mynext
    \top) ~\land~ \bigwedge_{\avariable \in \variabs(\varphi)}
    (\mathtt{k} \not\approx \avariable)\right).
  \]
\item   Second, any variable $\mathtt{v}_{\avariable,\psi}$ has either the value of
  $\mathtt{k}$ or that of $\avariable$. Further, the
  latter holds if, and only if, $\psi$ is true.
  \[
  \textit{val-}\mathtt{v}_{\avariable,\psi} ~~=~~ \always( \mathtt{v}_{\avariable,\psi} \approx
  \mathtt{k} ~\lor~ \mathtt{v}_{\avariable,\psi} \approx
  \avariable) ~\land~ \always( {\mathtt{v}_{\avariable,\psi} \approx \avariable}
  ~\Leftrightarrow~ \psi).
  \]
\item   Finally, let $\varphi^{\top}$ be the result of replacing every
  appearance of $\oblieq{\avariable}{\avariablebis}{\psi}$ by
  $\oblieq{\avariable}{\mathtt{v}_{\avariablebis,\psi}}{\top}$ in $\varphi$.
\end{itemize}
We then define $\varphi'$ as follows.
\[
\varphi' ~~=~~ 
\varphi^\top
~~\land~~
\textit{const}
~~\land~~
\bigwedge_{\psi \in \subft{\varphi}}\textit{val-}\mathtt{v}_{\avariable,\psi}.
\]
Notice that $\varphi'$ can be computed from $\varphi$ in polynomial time.

\smallskip

\begin{clm}
$\varphi$ is satisfiable if, and only if,
  $\varphi'$ is satisfiable.
\end{clm}

\begin{proof}

[$\Rightarrow$]
Note that one direction would follow from
\eqref{eq:lrvapprox-to-lrvtop:1}: if $\varphi$ is satisfiable by
some $\sigma$, then $\varphi'$ is satisfiable in $\sigma_\varphi$. In
order to establish \eqref{eq:lrvapprox-to-lrvtop:1}, we show that 
\begin{align}
  \text{ for every subformula $\gamma$ of $\varphi$ and $0 \leq i < |\sigma_\varphi|$, we have
    $\sigma, i \models \gamma$ if{f} $\sigma_{\varphi}, i
    \models \gamma^{\top}$.}
  \label{eq:psi-equiv-psitop}
\end{align}
Since $\sigma_\varphi \models \textit{const}$ by
condition \eqref{it:sigmaphi:3} and $\sigma_\varphi \models
\textit{val-}\mathtt{v}_{\avariable,\psi}$ for every $\psi$ due to
\eqref{it:sigmaphi:4}, this is sufficient to conclude then $\sigma
\models \varphi$ iff $\sigma_\varphi \models \varphi'$.
We show it by structural induction.  If $\sigma, i \models
\oblieq{\avariable}{\avariablebis}{\psi}$ then there is some $j > i$
where $\sigma(j)(\avariablebis) = \sigma(i)(\avariable)$ and $\sigma,
j \models \psi$. Then, by condition \eqref{it:sigmaphi:4},
$\sigma_\varphi(j)(\mathtt{v}_{\avariablebis,\psi}) =
\sigma_\varphi(j)(\avariablebis)$ and thus $\sigma_\varphi, i \models
\oblieq{\avariable}{\mathtt{v}_{\avariablebis,\psi}}{\top}$.  Conversely, if
$\sigma_{\varphi}, i \models
\oblieq{\avariable}{\mathtt{v}_{\avariablebis,\psi}}{\top}$ this means that there
is some $j > i$ such that $\sigma_{\varphi}(i)(\avariable) =
\sigma_{\varphi}(j)(\mathtt{v}_{\avariablebis,\psi})$. By \eqref{it:sigmaphi:3}, we
have that $\sigma_{\varphi}(j)(\mathtt{v}_{\avariablebis,\psi}) \neq
\sigma_{\varphi}(j)(\mathtt{k})$, and by
\eqref{it:sigmaphi:4} this means that
$\sigma_{\varphi}(j)(\mathtt{v}_{\avariablebis,\psi}) =
\sigma_{\varphi}(j)(\avariablebis)$ and $\sigma_{\varphi}, j \models
\psi$. Therefore, $\sigma, i \models
\oblieq{\avariable}{\avariablebis}{\psi}$.  The proof for the base
case of the form $\aformulabis = \avariable \oblieqlocal \mynext^i
\avariablebis$ and for all the cases of the induction step are by an
easy verification since $(\cdot)^\top$ is homomorphic. Hence, \eqref{eq:psi-equiv-psitop} holds.

\smallskip

[$\Leftarrow$] Suppose that $\sigma \models \varphi'$. Note that due
to $\textit{const}$ condition \eqref{it:sigmaphi:3} holds in $\sigma$,
and due to $\textit{val-}\mathtt{v}_{\avariable,\psi}$, condition
\eqref{it:sigmaphi:4} holds.  We show that for every position $i$ and
subformula $\gamma$ of $\varphi$, if $\sigma, i \models \gamma^\top$,
then $\sigma, i \models \gamma$.

The only interesting case is when $\gamma =
\oblieq{\avariable}{\avariablebis}{\psi}$, and $\gamma^\top =
\oblieq{\avariable}{\mathtt{v}_{\avariablebis,\psi}}{\top}$, since for all
boolean and temporal operators $(\cdot)^\top$ is homomorphic. If
$\sigma, i \models \oblieq{\avariable}{\mathtt{v}_{\avariablebis,\psi}}{\top}$
there must be some $j>i$ so that $\sigma(j)(\mathtt{v}_{\avariablebis,\psi}) =
\sigma(i)(\avariable)$. By condition \eqref{it:sigmaphi:3},
$\sigma(j)(\mathtt{k}) \neq \sigma(j)(\mathtt{v}_{\avariablebis,\psi}) =
\sigma(i)(\avariable)$, and by condition \eqref{it:sigmaphi:4}, this
means that $\sigma, j \models \psi$ and
$\sigma(j)(\mathtt{v}_{\avariablebis,\psi}) = \sigma(j)(\avariablebis)$. Hence,
$\sigma, i \models\oblieq{\avariable}{\avariablebis}{\psi}$. The remaining cases are straightforward since $(\cdot)^\top$ is homomorphic on temporal and boolean operators.
\end{proof}
Finally, note that we can extend this coding to treat past obligations in the obvious way. 
\end{proof}

\else \begin{proof}[Proof sketch]
  For every $\varphi \in
  \mainlogiceq$, we compute in polynomial time $\varphi' \in \mainlogic^{\top}$ that
  preserves satisfiability. Besides all
  the variables from $\varphi$, $\varphi'$ uses a new distinguished
  variable $\mathtt{k}$, and a variable
  $\mathtt{v}_{\avariablebis,\psi}$ for every subformula $\psi$ of
  $\varphi$ and every variable $\avariablebis$ of $\varphi$.
  We enforce $\mathtt{k}$ to have a constant value different from all
  values of variables of $\varphi$. At every
  \begin{center}
    \includegraphics{EliminatePosTest}
  \end{center}
  position, we enforce
  $\psi$ to hold if $\mathtt{v}_{\avariablebis,\psi} \approx \avariablebis$,
  and $\psi$ not to  hold if $\mathtt{v}_{\avariablebis,\psi} \approx
  \mathtt{k}$ as shown above. Then
  $\oblieq{\avariable}{\avariablebis}{\psi}$ is replaced by
  $\oblieq{\avariable}{\mathtt{v}_{\avariablebis,\psi}}{\top}$.
\end{proof}
\fi

By combining Proposition~\ref{prop:LRV2LRVapprox} and Proposition~\ref{prop:LRVapprox-to-LRVtop},
we get the following result. 

\begin{cor}\label{cor:reduction-LRVtop}
There is a polynomial-time reduction from
    \SAT{$\mainlogic$} into \SAT{$\mainlogic^{\top}$} [resp.  from
    \SAT{$\pmainlogic$} into \SAT{$\pmainlogic^{\top}$}].
  \end{cor}

  Since \SAT{$\pmainlogic^{\top}$} is decidable \cite{Demri&DSouza&Gascon12}, we obtain the decidability of 
\SAT{$\pmainlogic$}.

\begin{cor}
    \SAT{$\pmainlogic$} is decidable.
  \end{cor}

We have seen  how to eliminate test
formulas $\aformula$ from
$\oblieq{\avariable}{\avariablebis}{\aformula}$ and
$\poblieq{\avariable}{\avariablebis}{\aformula}$. Combining this with the decidability proof for $\pmainlogic^{\top}$ satisfiability
from~\cite{Demri&DSouza&Gascon12}, we get  that
 both 

\begin{cor}
\SAT{$\pmainlogic$} and \SAT{$\pmainlogic^{\top}$}
are equivalent to 
Reach(VASS)
modulo polynomial-space reductions.
\end{cor}
 
\makeatletter{}\subsection{From $\mainlogic^{\top}$  Satisfiability to Control State Reachability}
\label{section-reduction-from-top}
Recall that in~\cite{Demri&DSouza&Gascon12}, \SAT{$\mainlogic^{\top}$} is reduced
to the reachability problem for a subclass of VASS. 
Herein, this is refined  by introducing \defstyle{incremental errors}
in order to improve the complexity. 

In~\cite{Demri&DSouza&Gascon12}, the standard concept of atoms from
the Vardi-Wolper construction of automaton for LTL is used (see also Appendix~\ref{section-appendix-aphifrompreviouspaper}). Refer to
the diagram at the top of Figure~\ref{fig:OptionalDecrExample}. 

\begin{figure}[!htp]
  \center
  \includegraphics[scale=.9]{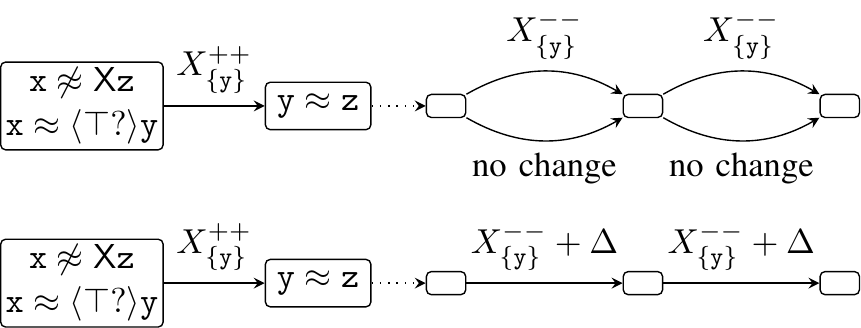}
  \caption{Automaton constructions from
  \cite{Demri&DSouza&Gascon12} (top) and from this paper (bottom).}
  \label{fig:OptionalDecrExample}
\end{figure}
The
formula $\oblieq{\avariable}{\avariablebis}{\top}$ in the left atom
creates an obligation for the current value
of $\avariable$ to appear some time in the future in $\avariablebis$.
This obligation cannot be satisfied in the second atom, since
$\avariablebis$ has to satisfy some other constraint there
($\avariablebis \oblieqlocal \avariableter$). To remember this
unsatisfied obligation about $\avariablebis$ while taking the
transition from the first atom to the second, the counter
$\aset_{\set{\avariablebis}}$ is incremented.  The counter can be
decremented later in transitions that allow the repetition in
$\avariablebis$. If several transitions allow such a repetition, only
one of them needs to decrement the counter (since there was only one
obligation at the beginning). The other transitions which should not
decrement the counter can take the alternative labelled ``no change''
in the right part of Figure~\ref{fig:OptionalDecrExample}.

The idea here is to replace the combination of the decrementing
transition and the ``no change'' transition in the top of
Figure~\ref{fig:OptionalDecrExample} with a single transition with
incremental errors as shown in the bottom. After Lemma~\ref{lem:APhiFromEarlierPaper} below
that formalises ideas from~\cite[Section 7]{Demri&DSouza&Gascon12} (see also Appendix~\ref{section-appendix-aphifrompreviouspaper}), 
we prove that
the transition with incremental errors is sufficient.

\begin{lem}[\cite{Demri&DSouza&Gascon12}]
  \label{lem:APhiFromEarlierPaper}
  For a $\mainlogic^{\top}$ formula $\aformula$ that
  uses the variables $\set{\avariable_{1}, \ldots,
  \avariable_{\numvars}}$, a VASS $\Aphi = \langle
  \states,\counters, \transitions \rangle$ can be defined, along with sets
  $\states_{0}, \states_{f} \subseteq \states$ of \defstyle{initial}
  and \defstyle{final} states resp., such that
  \begin{itemize}
  \itemsep 0 cm 
  \item the set of counters $\counters$ consists of all nonempty
    subsets of $ \{\avariable_{1}, \dots, \avariable_{\numvars}\}$.
  \item For all $q,q' \in \states$, either $ \set{\aupfunc \mid
    (q,\aupfunc,q') \in \transitions} = [n_1,m_1] \times \dotsb \times
    [n_{|\counters|},m_{|\counters|}]$ for  some $n_1,m_1, \dotsc,
    n_{|\counters|},m_{|\counters|} \in [-\numvars, \numvars]$, or $\set{\aupfunc \mid (q,\aupfunc,q') \in
    \transitions}=\emptyset$. We call this property \defstyle{closure
    under component-wise interpolation}.
    \item If $\transitions \cap (\set{\astate} \times [-\numvars,
      \numvars]^{\counters} \times \set{\astate'})$ is not empty, then
      for every $X \in \counters$ there is $(q,\aupfunc,q') \in
      \transitions$ so that $\aupfunc(X) \geq 0$.
      We call this property \defstyle{optional decrement}.
    \item Let $\vect{0}$ be the counter valuation that assigns $0$ to all
  counters. Then $\aformula$ is satisfiable iff $\langle \ainitst,
  \vect{0} \rangle \step{*} \langle \afinst, \vect{0} \rangle$ for
  some $\ainitst \in \states_{0}$ and $\afinst \in \states_{f}$.
  \end{itemize}
\end{lem}\smallskip

\noindent At the top of Figure~\ref{fig:OptionalDecrExample}, only one counter
$\aset_{\set{\avariablebis}}$ is shown and is decremented by $1$ for
simplicity. In general, multiple counters can be changed and they can
be incremented/de\-cre\-men\-ted by any number up to $\numvars$, depending
on the initial and target atoms of the transition. If a
counter can be incremented by $\numvars_{1}$ and can be decremented by
$\numvars_{2}$, then there will also be
transitions between the same pair of atoms allowing changes of
$\numvars_{1}-1, \ldots, 1, 0 , -1, \ldots -(\numvars_{2}-1)$. This
corresponds to the closure under component-wise interpolation
mentioned in Lemma~\ref{lem:APhiFromEarlierPaper}. 
The optional decrement property
corresponds to the fact that there will always be a ``no change''
transition that does not decrement any counter.

Now, we show that a single transition that decrements all counters by
the maximal possible number can simulate the set of all transitions
between two atoms, using incremental errors.

Let $\Ainc = \langle
\states, \counters, \transitions^{\mathit{min}} \rangle$ and
$\states_{0}, \states_{f} \subseteq \states$, where $\states,
\states_{0},\states_{f}$ and $\counters$ are same as those of $\Aphi$
and $\transitions^{\mathit{min}}$ is defined as follows: $(\astate,
\minupfunc_{\astate, \astate'}, \astate') \in
\transitions^{\mathit{min}}$ iff $\transitions \cap (\set{\astate}
\times [-\numvars, \numvars]^{\counters} \times \set{\astate'})$ is
not empty and $\minupfunc_{\astate, \astate'}(\aset) \egdef \min_{\aupfunc:
(\astate, \aupfunc, \astate') \in \transitions}\set{\aupfunc(\aset)}$
for all $\aset \in \counters$.
Similarly, $\maxupfunc_{\astate, \astate'}(\aset) \egdef \max_{\aupfunc:
(\astate, \aupfunc, \astate') \in \transitions}\set{\aupfunc(\aset)}$
for all $\aset \in \counters$

\begin{lem}
  \label{lem:IncAutSimReal}
  If $\langle \astate, \acountval \rangle \step{} \langle \astate',
  \acountval' \rangle$ in $\Aphi$, then $\langle \astate, \acountval
  \rangle \stepinc{} \langle \astate', \acountval' \rangle$ in
  $\Ainc$.
\end{lem}
\makeatletter{}\ifLONG
\begin{proof}
  Since $\langle \astate, \acountval \rangle \step{} \langle \astate',
  \acountval' \rangle$, there is a transition $(\astate, \aupfunc,
  \astate') \in \transitions$ such that $\aupfunc + \acountval =
  \acountval'$.  By definition of $\Ainc$, there is
  a transition $(\astate, \minupfunc_{\astate, \astate'}, \astate')
  \in \transitions^{\mathit{min}}$ where $\aupfunc -
  \minupfunc_{\astate, \astate'} \succeq \vect 0$. Let $\incer =
  \aupfunc - \minupfunc_{\astate, \astate'}$ (this will be the
  incremental error used in $\Ainc$).  Now we have $\langle \astate,
  \acountval + \incer \rangle \step{\minupfunc_{\astate, \astate'}}
  \langle \astate', \acountval + \aupfunc\rangle = \langle \astate',
  \acountval'\rangle$ in $\Ainc$. Hence, $\langle \astate, \acountval
  \rangle \stepinc{} \langle \astate', \acountval' \rangle$ in
  $\Ainc$.
\end{proof}
\else
\begin{proof}[Proof sketch]
  If $\Aphi$ takes one of the transitions in the top of
  Figure~\ref{fig:OptionalDecrExample}, $\Ainc$ takes the
  corresponding transition at the bottom, adjusting the incremental
  error $\Delta$ accordingly.
\end{proof} 
\fi
 
\begin{lem}
  \label{lem:RealAutSimInc}
  If $\langle \astate_{1}, \acountval_{1} \rangle \stepinc { *}
  \langle \astate_{2}, \vect{0} \rangle$ in $\Ainc$ and $
  \acountval_{1}' \preceq \acountval_{1}$, then $\langle \astate_{1},
  \acountval_{1}' \rangle \step { *} \langle \astate_{2}, \vect{0}
  \rangle$ in $\Aphi$.
\end{lem}
\makeatletter{}\ifLONG
\begin{proof}
  By induction on the length $\runlen$ of the run $\langle
  \astate_{1}, \acountval_{1} \rangle \stepinc { *} \langle
  \astate_{2}, \vect{0} \rangle$. The base case $\runlen = 0$ is
  trivial since there is no change in the configuration.

  Induction step: the idea is to simulate the first gainy transition
  by a normal transition that decreases each counter as much as
  possible while ensuring that (1) the resulting value is non-negative
  and (2) we can apply the induction hypothesis to the resulting
  valuation. We calculate the update required for each counter
  individually and by closure under component-wise interpolation, there
  will always be a transition with the required update function. Let
  $\langle \astate_{1}, \acountval_{1} \rangle \stepinc{} \langle
  \astate_{3}, \acountval_{3} \rangle \stepinc { *} \langle
  \astate_{2}, \vect{0} \rangle$ and $\acountval_{1}' \preceq
  \acountval_{1}$. We will define an update function $\aupfunc'$ such
  that $\langle \astate_{1}, \acountval_{1}' \rangle \step{} \langle
  \astate_{3}, \acountval_{3}' \rangle$, $\acountval_{3}' (\aset) =
  \acountval_{1}' (\aset) + \aupfunc' ( \aset)$ for each $\aset \in
  \counters$ and $ \acountval_{3}' \preceq \acountval_{3}$. For each
  counter $ \aset \in \counters$, $ \aupfunc' ( \aset)$ is defined as
  follows:

\noindent
  \emph{Case 1}: $ \acountval_{1}' ( \aset) + \minupfunc_{\astate_{1},
  \astate_{3}} ( \aset) \ge 0$. Let $ \aupfunc' ( \aset) =
  \minupfunc_{\astate_{1}, \astate_{3}} ( \aset)$.
  \begin{align*}
    \acountval_{3} ( \aset ) & \ge \acountval_{1} ( \aset) +
    \minupfunc_{\astate_{1}, \astate_{3}} (
    \aset) & \text{[From the semantics of } \Ainc \text{]}\\
    & \ge \acountval_{1}' ( \aset) +
    \minupfunc_{\astate_{1}, \astate_{3}} (
    \aset) & \text{[Since } \acountval_{1}' \preceq \acountval_{1}
    \text{]}\\
    & = \acountval_{1}' ( \aset) + \aupfunc'( \aset) & \text{[By definition of } \aupfunc'
    ( \aset) \text{]}\\
    & = \acountval_{3}' ( \aset) &
  \end{align*}

\noindent
  \emph{Case 2}: $ \acountval_{1}' ( \aset) + \minupfunc_{\astate_{1},
  \astate_{3}} ( \aset) < 0$. Therefore, $\minupfunc_{\astate_{1},
  \astate_{3}} ( \aset) < -\acountval_{1} ' ( \aset)$.  Moreover,
  since $\maxupfunc_{\astate_{1}, \astate_{3}} ( \aset) \ge 0$ (due to
  the optional decrement property) and $ \acountval_{1} ' ( \aset) \ge
  0$, $-\acountval_{1} ' ( \aset) \le
  \maxupfunc_{\astate_{1}, \astate_{3}} ( \aset)$. Let $ \aupfunc' (
  \aset) = -\acountval_{1}' ( \aset) $. Now, $ \acountval_{3} ' (
  \aset) = \acountval_{1}' ( \aset) + \aupfunc'( \aset) = 0 \preceq
  \acountval_{3} ( \aset)$.

  By definition of $\minupfunc_{\astate_{1}, \astate_{3}}$ and the
  closure of the set of transitions $\transitions$ of $\Aphi$ under
  component-wise interpolation, we have $(\astate_{1}, \aupfunc',
  \astate_{3}) \in \transitions$ and hence $\langle \astate_{1},
  \acountval_{1}' \rangle \step{}\langle \astate_{3}, \acountval_{3}'
  \rangle$. Since $ \acountval_{3}' \preceq \acountval_3$ and $\langle
  \astate_{3}, \acountval_{3} \rangle \stepinc { *} \langle
  \astate_{2}, \vect{0} \rangle$, we can use the induction hypothesis
  to conclude that $\langle \astate_{3}, \acountval_{3}' \rangle \step
  { *} \langle \astate_{2}, \vect{0} \rangle$. So we conclude that
  $\langle \astate_{1}, \acountval_{1}' \rangle \step{}\langle
  \astate_{3}, \acountval_{3}' \rangle \step {*} \langle \astate_{2},
  \vect{0} \rangle$.
\end{proof} 
\else
 \begin{proof}[Proof sketch]
   If $\Ainc$ takes a transition at the bottom of
   Figure~\ref{fig:OptionalDecrExample}, $\Aphi$ takes the
   corresponding decrementing transition at the top, ignoring any
   incremental errors. This may result in $\Aphi$ reaching $\vect{0}$
   earlier in the run, in which case ``no change'' transitions are
   used in the rest of the run.
 \end{proof}
\fi

\begin{thm}
  \label{thm:2ExpspUpBound}
  \SAT{$\mainlogic^{\top}$} is in
  \twoexpspace{}.
\end{thm}
\makeatletter{}\ifLONG
\begin{proof}
  The proof is in four steps.

\begin{description}
\itemsep 0 cm 

\item[Step 1]
  From \cite{Demri&DSouza&Gascon12}, a $\mainlogic^{\top}$ formula
  $\aformula$ is satisfiable iff $\langle \ainitst, \vect{0} \rangle
  \step{*} \langle \afinst, \vect{0} \rangle$ in $\Aphi$ for some
  $\ainitst \in \states_{0}$ and $\afinst \in \states_{f}$.
  
\item[Step 2] This is the step that requires new insight. From
  Lemmas \ref{lem:IncAutSimReal} and \ref{lem:RealAutSimInc}, $\langle
  \ainitst, \vect{0} \rangle \step{*} \langle \afinst, \vect{0}
  \rangle$ in $\Aphi$ iff  $\langle \ainitst, \vect{0} \rangle
  \stepinc{*} \langle \afinst, \vect{0} \rangle$ in $\Ainc$.

\item[Step 3] This is a standard trick.  Let $\Adec = \langle
  \states,\counters, \transitions^{\mathit{rev}} \rangle$ be a VASS
  such that for every transition $(\astate,\aupfunc,\astate') \in
  \transitions^{\mathit{min}}$ of $\Ainc$, $\Adec$ has a transition
  $(\astate',-\aupfunc,\astate) \in \transitions^{\mathit{rev}}$,
  where $-\aupfunc: \counters \to \Zed$ is the function such that
  $-\aupfunc (\aset) = -1 \times \aupfunc(\aset)$ for all $\aset \in
  \counters$.  We infer that $\langle \ainitst, \vect{0} \rangle
  \stepinc{*} \langle \afinst, \vect{0} \rangle$ in $\Ainc$ iff
  $\langle \afinst, \vect{0} \rangle \stepdec{*} \langle \ainitst,
  \vect{0} \rangle$ in $\Adec$ (we can simply reverse every transition
  in the run of $\Ainc$ to get a run of $\Adec$ and vice-versa).
  
\item[Step 4] This is another standard trick.
  Since $\Adec$ does not have zero-tests, we can remove all
  decrementing errors from a run of $\Adec$ from $\langle \afinst,
  \vect{0} \rangle$ to $\langle \ainitst, \vect{0} \rangle$, to get
  another run from $\langle \afinst, \vect{0} \rangle$ to $\langle
  \ainitst, \acountval \rangle$, where $\acountval$ is some counter
  valuation (possibly different from $\vect{0}$). Using decremental
  errors at the last configuration, $\Adec$ can then reach the
  configuration $\langle \ainitst, \vect{0} \rangle$. In other words,
  $\langle \afinst, \vect{0} \rangle \stepdec{*} \langle \ainitst,
  \vect{0} \rangle$ iff $\langle \afinst, \vect{0} \rangle \step{*}
  \langle \ainitst, \acountval \rangle$ for some counter valuation
  $\acountval$. Checking the latter condition is precisely the control
  state reachability problem for VASS.

\end{description} 
 
\noindent If the control state in the above instance is reachable, then
  Rackoff's proof gives a bound on the length of a shortest run reaching it
  \cite{Rackoff78} (see also~\cite{Demrietal09}). The bound is doubly exponential in the size of the
  VASS. Since in our case, the size of the VASS is exponential in the
  size of the $\mainlogic^{\top}$ formula $\aformula$, the bound is
  triply exponential. A non-deterministic Turing machine can maintain
  a binary counter to count up to this bound, using doubly exponential
  space.  The machine can start by guessing some initial state
  $\ainitst$ and a counter valuation set to $\vect{0}$.  This can be
  done in polynomial space. In one step, the machine guesses a
  transition to be applied next and updates the current configuration
  accordingly, while incrementing the binary counter. At any step, the
  space required to store the current configuration is at most doubly
  exponential. By the time the binary counter reaches its triply
  exponential bound, if a final control state is not reached, the
  machine rejects its input. Otherwise, it accepts. Since this
  non-deterministic machine operates in doubly exponential space, an
  application of Savitch's Theorem~\cite{Savitch70} gives us the required
  \twoexpspace{} upper bound for the satisfiability problem of
  $\mainlogic^{\top}$.
\end{proof}
\else
\begin{proof}[Proof sketch]
  The proof is in four steps (standard arguments are used for 3. and 4.).
  \begin{itemize}
  \itemsep 0 cm 
  \item[1] From \cite{Demri&DSouza&Gascon12}, a
    $\mainlogic^{\top}$ formula $\aformula$ is satisfiable if{f}
    $\langle \ainitst, \vect{0} \rangle \step{*} \langle \afinst,
    \vect{0} \rangle$ in $\Aphi$ for some $\ainitst \in \states_{0}$
    and $\afinst \in \states_{f}$.

  \item[2] This is the step that requires new insight. From
    Lemmas \ref{lem:IncAutSimReal} and \ref{lem:RealAutSimInc},
    $\langle \ainitst, \vect{0} \rangle \step{*} \langle \afinst,
    \vect{0} \rangle$ in $\Aphi$ if{f} $\langle \ainitst, \vect{0}
    \rangle \stepinc{*} \langle \afinst, \vect{0} \rangle$ in $\Ainc$.

  \item[3]  Let $\Adec$ be the VASS
    obtained from $\Ainc$ by ``reversing'' each transition. By
    replacing each gainy transition of $\Ainc$ by the reverse lossy
    transtion of $\Adec$, we infer that $\langle \ainitst, \vect{0}
    \rangle \stepinc{*} \langle \afinst, \vect{0} \rangle$ in $\Ainc$
    if{f} $\langle \afinst, \vect{0} \rangle \stepdec{*} \langle
    \ainitst, \vect{0} \rangle$ in $\Adec$.

  \item[4] $\langle \afinst,
    \vect{0} \rangle \stepdec{*} \langle \ainitst, \vect{0} \rangle$
    if{f} $\langle \afinst, \vect{0} \rangle \step{*} \langle \ainitst,
    \acountval \rangle$ for some 
    $\acountval$. Checking the latter condition is precisely the
    control state reachability problem for VASS.
  \end{itemize}
The number of control states and counters in $\Aphi$ and in
$\Adec$ is  exponential in $\size{\aformula}$ (size of $\aformula$).  Each
control state and transition function of $\Adec$ can be represented in
space polynomial in $\size{\aformula}$. Given a control state,
testing if it is an initial or a final state can be done in linear time
in the size of the state.  The automaton $\Adec$ can be
constructed in exponential time  in $\size{\aformula}$. Hence, the
\expspace{} upper bound for the control state reachability problem in
VASS \cite{Rackoff78} gives the \twoexpspace{} upper bound for \SAT{$\mainlogic^{\top}$}.
\end{proof}
\fi

\begin{cor}\label{cor:mainlogic-twoexpspace}
    \SAT{$\mainlogic$} is in \twoexpspace.
  \end{cor}

\makeatletter{}\section{Simulating Exponentially Many Counters}
\label{section-simulating}

In Section~\ref{section-upper-bound}, the reduction from \SAT{$\mainlogic$}
to control state reachability for VASS involves an exponential blow up, since
we use one counter  for each nonempty subset of variables. The question whether
this can be avoided depends on whether $\mainlogic$ is powerful enough
to reason about subsets of variables or whether there is a smarter
reduction without a blow-up. 
Similar
questions are open in other related areas~\cite{Manuel&Ramanujan09,Manuel&Ramanujam12}.

Here we prove that $\mainlogic$ is indeed powerful enough to reason
about subsets of variables. We establish a \twoexpspace{} lower
bound. The main idea behind this lower bound proof is that the power
of LRV to reason about subsets of variables can be used to simulate
exponentially many counters.
The lower bound is developed in three parts, with each part explained
in a sub-section of this section. The first part defines chain
systems, which are like VASS, except that transitions can not access
counters directly, but can access a pointer to a chain of counters.
The second part shows lower bounds for the control state reachability
problem for chain systems.  The third part shows that $\mainlogic$ can
reason about chain systems.

\makeatletter{}\subsection{Chain Systems}
\label{section-chain-automata}We introduce a new class of counter systems that is instrumental to show that
 \SAT{$\mainlogic^{\top}$} is \twoexpspace-hard. This is an intermediate formalism 
between counter automata with zero-tests with counters 
bounded by triple exponential values (having  \twoexpspace-hard control state reachability problem)
and properties expressed in  $\mainlogic^{\top}$.  Systems with chained counters have no zero-tests
and the only updates are increments and decrements. However, the systems are equipped
with a finite family of counters, each family having an exponential number of counters in some part of the input
(see details below). 
Let $\mathtt{exp}(0, n) \egdef n$ and $\mathtt{exp}(k+1,n) \egdef 2^{\mathtt{exp}(k,n)}$ for every $k \geq 0$. 

\begin{defi}
\label{definition-ca-chained-counters}
A \defstyle{counter system with chained counters} (herein called a \defstyle{chain system}) 
is a tuple 
$\aautomaton = \triple{\locations,\amap,k}{\locations_0,\locations_F}{\transitions}$ where 
\begin{enumerate}
\itemsep 0 cm
\item $\amap: \interval{1}{n} \rightarrow \Nat$ where $n \geq 1$ is the \defstyle{number of chains}
      and $\mathtt{exp}(k,\amap(\achain))$  is  the number of counters for the chain $\achain$ where $k \geq 0$,
\item $\locations$ is a non-empty finite set of \defstyle{states},
\item $\locations_0 \subseteq \locations$ is the set of \defstyle{initial states} and 
      $\locations_F \subseteq \locations$ is the set of \defstyle{final states},
\item $\transitions$ is the set of \defstyle{transitions} in $\locations \times \instructions \times \locations$
      where 
      \begin{align*}
        \instructions = \set{ &\inc{\achain}, \dec{\achain},
          \cnext{\achain}, \cprevious{\achain}, \cisfirst{\achain},
          \cisnotfirst{\achain}, \\
                             &\cislast{\achain},
          \cisnotlast{\achain}: \achain \in \interval{1}{n} } .
      \end{align*}
\end{enumerate}
\end{defi}\smallskip

\noindent By convention, sometimes, we write $\alocation \step{\ainstruction} \alocation'$
instead of $\triple{\alocation}{\ainstruction}{\alocation'} \in \transitions$. 

The system
$\aautomaton = \triple{\locations,\amap,k}{\locations_0,\locations_F}{\transitions}$ is said to be 
at \defstyle{level k}.
In order to encode the natural numbers from $\amap$ and the value $k$, we  use a unary representation. 
We say that a transition containing $\inc{\achain}$ is \defstyle{$\achain$-incrementing}, and a transition containing 
$\dec{\achain}$ is \defstyle{$\achain$-decrementing}. 
The idea is that for each chain $\achain \in \interval{1}{n}$, we have $\mathtt{exp}(k,\amap(\achain))$ 
counters, but we cannot 
access them directly in the transitions as we do in VASS. Instead, we have a 
pointer to a counter that we can move. We can ask if we are pointing to the first counter ($\cisfirst{\achain}$) 
or not ($\cisnotfirst{\achain}$), or to the last counter ($\cislast{\achain}$) or not 
($\cisnotlast{\achain}$), and we can change the pointer to the next ($\cnext{\achain}$) or 
previous ($\cprevious{\achain}$) counter.

A \defstyle{run} is a finite sequence $\arun$ in $\transitions^*$ such that
\begin{enumerate}
\itemsep 0 cm
\item for every two $\arun(i) = \alocation \step{\ainstruction} \alocationbis$
 and 
     $\arun(i+1) = \alocation' \step{\ainstruction'} \alocationbis'$ we have
  $\alocationbis = \alocation'$,
\item for every chain $\achain \in \interval{1}{n}$, for every $i \in \interval{1}{|\arun|}$,
we have $0 \leq c_i^{\achain} < \mathtt{exp}(k,\amap(\achain))$ where 
\begin{align}\label{eq:c-i-chain}
  \begin{split}
    c_i^{\achain} &= \card{\set{i' < i \mid \arun(i') = \alocation
        \step{\cnext{\achain}} \alocationbis}} ~-~~~ \\
    &~~~~\card{\set{i' < i \mid \arun(i') = \alocation
        \step{\cprevious{\achain}} \alocationbis }},
  \end{split}
\end{align}
\item  for every $i \in \interval{1}{|\arun|}$ and for every chain $\achain \in \interval{1}{n}$, 
  \begin{enumerate}[label=\({\alph*}]
  \item if $\arun(i) = \alocation \step{\cisfirst{\achain}}
    \alocation'$, then $c_i^{\achain} = 0$;
  \item if $\arun(i) =
    \alocation \step{\cisnotfirst{\achain}} \alocation'$, then
    $c_i^{\achain} \neq 0$;
  \item if $\arun(i) = \alocation
    \step{\cislast{\achain}} \alocation'$, then $c_i^{\achain} =
    \mathtt{exp}(k,\amap(\achain)) - 1$;
  \item if $\arun(i) = \alocation
    \step{\cisnotlast{\achain}} \alocation'$, then $c_i^{\achain} \neq
    \mathtt{exp}(k,\amap(\achain)) -1$.
  \end{enumerate}
\end{enumerate}\medskip

\noindent A run is \defstyle{accepting} whenever $\arun(1)$ starts with an initial state from $\locations_0$ and
$\arun(\length{\arun})$ ends with a final state from $\locations_F$. 
A run is \defstyle{perfect} iff for every $\achain \in \interval{1}{n}$, there is some
injective function 
\[
\gamma : \set{i \mid \arun(i) \text{ is}  \ \achain \text{-decrementing}} \to \set{i \mid \arun(i) 
\text{ is} \ \achain \text{-incrementing}}
\]
such that for every $\gamma(i) =j$ we have that  $j < i$ and $c_i^{\achain} = c_j^{\achain}$.
A run is \defstyle{gainy and ends at zero} (different from `ends in zero' defined below) iff 
for every chain  $\achain \in \interval{1}{n}$, there is some injective function
\[
\gamma : \set{i \mid \arun(i) \text{ is}  \ \achain \text{-incrementing}} \to \set{i \mid \arun(i) 
\text{ is} \ \achain \text{-decrementing}}
\]
such that for every $\gamma(i) =j$ we have that  $j > i$ and $c_i^{\achain} = c_j^{\achain}$.
In the sequel, we shall simply say that the run is gainy. 
Below, we define two problems for which we shall characterize the computational complexity.

\begin{center}
  \begin{tabular}{|@{\hspace{.1cm}}r@{\hspace{.1cm}}l@{\hspace{.1cm}}|}
    \hline
    \textsc{Problem:}& Existence of a perfect  accepting run \\& of level $k \geq 0$ ~~(Per($k$))
    \\
    \hline
    \textsc{Input:} & A chain system $\aautomaton$ of level $k$.
    \\
    \textsc{Question:} & Does $\aautomaton$ have a \emph{perfect} accepting run?\\
    \hline
  \end{tabular}

  \medskip

  \begin{tabular}{|@{\hspace{.1cm}}r@{\hspace{.1cm}}l@{\hspace{.1cm}}|}
    \hline
    \textsc{Problem:}& Existence of a gainy  accepting run \\& of level $k \geq 0$ ~~(Gainy($k$))
    \\
    \hline
    \textsc{Input:} & A chain system $\aautomaton$  of level $k$.
    \\
    \textsc{Question:} & Does $\aautomaton$ have a \emph{gainy} accepting run?\\
    \hline
  \end{tabular}
\end{center}\medskip

\noindent Per($k$) is actually a control state reachability problem in VASS
where the counters are encoded succinctly. 
 Gainy($k$) is  a reachability problem in  VASS with
incrementing errors and the reached counter values are equal to zero. Here, the counters
are encoded succinctly too.

 Let $\aautomaton = \triple{\locations,\amap,k}{\locations_0,\locations_F}{\transitions}$ be a 
chain system of level $k$. A run $\arun$ \defstyle{ends in zero} whenever
for every chain $\achain \in \interval{1}{n}$, $c_{L}^{\achain} = 0$ with $L = \length{\arun}$.
We write Per$^{{\rm zero}}$($k$) and  Gainy$^{{\rm zero}}$($k$) to denote the variant problems of
Per($k$) and Gainy($k$), respectively, in which runs that end in zero are considered. 
First, note that  Per$^{{\rm zero}}$($k$) and Per($k$) are  interreducible in logarithmic space since it is
always possible to add adequately self-loops when a final state is reached in order to guarantee that 
 $c_{L}^{\achain} = 0$ for every chain $\achain \in \interval{1}{n}$. Similarly, 
Gainy$^{{\rm zero}}$($k$) and Gainy($k$) are  interreducible in logarithmic space. 

\begin{lem} \label{lemma-inter}
For every $k \geq 0$, Per($k$) and Gainy($k$) are interreducible in logarithmic space.
\end{lem}

\makeatletter{}\begin{proof}
Below, we show that  Per$^{{\rm zero}}$($k$) and Gainy$^{{\rm zero}}$($k$)  are interreducible in logarithmic space,
which allows us to get the proof of the lemma since logarithmic-space reductions are closed under composition.
From $\aautomaton$, let us define a reverse chain system $\reverse{\aautomaton}$ of level $k$,
where the reverse operation $\tilde{\cdot}$ is defined on instructions, transitions, sets of transitions
and systems as follows:
\begin{itemize}
\itemsep 0 cm
\item $\reverse{\inc{\achain}} \egdef \dec{\achain}$; $\reverse{\dec{\achain}} \egdef \inc{\achain}$; 
      $\reverse{\cnext{\achain}} \egdef \cprevious{\achain}$; $\reverse{\cprevious{\achain}} \egdef \cnext{\achain}$,
\item $\reverse{\cisfirst{\achain}} \egdef \cislast{\achain}$; $\reverse{\cislast{\achain}} \egdef \cisfirst{\achain}$;
      $\reverse{\cisnotfirst{\achain}} \egdef \cisnotlast{\achain}$; $\reverse{\cisnotlast{\achain}} \egdef 
      \cisnotfirst{\achain}$,
\item $\reverse{\alocation \step{\ainstruction} \alocation'} \egdef
       \alocation' \step{\reverse{\ainstruction}} \alocation$,
\item $\reverse{\aautomaton} \egdef  \triple{\locations,\amap,k}{\locations_F,\locations_0}{\reverse{\transitions}}$
      with  $\reverse{\transitions} \egdef 
      \set{\reverse{\alocation \step{\ainstruction} \alocation'}: \alocation \step{\ainstruction} 
      \alocation' \in \transitions}$. Note that $\locations_0$ and $\locations_F$ have been swapped. 
\end{itemize}\medskip

\noindent The reverse operation can be extended to sequences of transitions as follows:
$\reverse{\varepsilon} \egdef \varepsilon$ and $\reverse{\atransition \cdot \aword} \egdef
\reverse{\aword} \cdot \reverse{\atransition}$. 
Note that $\tilde{\cdot}$ extends the reverse operation defined in the proof 
of Theorem~\ref{thm:2ExpspUpBound} (step 3). 

One can show the following implications, for any run $\rho$ fo $\aautomaton$ that ends in zero:
\begin{enumerate}
\itemsep 0 cm 
\item $\arun$ is perfect and accepting for $\aautomaton$ implies 
      $\reverse{\arun}$ is gainy  and accepting for $\reverse{\aautomaton}$.
\item $\arun$ is gainy  and accepting for $\reverse{\aautomaton}$ implies $\reverse{\arun}$ is perfect 
       and accepting for $\aautomaton$.
\item $\arun$ is gainy  and accepting for $\aautomaton$ implies 
      $\reverse{\arun}$ is perfect  and accepting for $\reverse{\aautomaton}$.
\item $\arun$ is perfect  and accepting for $\reverse{\aautomaton}$ implies $\reverse{\arun}$ is gainy 
       and accepting for $\aautomaton$.
\end{enumerate}

(1) and (4) [resp. (2) and (3)] have very similar proofs because $\reverse{\cdot}^{-1}$ is actually
equal to $\reverse{\cdot}$. 
(1)--(4) are  sufficient to establish that Per$^{{\rm zero}}$($k$) and Gainy$^{{\rm zero}}$($k$)  are 
interreducible in logarithmic space.
\end{proof}

\begin{lem}  
\label{lemma-upper-bound-perk}
Per($k$) is in $(k+1) \expspace$.
\end{lem}

The proof of Lemma~\ref{lemma-upper-bound-perk} consists of simulating perfect runs
by the runs of a VASS in which the control states record the positions of the pointers in the chains. 

\makeatletter{}\begin{proof}
It is sufficient to show that Per$^{\rm zero}$($k$) is in $(k+1) \expspace$.

Let $\aautomaton = \triple{\locations,\amap,k}{\locations_0,\locations_F}{\transitions}$ be a 
chain system with $\amap: \interval{1}{n} \rightarrow \Nat$. 
We reduce this instance of Per$^{\rm zero}$($k$) into several  
instances of the control state reachability problem for VASS such that the number of instances is bounded by
$\mathcal{O}(\length{\aautomaton}^2)$ 
and the size of  
each instance is in $\mathcal{O}(\mathtt{exp}(k,\length{\aautomaton})^{\length{\aautomaton}})$, which provides the 
$(k+1) \expspace$ upper bound by~\cite{Rackoff78}.
The only instructions in the VASS $\aautomaton'$ defined below are: increment a counter,
decrement a counter or the $\mathtt{skip}$ action, which just changes the control state without 
modifying the counter values (which can be obviously simulated by an increment followed by a 
decrement). 

Let us define a VASS $\aautomaton' = \pair{\locations',\counters'}{\transitions'}$ with
$$\counters' = \set{\acounter^1_{0}, \ldots,\acounter^1_{\mathtt{exp}(k,\amap(1))-1}, \ldots,\acounter^n_{0}, \ldots,
\acounter^n_{\mathtt{exp}(k,\amap(n))-1}}.
$$
In $\aautomaton'$, it will be possible to access the counters directly
by encoding the positions of the pointers in the states. 
Let $\locations' = \locations \times \interval{0}{\mathtt{exp}(k,\amap(1))-1} \times \cdots \times 
\interval{0}{\mathtt{exp}(k,\amap(n))-1}$. It remains to  define the transition relation $\transitions'$
($\tuple{\beta_1}{\beta_n}$ below is any element in $\interval{0}{\mathtt{exp}(k,\amap(1))-1} \times \cdots \times 
\interval{0}{\mathtt{exp}(k,\amap(n))-1}$, unless conditions apply):
\begin{itemize}
\itemsep 0 cm
\item Whenever $\alocation \step{\inc{\achain}} \alocation' \in \transitions$, we have
      $\tuple{\alocation,\beta_1}{\beta_{n}} \step{\inc{\acounter^{\achain}_{\beta_{\achain}}}} 
      \tuple{\alocation',\beta_1}{\beta_{n}} \in \transitions'$ (and similarly with decrements),
\item Whenever  $\alocation \step{\cnext{\achain}} \alocation' \in \transitions$, we have
      $$\tuple{\alocation,\beta_1, \ldots, \beta_{\achain}}{\beta_{n}} \step{\mathtt{skip}} 
      \tuple{\alocation',\beta_1, \ldots, \beta_{\achain}+1}{\beta_{n}} \in \transitions'$$
     if $\beta_{\achain}+1 < \mathtt{exp}(k,\amap(\achain))$ (and similarly with $\cprevious{\achain}$),
\item When  $\alocation \step{\cisfirst{\achain}} \alocation' \in \transitions$, 
      $\tuple{\alocation,\beta_1}{\beta_{n}} \step{\mathtt{skip}} 
      \tuple{\alocation',\beta_1}{\beta_{n}} \in \transitions'$ 
     if $\beta_{\achain} = 0$ (and similarly with $\cisnotfirst{\achain}$, $\cislast{\achain}$ and
     $\cisnotlast{\achain}$).
\end{itemize}
It is easy to show that there is a perfect accepting run that ends in zero  if{f}
there are $\mathbf{\alocation}_0 \in \locations_0 \times \set{0}^n$
and $\mathbf{\alocation}_F \in \locations_F \times  \set{0}^n$
such that
there is a run from the configuration $\pair{\mathbf{\alocation}_0}{\vect{0}}$ to the
configuration  $\pair{\mathbf{\alocation}_F}{\vect{x}}$ for some $\vect{x}$. 

In order to prove the above equivalence, we can use the transformations (I) and (II) stated below. 

(I) Let $\arun = \atransition_1 \cdots \atransition_L$ be a run of $\aautomaton$ such that 
for every $i \in \interval{1}{L-1}$, $\atransition_i = \alocation_{i-1} \step{\ainstruction_i} \alocation_{i}$
and the values $c_i^{\achain}$ are defined as before \eqref{eq:c-i-chain}. One can show that 
$\rho' = \atransition'_1 \dotsb \atransition'_L$ with $\atransition'_i = 
\pair{\alocation_{i-1}}{\vect{x}_{i-1}} \step{\ainstruction'_i} \pair{\alocation_i}{\vect{x}_i} $ for every $i$,
is a run of $\aautomaton'$ where
\begin{itemize}
\itemsep 0 cm
\item $\vect{x}_0  = \vect{0}$ 
      and for every $i \in \interval{1}{L}$,  $\vect{x}_i = 
     \tuple{c_i^{1}}{c_i^{n}}$,
\item for every $i \in \interval{1}{L}$, 
      \begin{itemize}
      \itemsep 0 cm
      \item if $\ainstruction_i=\inc{\achain}$ then
            $\ainstruction_i' = \inc{\acounter_{c_i^{\achain}}^{\achain}}$
            (a similar clause holds for decrements),
     \item otherwise (i.e., if $\ainstruction_i$ is not $\inc{\achain}$ nor $\dec{\achain}$ for any $\achain$), $\ainstruction_i' = \mathtt{skip}$. 
      \end{itemize}
\end{itemize} 

(II) Similarly, let $\rho' = \atransition'_1 \dotsb \atransition'_L$ be a run of $\aautomaton'$ where $\atransition'_i = (\alocation_{i-1}, \vect{x}_{i-1}) \step{\ainstruction'_i} (\alocation_i, \vect{x}_i)$ for every $i$, with $\vect{x}_0= \vect{0}$. One can show that $\rho = \atransition_1 \dotsb \atransition_L$ with $\atransition_{i} = \alocation_{i-1} \step{\ainstruction_i} \alocation_{i}$ for every $i$,  is a perfect accepting run of $\aautomaton$ such that for every $i \in \interval{1}{L}$, 
      \begin{itemize}
      \itemsep 0 cm
      \item if $\ainstruction'_i =\inc{\acounter_{j}^{\achain}}$ then $\ainstruction_i =\inc{\achain}$ 
            (a similar clause holds for decrements),
      \item otherwise, if $\vect{x}_{i+1}(\achain) = \vect{x}_i(\achain) + 1$ for some $\achain$, then $\ainstruction_i = \cnext{\achain}$
             (a similar clause holds
            for previous),
     \item otherwise, if 
           $\vect{x}_{i}(\achain) = 0$ and $\alocation_i \step{\cisfirst{\achain}} \alocation_{i+1}
           \in \transitions$ for some $\achain$, then $\ainstruction_i = \cisfirst{\achain}$ (a similar clause holds for $\cisnotfirst{\achain}$),
     \item otherwise, if 
           $\vect{x}_{i}(\achain) =  \mathtt{exp}(k,\amap(\achain)) -1$ and 
           $\alocation_i \step{\cislast{\achain}} \alocation_{i+1}
           \in \transitions$ for some $\achain$, then 
            $\ainstruction_i = \cislast{\achain}$ (a similar clause holds for $\cisnotlast{\achain}$).\qedhere
      \end{itemize} 
\end{proof}

\makeatletter{}\subsection{Hardness Results for Chain Systems}
\label{section-hardness-for-chain-automata}
We show that Per($k$) is $(k+1)$\expspace-hard. The implication is
that replacing direct access to counters by a pointer that can move
along a chain of counters does not decrease the power of VASS, while
providing access to more counters. To demonstrate this, we extend
Lipton's \expspace{}-hardness proof for the control state reachability
problem in VASS~\cite{Lipton76} (see also its exposition in~\cite{Esparza98}). 
Since our pointers can be moved only
one step at a time, this extension involves new insights into the
control flow of algorithms used in Lipton's proof, allowing us to
implement it even with the limitations imposed by step-wise pointer
movements.

\begin{thm}
  \label{thm:CounterAutToPerk}
  Per($k$) is $(k+1)$\expspace-hard.
\end{thm}
\begin{proof}
  We give a reduction from the control state reachability problem for
  the counter automata with zero tests whose counters are bounded by
  $2^{2^{N}}$, where $N = \mathtt{exp}(k, n^{\gamma})$ for some $\gamma \geq 1$ and $n$ is the
  size of counter automaton.  Without any loss of generality, we can
  assume that the initial counter values are zero. The main
  challenge in showing the reduction to Per($k$) is to simulate the
  zero tests of the counter automaton in a chain system. The main
  idea, as in Lipton's proof~\cite{Lipton76}, is to use the fact that
  counters are bounded by $2^{2^{N}}$. For each counter $\acounter$ we keep another counter $\overline{\acounter}$ so that the sum
  of the values of $\acounter$ and $\overline{\acounter}$ is always
  $2^{2^{N}}$. Testing $\acounter$ for zero is then equivalent to
  testing that $\overline{\acounter}$ has the value $2^{2^{N}}$, which can be done by decrementing $\overline{\acounter}$
  $2^{2^{N}}$ times. This involves extending ideas from~\cite{Lipton76}.

  In light of this, the chain
  system will have two chains: $\mathtt{counters}$ and
  $\overline{\mathtt{counters}}$. The value of the counter $\acounter_{j}$ of the
  counter automaton will be maintained in the
  $j$\textsuperscript{th} counter of the chain $\mathtt{counters}$,
  while the value of the complement counter
  $\overline{\acounter_{j}}$ will be maintained in the
  $j$\textsuperscript{th} counter of the chain
  $\overline{\mathtt{counters}}$. In the descriptions that follow,
  whenever we write ``increment $j$\textsuperscript{th} counter of
  $\mathtt{counters}$'', we implicitly assume that the increment
  is followed by the decrement of the $j$\textsuperscript{th} counter of
  $\overline{\mathtt{counters}}$, to maintain the sum at $2^{2^{N}}$.
  The transitions of the chain system are designed to ensure the
  following high-level structure:
  \begin{enumerate}
    \item\label{it:initialize} Begin by setting the values of the first $\counternb$ counters
      of the chain $\overline{\mathtt{counters}}$ to $2^{2^{N}}$, and other initializations that will be needed for simulating zero tests.
      (Here, $\counternb$ is the number of counters in the
      original counter automaton.)
    \item\label{it:simInc} For every transition $\langle \alocation, \acounter_{j}
      \leftarrow \acounter_{j} + 1, \alocation'\rangle$ of the counter
      automaton, the chain system will have transitions corresponding
      to the following listing:
      \begin{enumerate}[label=\({\alph*}]
         \item move the pointers to
                      $j$\textsuperscript{th} counters on
                      $\mathtt{counters}$ and $\overline{\mathtt{counters}}$
                      \begin{description}
                      \item $(\cnext{ \mathtt{counters}};\cnext{
                          \overline{\mathtt{counters}}})^{j};$
                      \end{description}
             (We use `$;$' to compose transitions and `$(\cdot)^j$' to repeat $j$ times a finite sequence of instructions.)
         \item increment the $j$\textsuperscript{th}
                      counter on $\mathtt{counters}$, to simulate
                      incrementing $\acounter_{j}$
                      \begin{description}
                      \item $\inc{ \mathtt{counters}};$
                      \end{description}

         \item decrement the $j$\textsuperscript{th}
                      counter on $\mathtt{\overline{ \mathtt{
                      counters}}}$, to maintain the sum of
                      $j$\textsuperscript{th} counters on $\mathtt{
                      counters}$ and $\overline{ \mathtt{ counters}}$
                      at $2^{2^{N}}$
                      \begin{description}
                      \item $\dec{ \overline{ \mathtt{ counters}}};$
                      \end{description}

         \item move the pointers on
                                  $\mathtt{counters}$ and
                                  $\overline{\mathtt{counters}}$ back
                                  to the first counter and go to
                              $\alocation'$
                              \begin{description}
                              \item $(\cprevious{\mathtt{counters}};\cprevious{
                              \overline{\mathtt{counters}}})^{j},
                              \text{ goto } \alocation'$
                              \end{description}

      \end{enumerate}
                                                                                                                                                                \item\label{it:simZeroTest} For every transition $\langle \alocation: \text{if } \acounter_{j} =
  0 \text{ goto } \alocation_{1} \text{ else } \acounter_{j}
  \leftarrow \acounter_{j} - 1 \text{ goto } \alocation_{2} \rangle$
  of the counter automaton, the chain system will have transitions
  corresponding to the following listing:
  \begin{enumerate}[label=\({\alph*}]
    \item non-deterministically guess that
              $j$\textsuperscript{th} counter is nonzero or zero
              \begin{description}
              \item[$\alocation$] goto nonzero or zero;
              \end{description}
    \item    if the guess is that $j$\textsuperscript{th} counter
          is nonzero; decrement it and goto $\alocation_{2}$
          \begin{description}
          \item[\macstyle{nonzero}] $(\cnext{ \mathtt{counters}}; \cnext{ \overline{
      \mathtt{counters}}})^{j};$ $\dec{
	\mathtt{counters}}; \inc{\overline{\mathtt{counters}}}$;\\ $(\cprevious{
      \mathtt{counters}}; \cprevious{ \overline{
          \mathtt{counters}}})^{j};$ goto $\alocation_{2}$
          \end{description}
      \item otherwise, the guess is that $j$\textsuperscript{th} counter
          is zero, hence the complement of the
          $j$\textsuperscript{th} counter is $2^{2^{N}}$; decrement
          the complement $2^{2^{N}}$ times, increment it back to its
          original value and go to $\alocation_{1}$
          \begin{description}
          \item[\macstyle{zero}] $(\cnext{ \mathtt{counters}}; \cnext{
	\overline{ \mathtt{counters}}})^{j};$ [decrement
	$\overline{\mathtt{counters}}$ $2^{2^{N}}$ times; increment
      $\overline{\mathtt{counters}}$ $2^{2^{N}}$ times];
	$(\cprevious{ \mathtt{counters}}; \cprevious{ \overline{
	  \mathtt{counters}}})^{j};$ goto $\alocation_{1}$.
          \end{description}
  \end{enumerate}
                                                \end{enumerate}
  In \eqref{it:simInc} above, the chain system simply imitates the
  counter automaton, maintaining the condition that $\acounter_{j} +
  \overline{\acounter_{j}} = 2^{2^{N}}$. In \eqref{it:simZeroTest},
  the zero test of the counter automaton is replaced by a
  non-deterministic choice between \macstyle{nonzero} and
  $\macstyle{zero}$. The counter $\acounter_{j}$ itself can be nonzero
  or zero, so there are four possible cases. In the case where the
  nondeterministic choice coincides with the counter value, the chain
  system continues to simulate the counter automaton. On the other
  hand, if \macstyle{nonzero}
  is chosen when $\acounter_{j} = 0$, the chain system gets stuck
  since $(\cnext{ \mathtt{counters}};
  \cnext{ \overline{ \mathtt{counters}}})^{j}; \dec{
  \mathtt{counters}}$ can not be executed (since the
  $j$\textsuperscript{th} counter of $\mathtt{counters}$ has value
  zero and can not be decremented). If \macstyle{zero} is
  chosen when $\acounter_{j} \ne 0$ (and hence
  $\overline{\acounter_{j}} < 2^{2^{N}}$), the chain system also gets
  stuck, since (as we shall show later) the sequence of transitions
  $(\cnext{ \mathtt{counters}}; \cnext{ \overline{
  \mathtt{counters}}})^{j};$ [decrement
  $\overline{\mathtt{counters}}$ $2^{2^{N}}$ times] cannot be
  executed. In the rest of this sub-section, we will show that
  [decrement $\overline{\mathtt{counters}}$ $2^{2^{N}}$ times; increment
  $\overline{\mathtt{counters}}$ $2^{2^{N}}$ times] can be
  implemented in such a way that there is one run that does exactly what is
  required and that any other run deviating from the expected behaviour gets
  stuck. This allows us to conclude that the final state in the given
  counter automaton is reachable if and only if it is reachable in the
  constructed chain system.
\end{proof}

The basic principle for decrementing some counter $2^{2^{N}}$ times is
the same one used in Lipton's proof~\cite{Lipton76}, which we now recall
briefly. We will have two chains of counters $\mathtt{s}$ and
$\overline{\mathtt{s}}$. We will denote the counters in these chains
as $\mathtt{s}_{N}, \mathtt{s}_{N-1}, \ldots, \mathtt{s}_{1}$ and
$\overline{\mathtt{s}_{N}}, \overline{\mathtt{s}_{N-1}}, \ldots,
\overline{\mathtt{s}_{1}}$ respectively. If the counter
$\mathtt{s}_{N}$ has the value $2^{2^{N}}$, we describe how to
decrement it $2^{2^{N}}$ times. Further, if the counter
$\mathtt{s}_{i}$ has the value $2^{2^{i}}$, we describe how to
decrement it $2^{2^{i}}$ times. For this, we will use four more chains
$\mathtt{y}, \overline{\mathtt{y}}, \mathtt{z}$ and
$\overline{\mathtt{z}}$. Our description assumes that for each $i$,
the counters $\mathtt{y}_{i}$ and $\mathtt{z}_{i}$ have the value
$2^{2^{i}}$ (initializing these values will be described later as part
of implementing step \eqref{it:initialize} in the proof of
Theorem~\ref{thm:CounterAutToPerk}). Decrementing
$\mathtt{s}_{i}$ $2^{2^{i}}$ times is done by nested loops as shown in
Figure~\ref{fig:ZeroTestGadget}. The outer loop is indexed by
$\mathtt{y}_{i-1}$ and the inner loop by $\mathtt{z}_{i-1}$. Since both
$\mathtt{y}_{i-1}$ and $\mathtt{z}_{i-1}$ have the value $2^{2^{i-1}}$, the
instruction inside the inner loop (decrementing $\mathtt{s}_{i}$) is
executed $2^{2^{i-1}} \times 2^{2^{i-1}} = 2^{2^{i}}$ times. To
implement these loops, we will need to test when
$\mathtt{y}_{i-1}$ and $\mathtt{z}_{i-1}$ become zero. This is done
the same way as above, replacing $i, i-1$ by $i-1, i-2$ respectively.
\begin{figure}
  \begin{center}
\includegraphics[scale=.9]{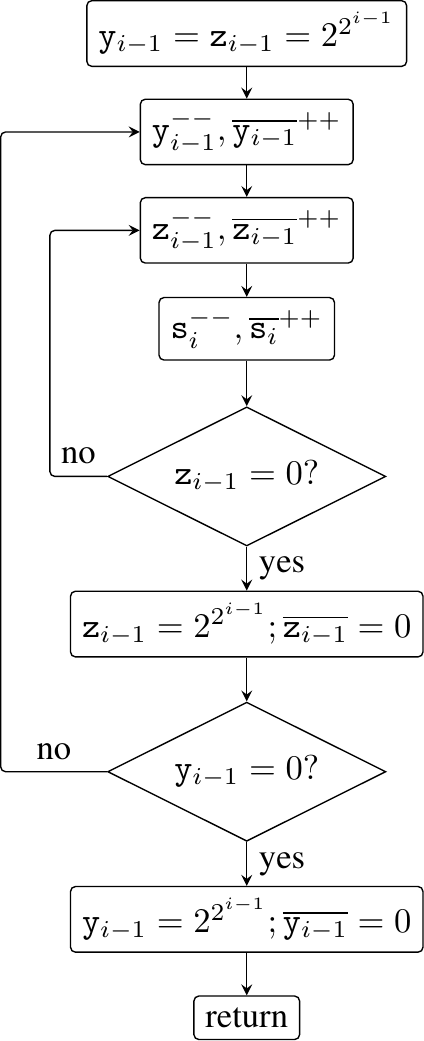}
  \end{center}
  \caption{Control flow of algorithm that decrements
  $\mathtt{s}_{i}$ $2^{2^{i}}$ times.}
  \label{fig:ZeroTestGadget}
\end{figure}
These two zero tests are done recursively by the same set of
transitions. After the recursive zero test, there needs to be a
mechanism to determine if the recursive call was made from the inner
loop or the outer loop. This mechanism is provided by a chain of
counters $\mathtt{stack}$: if the $i$\textsuperscript{th} counter in
the chain $\mathtt{stack}$ has the value $1$ (respectively $0$), then
the recursive call was made by the inner (respectively outer) loop.

We now give the implementation of instructions that follow
the control state named \macstyle{zero} in item~\eqref{it:simZeroTest} of page \pageref{it:simZeroTest} above.
In some intermediate configurations when the implementation is
executing, the four pointers of the four chains $\mathtt{y}$,
$\mathtt{z}$, $\mathtt{s}$ and $\mathtt{stack}$ will all be pointing
to the $i$\textsuperscript{th} counters of their chains. In this case,
we say that the system is \emph{at stage $i$}. The $N$\textsuperscript{th}
stage is the last one. We frequently need to move all four pointers to
the next stage or previous stage, for which we introduce some macros.
The macro \macstyle{Nextstage} is short for the sequence of
transitions $\cnext{\mathtt{y}}; \cnext{\overline{\mathtt{y}}};
\ldots; \cnext{ \mathtt{stack}}; \cnext{\overline{\mathtt{stack}}}$,
which moves all the pointers one stage up. The macro
\macstyle{Prevstage} similarly moves all the pointers one stage down.
The macro \macstyle{transfer($\overline{ \mathtt{z}}$, $\mathtt{s}$)}
is intended to transfer the value of the counter in $\overline{
\mathtt{z}}$ at current stage to the respective counter in $\mathtt{s}$. 
The macro consists of the following transitions.\newpage
\begin{description}
  \item \macstyle{transfer($\overline{ \mathtt{z}}$, $\mathtt{s}$)}
    
  \item \{
      \begin{description}
      \item \macstyle{subtract}: $\inc{\mathtt{z}}; \dec{\overline{ \mathtt{ z}}};
	  \inc{\mathtt{s}}; \dec{\overline{ \mathtt{ s}}};$
          \commstyle{subtract a count from
              $\overline{\mathtt{z}}$ and add it to $\mathtt{s}$}
          \item goto subtract or exit macro
              \commstyle{non-deterministically choose to repeat another
              step of the count transfer or to exit the macro}
      \end{description}
  \item \}
\end{description}\medskip

\noindent The macro \macstyle{transfer($\overline{ \mathtt{y}}$,
$\mathtt{s}$)} is similar to the above, with $\mathtt{y}$ and
$\overline{\mathtt{y}}$ replacing $\mathtt{z}$ and
$\overline{\mathtt{z}}$ respectively. The macro
\macstyle{inc(next($\mathtt{z}$))} moves the pointer of the chain
$\mathtt{z}$ one stage up, increments $\mathtt{z}$ once and moves the
pointer back one stage down: $\cnext{\mathtt{z}};
\inc{\mathtt{z}};
\cprevious{\mathtt{z}}$. The macros \macstyle{inc(next($\mathtt{y}$))}
and \macstyle{inc(next($\overline{\mathtt{s}}$))} are similar to
\macstyle{inc(next($\mathtt{z}$))} with $\mathtt{y}$ and $\overline{ \mathtt{
s}}$ replacing $\mathtt{z}$ respectively. The macro
\macstyle{dec(next($\mathtt{s}$))} is similarly defined to decrement
the next counter in the chain $\mathtt{s}$. 
The macro \macstyle{$\mathtt{stack}$ ==
1} tests if the counter in the current stage in the chain
$\mathtt{stack}$ is greater than $0$: $\dec{\mathtt{stack}};
\inc{\mathtt{stack}}$.

Following is the set of transitions at the control state
\macstyle{zero} in \eqref{it:simZeroTest} in the proof of
Theorem~\ref{thm:CounterAutToPerk} above. The listing below is written
in the form of a program in a low level programming language for
readability. It can be easily translated into a set of transitions of
a chain system. There is a control state between every two consecutive
instructions below, but only the important ones are given names like
``outerloop2''. A line such as ``innernonzero2: $\dec{ \mathtt{z}}$;
$\inc{ \mathtt{z}}$; goto innerloop2'' actually represents the set of
transitions $\set{ \langle \text{innernonzero2}, \dec{ \mathtt{z}},
\alocation\rangle, \langle\alocation, \inc{ \mathtt{ z}},
\text{innerloop2}\rangle}$. The instruction ``\alocation: if
$\cisfirst{\achain}$ then \{program 1\} else \{program 2\}''
represents the set of transitions \enlargethispage{2\baselineskip}
$$\set{ \langle\alocation, \cisfirst{
\achain}, \alocation_{1}\rangle, \langle \alocation, \cisnotfirst{
\achain}, \alocation_{2}\rangle} \cup$$ 
$$\set{ \text{transitions for
program 1 from } \alocation_{1}} \cup \set{ \text{transitions for
program 2 from } \alocation_{2}}.$$ 
An instruction of the form
``$\alocation: \inc{ \overline{ \mathtt{stack}}}; \text{goto
outernonzero2 or outerzero2}$'' represents the set of transitions
$\set{ \langle\alocation, \inc{ \overline{ \mathtt{stack}}},
\text{outernonzero2}\rangle, \langle\alocation, \inc{ \overline{
\mathtt{stack}}}, \text{outerzero2}\rangle}$. Depending on the
non-deterministic choices made at control states that have multiple
transitions enabled, there will be several different runs. We prove
later that there is one run which has the intended effect and the
other runs will never reach the final state.

The action [decrement $\overline{\mathtt{counters}}$ $2^{2^{N}}$
times] in item \eqref{it:simZeroTest} of page \pageref{it:simZeroTest}
is implemented in Chain~system~\ref{csm:dec} by ``transfer($\overline{ \mathtt{ counters}}$,
$\mathtt{s}$); goto $\text{Dec}_{\macstyle{zerorep}}$'' in the first
line. There is a non-deterministic choice for exiting ``transfer($\overline{ \mathtt{ counters}}$,
$\mathtt{s}$)'', and we wish to block any run that
exits before incrementing $\mathtt{s}$ $2^{2^{N}}$ times. This is done
by ``goto $\text{Dec}_{\macstyle{zerorep}}$'', which
forces the chain system to decrement $\mathtt{s}$ $2^{2^{N}}$ times.

\begin{chainsystem}
\begin{description}
\item \macstyle{zero}: $(\cnext{ \mathtt{counters}}; \cnext{
  \overline{ \mathtt{counters}}})^{j};$
  transfer($\overline{ \mathtt{ counters}}$,
  $\mathtt{s}$); goto $\text{Dec}_{\macstyle{zerorep}}$
\item \macstyle{zerorep}: transfer($\mathtt{ counters}$,
  $\mathtt{s}$); goto $\text{Dec}_{\macstyle{zeropass}}$
\item \macstyle{zeropass}: $(\cprevious{ \mathtt{counters}}; \cprevious{
  \overline{ \mathtt{counters}}})^{j};$ goto $\alocation_{1}$
\item \macstyle{$\text{Dec}_{\langle\text{address}\rangle}$}:
  \item if $\cisfirst{\mathtt{stack}}$ then
    \begin{description}
      \item $(\dec{\mathtt{s}})^{4};
	(\inc{\overline{\mathtt{s}}})^{4}$; goto
	DecFinished
    \end{description}
  \item else \mbox{}
    \begin{description}
      \item Prevstage
      \item \macstyle{outerloop2}: $\dec{\mathtt{y}};
	\inc{\overline{\mathtt{y}}}$ \commstyle{$\mathtt{y}$ is the
	index for outer loop}
	\begin{description}
	    \item \macstyle{innerloop2}: $\dec{\mathtt{z}};
	      \inc{\overline{\mathtt{z}}}$ \commstyle{$\mathtt{z}$ is the
	      index for inner loop}
	      \begin{description}
		\item dec(next($\mathtt{s}$)); inc(next($\overline{
		  \mathtt{ s}}$))
		\item \macstyle{innertest2}: goto innernonzero2 or innerzero2
		\item \macstyle{innernonzero2}: $\dec{\mathtt{z}}; \inc{\mathtt{z}}$; goto
		  innerloop2 \commstyle{inner loop not yet complete}
	      \end{description}
	    \item \macstyle{innerzero2}: transfer($\overline{\mathtt{z}}$,
	      $\mathtt{s}$); $\inc{\mathtt{stack}}; \dec{
	      \overline{ \mathtt{ stack}}};$ goto $\text{Dec}_{\langle\text{address}\rangle}$ \commstyle{inner loop
                  complete. $(i-1)$\textsuperscript{th} counter of
                  $\mathtt{stack}$ is set to 1, so the recursive call
                  to $\text{Dec}_{\langle\text{address}\rangle}$
              returns to outertest2}
	    \item \macstyle{outertest2}: $\dec{\mathtt{stack}}; \inc{
	      \overline{ \mathtt{stack}}};$ goto outernonzero2 or outerzero2
	    \item \macstyle{outernonzero2}: $\dec{\mathtt{y}}; \inc{\mathtt{y}}$; goto
	      outerloop2 \commstyle{outer loop not yet complete}
	\end{description}
      \item \macstyle{outerzero2}: transfer($\overline{\mathtt{y}}$,
	$\mathtt{s}$); goto $\text{Dec}_{\langle\text{address}\rangle}$ \commstyle{outer loop
	complete. $(i-1)$\textsuperscript{th} counter of
                  $\mathtt{stack}$ is set to 0, so the recursive call
              to $\text{Dec}_{\langle\text{address}\rangle}$ returns
          to outerexit2}
      \item \macstyle{outerexit2}: Nextstage;  goto
	DecFinished
    \end{description}
    \item fi
    \item \macstyle{DecFinished}: goto backtrack
      \item
      \item \macstyle{backtrack}:
      \item if ($\cisnotlast{\mathtt{stack}}$) then
          \commstyle{is the pointer at some intermediate stage, which
          means we just completed a recursive call?}
	\begin{description}
	  \item if ($\mathtt{stack}$ ==  1) then goto outertest2
              \commstyle{$i$\textsuperscript{th} counter of
              $\mathtt{stack}$ is set to $1$, so return to outertest2}
	  \item else goto outerexit2 \commstyle{$i$\textsuperscript{th} counter of
              $\mathtt{stack}$ is set to $0$, so return to outerexit2}
	\end{description}
    \item else \commstyle{pointer is at the last stage, decrementation
        is complete}
	\begin{description}
	  \item goto $\langle$address$\rangle$
	\end{description}
      \item fi
\end{description}
    \caption{Decrement $j$\textsuperscript{th} counter in $\overline{
        \mathtt{counters}}$ $2^{2^{N}}$ times.}
        \label{csm:dec}
\end{chainsystem}
Suppose, as in the proof of Theorem~\ref{thm:CounterAutToPerk}, we
want to simulate the $\acounter_{j} = 0$ case of the counter
automaton's transition $\langle \alocation: \text{if } \acounter_{j} =
0 \text{ goto } \alocation_{1} \text{ else } \acounter_{j}
\leftarrow \acounter_{j} - 1 \text{ goto } \alocation_{2} \rangle$.
The action [increment $\overline{\mathtt{counters}}$ $2^{2^{N}}$
times] in the proof of Theorem~\ref{thm:CounterAutToPerk} is
implemented above by ``transfer($\mathtt{counters}$, $\mathtt{s}$);
goto $\text{Dec}_{\macstyle{zeropass}}$'' in the second line. If
$\text{Dec}_{\macstyle{zeropass}}$ returns successfully, the control
can go to $\alocation_{1}$ as required. The
only difference between $\text{Dec}_{\macstyle{zerorep}}$ and
$\text{Dec}_{\macstyle{zeropass}}$ is the return address
\macstyle{zerorep} and \macstyle{zeropass}. A generic
$\text{Dec}_{\langle\text{address}\rangle}$ is shown above, which
implements the algorithm shown in Figure~\ref{fig:ZeroTestGadget}.
Formally, every instruction label (such as \macstyle{outerloop2},
\macstyle{innerloop2} etc.) should also be parameterized with
\macstyle{$\langle \text{address} \rangle$}, but these have been ommitted for
the sake of readability.
The case ``if $\cisfirst{\mathtt{stack}}$'' implements the base case
$i=1$. If $i >1$, the else branch is taken, where the first
instruction is to move all pointers one stage down, so that they now
point to $\mathtt{y}_{i-1}$ and $\mathtt{z}_{i-1}$ as required by
the algorithm of Figure~\ref{fig:ZeroTestGadget}. The loop between
\macstyle{outerloop2} and \macstyle{outerzero2} implements the outer
loop of Figure~\ref{fig:ZeroTestGadget} and the loop between
\macstyle{innerloop2} and \macstyle{innerzero2} implements the inner
loop of Figure~\ref{fig:ZeroTestGadget}. The
non-deterministic choice in \macstyle{innertest2} decides whether to
continue the inner loop or to exit. The macro transfer($\overline{\mathtt{z}}$,
$\mathtt{s}$) at the beginning of \macstyle{innerzero2} will reset the
$(i-1)$\textsuperscript{th} counter of $\mathtt{z}$ back to
$2^{2^{i-1}}$, at the same time setting
$(i-1)$\textsuperscript{th} counter of $\mathtt{s}$ to
$2^{2^{i-1}}$. So the run that behaves as
expected increments $\overline{\mathtt{z}_{i-1}}$ exactly
$2^{2^{i-1}}$ times. All other runs are blocked by the recursive call
to $\text{Dec}_{\langle\text{address}\rangle}$ at the end of
\macstyle{innerzero2}. The purpose of $\inc{\mathtt{stack}}; \dec{
\overline{ \mathtt{ stack}}};$ in the middle of \macstyle{innerzero2}
is to ensure that the recursive call to
$\text{Dec}_{\langle\text{address}\rangle}$ returns to the correct
state. After similar tests in \macstyle{outerzero2}, ``Nextstage'' at
the beginning of \macstyle{outerexit2} moves all the pointers back to
stage $i$. Then in \macstyle{backtrack}, the correct return state is
figured out. Exactly how this is done will be clear in the proof of
the following lemma, where we construct runs of the chain system by
induction on stage.

\begin{lem}
  \label{lem:SimZeroTest}
  In the state \macstyle{zero} in the listing above, if the
  $j$\textsuperscript{th} counter in the chain $\mathtt{counters}$ has
  the value $0$, there is a run that reaches the state $q_1$ without
  causing any changes. Any run from the state \macstyle{zero} will be
  blocked if the $j$\textsuperscript{th} counter in the chain
  $\mathtt{counters}$ has a value other than $0$.
\end{lem}
\begin{proof}
  We first prove the existence of a run reaching $q_{1}$ when the
  $j$\textsuperscript{th} counter in the chain $\mathtt{counters}$ has
  the value $0$. The $j$\textsuperscript{th} counter in the chain
  $\overline{\mathtt{counters}}$ will have the value $2^{2^{N}}$. Our
  run executes transfer($\overline{ \mathtt{ counters}}$,
  $\mathtt{s}$) to set $\mathtt{counters}_{j}$ and $\mathtt{s}_{N}$ to
  $2^{2^{N}}$ and set $\overline{\mathtt{s}_{N}}$ and
  $\overline{\mathtt{counters}_{j}}$ to $0$. With these values in the
  counters, we will prove next that there is run from
  $\text{Dec}_{\macstyle{zerorep}}$ to \macstyle{zerorep} that sets
  $\mathtt{s}_{N}$ to $0$ and $\overline{\mathtt{s}_{N}}$ to
  $2^{2^{N}}$. From \macstyle{zerorep}, our run continues to execute
  transfer($\mathtt{ counters}$, $\mathtt{s}$) to set
  $\mathtt{counters}_{j}$ and $\overline{\mathtt{s}_{N}}$ to
  $0$ and set $\overline{\mathtt{counters}_{j}}$ and
  $\mathtt{s}_{N}$ to $2^{2^{N}}$. Then there is again a run from
  $\text{Dec}_{\macstyle{zeropass}}$ to \macstyle{zeropass} which sets
  $\mathtt{s}_{N}$ to $0$ and sets $\overline{\mathtt{s}_{N}}$ to
  $2^{2^{N}}$. From \macstyle{zeropass}, our run moves the pointers on
  the chains $\mathtt{counters}$ and $\overline{\mathtt{counters}}$
  back to the first position and goes to state $q_1$.

  Now suppose that in the state
  $\text{Dec}_{\langle\text{address}\rangle}$, $\mathtt{s}_{N}$ is set
  to $2^{2^{N}}$, $\overline{\mathtt{s}_{N}}$ is set to $0$ and the
  pointer in the chain $\mathtt{stack}$ is at the last position, as
  mentioned in the previous paragraph. We will prove that there is a
  run that reaches $\langle\text{address}\rangle$, setting
  $\mathtt{s}_{N}$ to $0$ and $\overline{\mathtt{s}_{N}}$ to
  $2^{2^{N}}$. We will in fact prove the following claim by induction
  on $i$.

  \emph{Claim:} Suppose that in the state
  $\text{Dec}_{\langle\text{address}\rangle}$, $\mathtt{s}_{i}$ is set
  to $2^{2^{i}}$, $\overline{\mathtt{s}_{i}}$ is set to $0$ and the
  pointers in the chains $\mathtt{stack}, \mathtt{s},
  \mathtt{y}, \mathtt{z}$ are at position $i$. There is a
  run $\arun_{i}$ that reaches \macstyle{DecFinished}, setting
  $\mathtt{s}_{i}$ to $0$ and $\overline{\mathtt{s}_{i}}$ to
  $2^{2^{i}}$, with the pointers in the same position.

  When we first call $\text{Dec}_{\macstyle{zerorep}}$ or
  $\text{Dec}_{\macstyle{zeropass}}$, the pointers in the chains
  $\mathtt{stack}, \mathtt{s}, \mathtt{y}, \mathtt{z}$ are at the
  last position ($N$). Hence, when the run $\arun_{N}$ given by the above
  claim reaches \macstyle{DecFinished}, it can continue to the state
  \macstyle{backtrack} and then to \macstyle{zerorep} or
  \macstyle{zeropass} respectively.

  \emph{Base case of the claim:} $i=1$. The run $\arun_{1}$ is as
  follows: if $\cisfirst{\mathtt{stack}}$ then
  $(\dec{\mathtt{s}})^{4}; (\inc{\overline{\mathtt{s}}})^{4}$; goto
  DecFinished .

  \emph{Induction step of the claim:}. The run $\arun_{i+1}$ is as
  follows.
    \begin{description}
      \item Prevstage
      \item \macstyle{outerloop2}: $\dec{\mathtt{y}};
	\inc{\overline{\mathtt{y}}}$
	\begin{description}
	    \item \macstyle{innerloop2}: $\dec{\mathtt{z}};
	      \inc{\overline{\mathtt{z}}}$
	      \begin{description}
		\item dec(next($\mathtt{s}$)); inc(next($\overline{
		  \mathtt{ s}}$))
		\item \macstyle{innertest2}: goto innernonzero2
		  $2^{2^{i}} - 1$ times, then goto innerzero2
		\item \macstyle{innernonzero2}: $\dec{\mathtt{z}}; \inc{\mathtt{z}}$; goto
		  innerloop2
	      \end{description}
	    \item \macstyle{innerzero2}: transfer($\overline{\mathtt{z}}$,
	      $\mathtt{s}$); $\inc{\mathtt{stack}}; \dec{
	      \overline{ \mathtt{ stack}}};$ goto
	      $\text{Dec}_{\langle\text{address}\rangle}$
	      \commstyle{Follow $\arun_{i}$ here. Since we set
	      $\mathtt{stack}_{i}$ to $1$, we can return to outertest2
	      at the end of $\arun_{i}$ to continue our
	      $\arun_{i+1}$}
	    \item \macstyle{outertest2}: $\dec{\mathtt{stack}}; \inc{
	      \overline{ \mathtt{stack}}};$ goto outernonzero2
	      $2^{2^{i}} - 1$ times, then goto outerzero2
	    \item \macstyle{outernonzero2}: $\dec{\mathtt{y}}; \inc{\mathtt{y}}$; goto
	      outerloop2
	\end{description}
      \item \macstyle{outerzero2}: transfer($\overline{\mathtt{y}}$,
	$\mathtt{s}$); goto
	$\text{Dec}_{\langle\text{address}\rangle}$ \commstyle{Follow
	  $\arun_{i}$ here. Since now $\mathtt{stack}_{i}$ is
	  $0$, we can return to outerexit2 at the end of
	  $\arun_{i}$ to continue our $\arun_{i+1}$}
      \item \macstyle{outerexit2}: Nextstage;  goto
	DecFinished
    \end{description}
    This completes the induction step and the proof of the claim.

    Finally, we will prove that any run from \macstyle{zero} will be
    blocked if the $j$\textsuperscript{th} counter in the chain
    $\mathtt{counters}$ has a value other than $0$. Indeed, the
    $j$\textsuperscript{th} counter in the chain
    $\overline{\mathtt{counters}}$ has a value less than $2^{2^{N}}$
    when the state is $\text{Dec}_{\macstyle{zerorep}}$.
    Hence no run can execute \macstyle{innerloop2} and
    \macstyle{outerloop2} $2^{2^{N-1}}$ times. Hence, any run will
    either get blocked when the pointer in the chain $\mathtt{s}$ is
    still at position $N$, or it will go to state
    $\text{Dec}_{\langle\text{address}\rangle}$ with
    $\mathtt{s}_{N-1}$ less than $2^{2^{N-1}}$. In the latter case, we
    can argue in the same way to conclude that any run is blocked
    either when the pointer in the chain $\mathtt{s}$ is still at
    $N-1$ or it will go to state
    $\text{Dec}_{\langle\text{address}\rangle}$ with
    $\mathtt{s}_{N-2}$ less than $2^{2^{N-2}}$ and so on. Any run is
    blocked at some stage between $N$ and $1$.
\end{proof}

Next, we explain the implementation of step \eqref{it:initialize} in
the proof of Theorem~\ref{thm:CounterAutToPerk}. We show how to get
counters to their required initial values (at the beginning of the
runs all the values are equal to zero). We briefly recall the required
initial values:
\begin{enumerate}
\itemsep 0 cm
\item each counter $\acounter$ has value zero,
\item for every $i \in \interval{1}{N}$, 
      $\overline{\mathtt{y}}_i$, $\overline{\mathtt{z}}_i$ and  $\mathtt{s}_i$ are equal to zero,
\item for every $i \in \interval{1}{N}$, $\mathtt{stack}_i$ is equal to zero and 
      $\overline{\mathtt{stack}}_i$ is equal to one, 
\item each complement counter $\overline{\acounter}$ has value $2^{2^N}$,
\item for every $i \in \interval{1}{N}$, 
      $\mathtt{y}_i$, $\mathtt{z}_i$ and  $\overline{\mathtt{s}}_i$ have the value $2^{2^i}$.
\end{enumerate} 
The initialization for the points (1)--(3) is easy to perform and below we
focus on the initialization for the point (5). Initialization of the
complement counter in point (4) will be dealt with later by simply
adjusting what is done below.  To help achieve these initial values
for 5), we have another chain $\mathtt{init}$, with $N$ counters. We
follow the convention that if at stage $i$, the counter in the chain
$\mathtt{init}$ has value $1$, then all the counters in all the chains
at stage $i$ or below have been properly initialized.  By convention,
the condition $\cislast{\mathtt{init}}$ is true if the pointer in the
chain $\mathtt{init}$ is at stage $N$. The macros \macstyle{Nextstage}
and \macstyle{Prevstage} moves the pointers of the chain
$\mathtt{init}$, in addition to all the other chains.

We now give the code used to initialize
$\mathtt{y}$, $\mathtt{z}$, $\mathtt{s}$ and $\mathtt{stack}$ to the
required values, assuming that all counters are initially set to $0$
and all the pointers in all the chains are pointing to stage $0$. The
counters $\mathtt{y}_{0}$, $\mathtt{z}_{0}, \mathtt{s}_{0}$ and
$\mathtt{stack}_{0}$ are initialized directly. For the counters at
higher stages, we again use nested loops similar to those shown in
Figure~\ref{fig:ZeroTestGadget}, assuming that counters at lower
stages have been initialized. These nested loops are implemented
between \macstyle{innerinit}--\macstyle{innerzero1} and
\macstyle{outerinit}--\macstyle{outerzero1}. The zero tests involved
in terminating these loops are performed by the same decrementing
algorithm, assuming that counters at lower stages are already
initialized. This will introduce some complications while backtracking
from the decrementing algorithm --- we have to figure out whether the
call to \macstyle{Dec} was from \macstyle{innerzero1} or
\macstyle{outerzero1}, or from within \macstyle{Dec} itself from a
recursive call. To handle this, the ``backtrack'' part of the code
below has been updated to check if (next($\mathtt{init}$) ==  1): if
so, we are inside some recursive call to \macstyle{Dec}. Otherwise,
the call to \macstyle{Dec} was from one of the loops initializing a
higher level.

\begin{description}
  \item \macstyle{begininit}:
    $(\inc{\mathtt{y}})^{4};(\inc{\mathtt{z}})^{4};
    (\inc{\overline{\mathtt{s}}})^{4};
    \inc{\mathtt{init}}; \inc{\overline{\mathtt{stack}}}$
  \item \macstyle{initialize}: If $\cislast{\mathtt{init}}$ goto beginsim else
    goto outerinit
  \item \macstyle{outerinit}: $\dec{\mathtt{y}}; 
    \inc{\overline{\mathtt{y}}}$ \commstyle{$\mathtt{y}$ is the index for
    outer loop}
    \begin{description}
    \item \macstyle{innerinit}: $\dec{\mathtt{z}};
      \inc{\overline{\mathtt{z}}}$ \commstyle{$\mathtt{z}$ is the index for
      inner loop}
      \begin{description}
	\item  \macstyle{INC}: inc(next($\mathtt{y}$)); inc(next($\mathtt{z}$));
	  inc(next($\overline{\mathtt{s}}$))
	\item \macstyle{innertest1}: goto innernonzero1 or innerzero1
	\item \macstyle{innernonzero1}: $\dec{\mathtt{z}}; \inc{\mathtt{z}};$ goto
	  innerinit \commstyle{inner loop not yet complete}
      \end{description}
    \item \macstyle{innerzero1:} transfer($\overline{\mathtt{z}}$,
      $\mathtt{s}$); $\inc{\mathtt{stack}};
      \dec{\overline{\mathtt{stack}}};$ goto Dec \commstyle{inner loop
      complete. (next($\mathtt{init}$) ==  0) and ($\mathtt{stack}$ ==
  1), so Dec returns to outertest1}
    \item \macstyle{outertest1}: $\dec{\mathtt{stack}};
      \inc{\overline{\mathtt{stack}}};$ goto outernonzero1 or outerzero1
      \item \macstyle{outernonzero1}: $\dec{\mathtt{y}}; \inc{\mathtt{y}};$ goto
	outerinit \commstyle{outer loop not yet complete}
    \end{description}
  \item \macstyle{outerzero1}: transfer($\overline{\mathtt{y}}$,
    $\mathtt{s}$); goto Dec \commstyle{outer loop
    complete. (next($\mathtt{init}$) ==  0) and ($\mathtt{stack}$ ==
  0), so Dec returns to outerexit1}
  \item \macstyle{outerexit1}: Nextstage; $\inc{\mathtt{init}};
    \inc{\overline{\mathtt{stack}}};$ goto initialise
  \item \macstyle{beginsim}: start simulating the counter machine (see
    part of the proof of Theorem~\ref{thm:CounterAutToPerk} related to
    the simulation) 
  \item
  \item \macstyle{Dec}:
  \item if $\cisfirst{\mathtt{init}}$ then
    \begin{description}
      \item $(\dec{\mathtt{s}})^{4};
	(\inc{\overline{\mathtt{s}}})^{4}$; goto
	DecFinished
    \end{description}
  \item else \mbox{}
    \begin{description}
      \item Prevstage
      \item \macstyle{outerloop2}: $\dec{\mathtt{y}};
	\inc{\overline{\mathtt{y}}}$ \commstyle{$\mathtt{y}$ is the
	index for outer loop}
	\begin{description}
	  \item \macstyle{innerloop2}: $\dec{\mathtt{z}};
	    \inc{\overline{\mathtt{z}}}$ \commstyle{$\mathtt{z}$ is the
	    index for inner loop}
	    \begin{description}
	      \item dec(next($\mathtt{s}$)); inc(next($\overline{
		\mathtt{ s}}$))
	      \item \macstyle{innertest2}: goto innernonzero2 or innerzero2
	      \item \macstyle{innernonzero2}: $\dec{\mathtt{z}}; \inc{\mathtt{z}}$; goto
		innerloop2 \commstyle{inner loop not yet complete}
	    \end{description}
	  \item \macstyle{innerzero2}: transfer($\overline{\mathtt{z}}$,
	    $\mathtt{s}$); $\inc{\mathtt{stack}}; \dec{
	    \overline{ \mathtt{ stack}}};$ goto Dec \commstyle{inner loop
	    complete. (next($\mathtt{init}$) ==  1) and ($\mathtt{stack}$ ==
  1), so Dec returns to outertest2}
	  \item \macstyle{outertest2}: $\dec{\mathtt{stack}}; \inc{
	    \overline{ \mathtt{stack}}};$ goto outernonzero2 or outerzero2
	  \item \macstyle{outernonzero2}: $\dec{\mathtt{y}}; \inc{\mathtt{y}}$; goto
	    outerloop2 \commstyle{outer loop not yet complete}
	\end{description}
      \item \macstyle{outerzero2}: transfer($\overline{\mathtt{y}}$,
	$\mathtt{s}$); goto Dec \commstyle{outer loop
	complete. (next($\mathtt{init}$) ==  1) and ($\mathtt{stack}$ ==
  0), so Dec returns to outerexit2}
      \item \macstyle{outerexit2}: Nextstage;  goto
	DecFinished
    \end{description}
    \item fi
    \item \macstyle{DecFinished}: goto backtrack
    \item \macstyle{backtrack}:
      \item if (next($\mathtt{init}$) ==  1) then \commstyle{we are not at the
	end of recursion}
	\begin{description}
	  \item if ($\mathtt{stack}$ ==  1) then goto outertest2
	  \item else goto outerexit2
	\end{description}
      \item else \commstyle{we are at the end of recursion}
	\begin{description}
	  \item if ($\mathtt{stack}$ ==  1) then goto outertest1
	  \item else goto outerexit1
	\end{description}
      \item fi
\end{description}
The ideas involved in the above listing are similar to the ones in the
code decrementing $\mathtt{counters}_{j}$ $2^{2^{N}}$ times.

\begin{lem}
  \label{lem:InitialisationFlow}
  Suppose the control state is \macstyle{begininit}, all the pointers
  in all the chains are at stage $1$ and all the counters at all
  stages have the value $0$. For any $i$ between $1$ and $N$, there is
  a run $\arun_{i}'$ ending at \macstyle{initialise}, such that the
  pointer is at stage $i$  in all the chains and all the counters at
  or below stage $i$ have been initialised. If a run starts from
  \macstyle{begininit} as above, it will not reach \macstyle{beginsim}
  unless all the counters have been properly initialized.
\end{lem}
\begin{proof}
  By induction on $i$. For the base case $i = 1$, $\arun_{1}'$ is the
  run that executes the instructions immediately after
  \macstyle{begininit}  and ends at \macstyle{initialise}.

  Now we assume the lemma is true up to $i$ and prove it for $i+1$. By
  induction hypothesis, there is a run $\arun_{i}'$ ending at
  \macstyle{initialise}, with all the pointers in all the chains at stage $i$
  and all the counters at or below stage $i$ initialised. The run
  $\arun_{i+1}'$ is obtained by appending the following sequence to
  $\arun_{i}'$. It uses the runs $\arun_{i}$ constructed in the proof
  of Lemma~\ref{lem:SimZeroTest}.
  \begin{description}
    \item \macstyle{initialise}: If $\cislast{\mathtt{init}}$ goto
      beginsim else goto outerinit \commstyle{pointer in chain
	$\mathtt{init}$ is at $i \ne N$; go to \macstyle{outerinit}}
    \item \macstyle{outerinit}: $\dec{\mathtt{y}}; 
      \inc{\overline{\mathtt{y}}}$ 
      \begin{description}
      \item \macstyle{innerinit}: $\dec{\mathtt{z}};
	\inc{\overline{\mathtt{z}}}$ 
	\begin{description}
	  \item  \macstyle{INC}: inc(next($\mathtt{y}$)); inc(next($\mathtt{z}$));
	    inc(next($\overline{\mathtt{s}}$))
	  \item \macstyle{innertest1}: goto innernonzero1
	    $2^{2^{i}} - 1$ times then goto innerzero1
	  \item \macstyle{innernonzero1}: $\dec{\mathtt{z}}; \inc{\mathtt{z}};$ goto
	    innerinit \commstyle{inner loop not yet complete}
	\end{description}
      \item \macstyle{innerzero1:} transfer($\overline{\mathtt{z}}$,
	$\mathtt{s}$); $\inc{\mathtt{stack}};
	\dec{\overline{\mathtt{stack}}};$ goto Dec \commstyle{follow
	  $\arun_{i}$ here; since (next($\mathtt{init}$) ==  0) and
	  $\mathtt{stack}$ ==  1, we can retrun to
	  \macstyle{outertest1} to continue our $\arun_{i+1}'$}
      \item \macstyle{outertest1}: $\dec{\mathtt{stack}};
	\inc{\overline{\mathtt{stack}}};$ goto outernonzero1
	$2^{2^{i}} - 1$ times then goto outerzero1
	\item \macstyle{outernonzero1}: $\dec{\mathtt{y}}; \inc{\mathtt{y}};$ goto
	  outerinit
      \end{description}
    \item \macstyle{outerzero1}: transfer($\overline{\mathtt{y}}$,
      $\mathtt{s}$); goto Dec \commstyle{follow $\arun_{i}$ here;
    since (next($\mathtt{init}$) ==  0) and
	  $\mathtt{stack}$ ==  0, we can retrun to
	  \macstyle{outerexit1} to continue our $\arun_{i+1}'$}
    \item \macstyle{outerexit1}: Nextstage; $\inc{\mathtt{init}};
      \inc{\overline{\mathtt{stack}}};$ goto initialise
    \end{description}
    This completes the induction step, proving the existence of a run
    ending at \macstyle{initialize} as required.

    Finally we argue that any run from \macstyle{begininit} that does
    not initialize all the counters properly will get stuck. Indeed,
    the only way out of the initialization code is to go to
    \macstyle{beginsim}, which is only reachable via
    \macstyle{initialize}. From \macstyle{initialize}, we can go to
    \macstyle{beginsim} only when the pointers are at stage $N$. We
    finish the argument by showing that any run that visits
    \macstyle{initialize} for the first time with pointers at stage
    $i$ would have initialized all counters at or below stage
    $i$. The only way to go to \macstyle{initialize} with pointers at
    stage $i$ is via \macstyle{begininit} (for $i = 1$) or via
    \macstyle{outerexit1} (for $i > 1$). It is clear in the case of
    $i=1$ that the counters at stage $1$ are properly initialized. In
    the case of $i > 1$, the only way to reach \macstyle{outerexit1}
    is via \macstyle{Dec}. In this case, all runs are forced to
    iterate the loops at \macstyle{outerinit} and \macstyle{innerinit}
    the correct number of times as we have seen in the proof of
    Lemma~\ref{lem:SimZeroTest}, which will again force all the
    counters at or below stage $i$ to be properly initialized.
\end{proof}

To complete the proof of Theorem~\ref{thm:CounterAutToPerk}, we
finally show how to initialize the first $\counternb \le n$ counters from
$\mathtt{counters}$ and $\overline{\mathtt{counters}}$ that
correspond to the counters of the original counter automaton (point (4)
of the initialization values listed after
Lemma~\ref{lem:SimZeroTest}). This is achieved by replacing the code for \macstyle{INC} from the code for the initialization phase by the following one:
\begin{description}
\item  \macstyle{INC}: inc(next($\mathtt{y}$)); inc(next($\mathtt{z}$)); inc(next($\overline{\mathtt{s}}$))
\item  $\cnext{\mathtt{init}}$ \\
       if $\cislast{\mathtt{init}}$ then
          \begin{description}
          \item $\inc{\overline{\mathtt{counters}}}$
          \item
	    ($\cnext{\overline{\mathtt{counters}}}$;$\inc{\overline{\mathtt{counters}}}$)$^{\counternb-1}$
          \item ($\cprevious{\overline{\mathtt{counters}}}$)$^{\counternb-1}$
          \end{description}
       $\cprevious{\mathtt{init}}$
\end{description}
This finishes the proof of Theorem~\ref{thm:CounterAutToPerk}.

\makeatletter{}\subsection{Reasoning About Chain Systems with $\mainlogic$}
\label{section-lower-bound}
Given a chain system $\aautomaton$ of level 1, we construct a polynomial-size formula in $\mainlogic$ that
is satisfiable iff  $\aautomaton$ has a gainy accepting run.
Hence, we get a \twoexpspace{} lower bound for
\SAT{$\mainlogic$}. The main idea  is to encode runs of chain
systems with $\mainlogic$ formulas that access the counters using
binary encoding, so that the formulas can handle exponentially many
counters.
\begin{lem}
  \label{lem:GainyToLRVSat}
  There is a polynomial-time reduction from Gainy($1$) into
  \SAT{$\mainlogic(\mynext,\sometimes)$}.
\end{lem}
\makeatletter{}\ifLONG
\begin{proof}
  Let $\aautomaton =
  \triple{\locations,\amap,1}{\locations_0,\locations_F}{\transitions}$
  be a chain system of level $1$ with $f :
  \interval{1}{n} \to \Nat$, having thus $n$ chains of counters, of
  respective size $2^{f(1)}, \dotsc, 2^{f(n)}$.

We encode a word $\rho \in \transitions^*$ that
  represents an accepting run. For this, we use the alphabet
  $\transitions$ of transitions. Note that we can easily simulate the labels $\delta=\set{t_1, \dotsc, t_m}$ with the variables $\mathtt{t}_0, \dotsc, \mathtt{t}_m$, where a node has an encoding of the label $t_i$  if{f} the formula $\tup{t_i} \egdef \mathtt{t}_0 \approx \mathtt{t}_i$ holds true. We build a
  $\mainlogic$ formula $\varphi$ so that there is an accepting gainy run $\rho \in \transitions^*$ of $\aautomaton$ if, and only if, there is a model $\sigma$ so that $\sigma \models \varphi$ and $\sigma$ encodes the run $\rho$.

  The following are standard counter-blind conditions to check.
  \begin{enumerate}[label=\({\alph*}] \itemsep 0 cm \item \label{it:a} Every position
      satisfies $\tup{\atransition}$ for some unique $\atransition \in \transitions$.
    \item\label{it:b} The first position satisfies
      $\tup{(q_0,\ainstruction,q)}$ for some $q_0 \in \locations_0$,
      $\ainstruction \in \instructions$, $q\in \locations$.
    \item\label{it:c} The last position satisfies
      $\tup{(q,\ainstruction,q')}$ for some $q \in \locations$,
      $\ainstruction \in \instructions$, $q' \in \locations_F$.
    \item\label{it:d} There are no two consecutive positions $i$ and
      $i+1$ satisfying $\tup{(q,u,q')}$ and $\tup{(p,u',p')}$
      respectively, with $q' \neq p$.
  \end{enumerate}\medskip

\noindent  We use a variable $\avariable$, and variables
  $\avariable^\achain_{inc}, \avariable^\achain_{dec},
  \avariable^\achain_{i}$ for every chain $\achain$ and for every $i \in
  \interval{1}{f(\achain)}$.  Let us fix the bijections $\chi^\achain :
  \interval{0}{2^{f(\achain)}-1}\to \powerset{\interval{1}{f(\achain)}}$
  for every $\achain \in \interval{1}{n}$, that assign to each number
  $m$ the set of $1$-bit positions of the representation of $m$ in base
  $2$.  We say that a position $i$ in a run is \defstyle{$\achain$-incrementing}
  [resp.\ \defstyle{$\achain$-decrementing}] if it satisfies
  $\tup{(q,u,q')}$ for some $q,q' \in \locations$ and $u =
  \inc{\achain}$ [resp.\ $u = \dec{\achain}$].
 
  In the context of a
  model $\sigma$ with the properties \eqref{it:a}--\eqref{it:d}, we say
  that a counter position $i$ in a run operates on the $\achain$-counter $\acounter$, if
  $\chi^\achain(\acounter) = X_i$, where  \begin{align}\label{eq:Xi}
    X_i = \set{b \in \interval{1}{\amap(\achain)} \mid
    \sigma(i)(\avariable)=\sigma(i)(\avariable^\achain_{b})}.
      \end{align} Note that thus $0 \leq
    \acounter < 2^{f(\achain)}$.

  For every chain $\achain$, let us
  consider the following properties:
  \begin{enumerate}
  \itemsep 0 cm 
    \item \label{it:counters:2} Any two
      positions of $\sigma$ have different values of
      $\avariable^\achain_{inc}$ [resp.\ of $\avariable^\achain_{dec}$].
    \item \label{it:counters:4} For every position $i$ of $\sigma$
      operating on an $\achain$-counter $\acounter$ containing an
      instruction `$\cisfirst{\achain}$' [resp.\
      `$\cisnotfirst{\achain}$', `$\cislast{\achain}$',
      `$\cisnotlast{\achain}$'], we have $\acounter =0$ [resp.\
      $\acounter \neq 0$, $ \acounter =2^{f(\achain)}-1$, $ \acounter
      \neq 2^{f(\achain)}-1$].  \item \label{it:counters:5} For every
      position $i$ of $\sigma$ operating on an $\achain$-counter
      $\acounter$, \begin{itemize} \item if the position contains an
	  instruction `$\cnext{\achain}$' [resp.\
	  `$\cprevious{\achain}$'], then the next position $i+1$
	  operates on the $\achain$-counter $\acounter+1$ [resp.\
	  $\acounter-1$], \item otherwise, the position $i+1$ operates
	  on the $\achain$-counter $\acounter$.  \end{itemize} \item
	  \label{it:counters:6} For every $\achain$-incrementing
	  position $i$ of $\sigma$ operating on an $\achain$-counter
	  $\acounter$ there is a future $\achain$-decrementing position
	  $j > i$ operating on the same $\achain$-counter $\acounter$,
	  such  that
	  $\sigma(i)(\avariable_{\textit{inc}}^{\achain})=\sigma(j)(\avariable_{\textit{dec}}^{\achain})$.
      \end{enumerate}

  \begin{clm}\label{cl:props:suff:nec} There is a model satisfying the
    conditions above if, and only if, $\aautomaton$ has a gainy and
    accepting run.
  \end{clm}
\begin{proof}
 \textbf{(Only if)~}  Suppose we have a model $\sigma$ that verifies all the properties  above.   
Since by \eqref{it:counters:2}, all the positions have different data values for 
$\avariable^\achain_{\textit{dec}}$, it then 
follows that the position $j$ to which property \eqref{it:counters:6} makes reference is unique. 
Since, also by \eqref{it:counters:2}, all the 
positions have different data values for $\avariable^\achain_{\textit{inc}}$, we have that the 
$\achain$-decrement corresponds 
to at most one $\achain$-increment position in condition \eqref{it:counters:6}.
Moreover, it is an $\achain$-decrement of the same counter, that is, $X_i = X_{j}$ 
 (cf.~\eqref{eq:Xi}). 
For
  these reasons, for every $\achain$-increment there is a future
  $\achain$-decrement of the same counter which corresponds in a
  unique way to that increment. More precisely, for every chain
  $\achain$, the function
  \[
  \gamma^\achain_\sigma : \set{ 0 \leq i < |\sigma| : \text{$i$ is
      $\achain$-incrementing} } \to \set{ 0 \leq i < |\sigma| :
    \text{$i$ is $\achain$-decrementing} }
  \]
  where $\gamma^\achain_\sigma(i) = j \text{ if{f} }\sigma(i)(\avariable^\achain_{\textit{inc}}) =
  \sigma(j)(\avariable^\achain_{\textit{dec}})$ is well-defined and injective, and for every
  $\gamma^\achain_\sigma(i)=j$ we have that $j>i$ and $X_i = X_j$.

Consider $\arun \in \transitions^*$ as having the same
  length as $\sigma$ and such that $\arun(i) = (q,\ainstruction,q')$ whenever
  $\sigma,i \models \tup{(q,\ainstruction,q')}$. From conditions 
  \eqref{it:counters:4} and \eqref{it:counters:5}, it follows that the
  counter $c^\alpha_i$ corresponding to position $i$ from $\rho$, is
  equal to $\chi^{-1}(X_i)$. Therefore, the functions
  $\set{\gamma^\achain_\sigma}_{\achain \in \interval{1}{n}}$ witness the fact that $\arun$ is a gainy and accepting run of $\aautomaton$.

  \smallskip

 \textbf{(If)~}
  Let us now focus on the converse, by exhibiting a model
  satisfying the above properties, assuming that $\aautomaton$ has an
  accepting gainy run $\arun \in \transitions^*$. We therefore
  have, for every $\achain \in \interval{1}{n}$, an injective function
  \[
  \gamma^\achain : \set{i \mid \arun(i) \text{ is
      $\achain$-incrementing}} \to \set{i \mid \arun(i) \text{ is
      $\achain$-decrementing}}
  \]
  where for every $\gamma^{\achain}(i) = j$ we have that $j > i$ and
  $c_i^\achain = c_j^\achain$.

  We build a model $\sigma$ of the same length of $\arun$, and we now
  describe its data values. Let us fix two
  distinct data values $\adatum,\adatumbis$. For any position $i$, we have
  $\sigma(i)(\mathtt{t}_0)= \adatum$; and $\sigma(i)(\mathtt{t}_j)= \adatum$ if 
  the $i$th transition in $\arun$ is $\atransition_j$ 
  and $\sigma(i)(\mathtt{t}_j) =  \adatumbis$ otherwise. In this way
  we make sure that the properties \eqref{it:a}--\eqref{it:d} hold.

  We use distinct data values $\adatum_0, \dotsc, \adatum_{|\sigma|-1}$ and $\adatum'_0,
  \dotsc, \adatum'_{|\sigma|-1}$ (all different from $\adatum, \adatumbis$). 
  We define the data values of the remaining variables for any
  position $i$:
  \begin{itemize}
  \item For every $\achain$,
    $\sigma(i)(\avariable^\achain_{\textit{inc}})= \adatum_i$.
  \item For every $\achain$,
    \begin{itemize}
    \item if $i$ is $\achain$-decrementing, we define $\sigma(i)(\avariable^\achain_{dec}) =      
     \sigma(i')(\avariable^\achain_{inc})$ where $i' = (\gamma^\alpha)^{-1}(i)$;
      if such $(\gamma^\achain)^{-1}(i)$ does not exist, then
      $\sigma(i)(\avariable^\achain_{dec}) = \adatum'_i$;
    \item otherwise, if $i$ is not $\achain$-decrementing,
      $\sigma(i)(\avariable^\achain_{dec}) =  \adatum'_i$.
    \end{itemize}
  \item $\sigma(i)(\avariable)= \adatum$.
  \item For every $\achain$, $\sigma(i)(\avariable^\achain_{l}) = \sigma(i)(\avariable)=\adatum$ if $l \in
      \chi(c^\alpha_i)$, otherwise $\sigma(i)(\avariable^\achain_{l}) = \adatumbis$.
  \end{itemize}
Observe that these last two items ensure that the 
properties~\eqref{it:counters:4} and \eqref{it:counters:5} hold.

By definition, every position of $\sigma$ has a different data value for 
$\avariable^\achain_{\textit{inc}}$. Let us show that 
the same holds for $\avariable^\achain_{\textit{dec}}$. Note that at any position $i$, 
$\avariable^\achain_{\textit{dec}}$ has: the data value of 
$\sigma(i')(\avariable^\achain_{\textit{inc}})= \adatum_{i'}$ 
for some $i'<i$; or 
the data value $\adatum'_i$. If there were two positions $i<j$ with 
$\sigma(i)(\avariable^\achain_{\textit{dec}}) = \sigma(j)(\avariable^\achain_{\textit{dec}})$ 
it would then be because: (i) $\adatum'_i = \adatum'_j$; (ii) $\adatum_{i'} = 
\adatum'_j$ for some $i'<i$; (iii) $\adatum'_i = \adatum_{j'}$ 
for some $j'<j$; or (iv) 
$\adatum_{i'} = \adatum_{j'}$ for some $i'<i$, $j'<j$. It is evident that none of (i), (ii), (iii) 
can hold since all 
$\adatum_0, \dotsc, \adatum_{|\sigma|-1},
\adatum'_0, \dotsc, \adatum'_{|\sigma|-1}$ are distinct. For this reason, if (iv) holds, 
it means that $i' = j'$, and 
hence that $(\gamma^\achain)^{-1}(i) = (\gamma^\achain)^{-1}(j)$, implying that 
$\gamma^\achain$ is not injective, which is a contradiction. Hence, 
all the positions have different data values for $\avariable^\achain_{\textit{dec}}$. Therefore, $\sigma$ has the  property \eqref{it:counters:2}.

To show that $\sigma$ has property \eqref{it:counters:6}, let $i$ be a $\achain$-incrementing position of
 $\sigma$, remember that 
$\sigma(i)(\avariable^\achain_{\textit{inc}})= \adatum_i$. Note that position $\gamma^\achain(i) = j$ must be
 $\alpha$-decrementing on the same counter. 
By definition of the value of $\avariable^\achain_{\textit{dec}}$, we have that 
$\sigma(j)(\avariable^\achain_{\textit{dec}})$ must be equal to 
$\sigma((\gamma^\achain(j))^{-1})(\avariable^\achain_{\textit{inc}}) = \sigma(i)(\avariable^\achain_{\textit{inc}})=\adatum_i$. 
Thus, property \eqref{it:counters:6} holds.
\end{proof}

  We complete the reduction by showing that all the properties expressed
  before can be efficiently encoded in our logic.

  \begin{clm}\label{cl:props:exp:LRV} Properties
    \eqref{it:a}--\eqref{it:d} and
    \eqref{it:counters:2}--\eqref{it:counters:6} can be expressed by
    formulas of $\mainlogic$, that can be constructed in polynomial time
    in the size of $\aautomaton$.
  \end{clm}
  \begin{proof}
    Along the proof, we use the following formulas, for any $\achain
    \in \interval{1}{n}$ and for any $i \in \interval{1}{f(\achain)}$,
    \newcommand{\bit}{\textit{bit\,}^\achain} \[ \bit_i \egdef
    \avariable \oblieqlocal \avariable^\achain_i ,  \] 
where $\bit_1$ represents the most significant bit. Now, we  show how
    to code each of the properties.
    \begin{itemize} \item Properties
	\eqref{it:a}--\eqref{it:d} are easy  to express in
	$\mainlogic$ and present no complications.

    \item     For expressing \eqref{it:counters:2}, we force
      $\avariable^\achain_{\textit{inc}}$ and
      $\avariable^\achain_{\textit{dec}}$ to have a different data value
      for every position of the model with the formula \[ \lnot
      \sometimes ~
      \oblieq{\avariable^\achain_{\textit{inc}}}{\avariable^\achain_{\textit{inc}}}{\top}
      ~~\land~~ \lnot \sometimes ~
      \oblieq{\avariable^\achain_{\textit{dec}}}{\avariable^\achain_{\textit{dec}}}{\top}
      .  \] \item     Property \eqref{it:counters:4} is straightforward
      to express in $\mainlogic$.

    \item     We express property \eqref{it:counters:5} by first
      determining which is the bit $i$ that must be flipped for the
      increment of the counter pointer of the chain $\achain$ (the least significant bit is the $\amap(\alpha)$th one).
      \begin{align*} \textit{flip}^\achain_i ~~\egdef~~ \lnot \bit_i  ~\land~
	\bigwedge_{j > i} \bit_j \end{align*} Then by making sure that
	all the bits before $i$ are preserved.  \begin{align*}
	  \textit{copy}^\achain_i ~~\egdef~~ \bigwedge_{j < i} ~( \bit_j
	  \Leftrightarrow \mynext(\bit_j)) \end{align*} And by making
	  every bit greater or equal to $i$ be a zero.  \begin{align*}
	    \textit{zero}^\achain_i ~~\egdef~~ \bigwedge_{j > i} \mynext
	    (\lnot \bit_j) \end{align*} And finally by swapping the bit $i$.
	    \begin{align*} \textit{swap}^\achain_i ~~\egdef~~ \mynext ~\bit_i
	    \end{align*} Hence, the property to check is, for every
	    $\achain \in \interval{1}{n}$, \begin{align*}
	      \bigwedge_{(q,\cnext{\achain},q') \in \transitions}
	      \Big( \tup{(q,\cnext{\achain},q')} ~~~\Rightarrow~~~
	      \bigwedge_{i\in \interval{1}{n}}
	      ~\big(\textit{flip}^\achain_i ~~\Rightarrow~~
	      \textit{copy}^\achain_i ~\land~ \textit{zero}^\achain_i
	      ~\land~ \textit{swap}^\achain_i\big) \Big) .  \end{align*} The
	      formula expressing the property for decrements of counter
	      pointers ($\cprevious{\achain}$) is  analogous.

    \item     Finally, we express property \eqref{it:counters:6} by
      testing, for every $\achain$-incrementing position, \[
      \oblieq{\avariable^\achain_{\textit{inc}}}{\avariable^\achain_{\textit{dec}}}{
      \achain\textit{-dec}} \qquad \text{ where }~ \achain\textit{-dec}
      ~=~ \bigvee_{(q,\dec{\achain},q') \in \transitions}
      \tup{(q,\dec{\achain},q')} \] and for every $i \in
      \interval{1}{f(\achain)}$, \[ \bit_i ~\Rightarrow~
      \oblieq{\avariable^\achain_{\textit{inc}}}{\avariable^\achain_{\textit{dec}}}{\bit_i}
      ~~~\land~~~ \lnot\bit_i ~\Rightarrow~
      \oblieq{\avariable^\achain_{\textit{inc}}}{\avariable^\achain_{\textit{dec}}}{\lnot\bit_i}
      .  \] If the $\achain$-increment has some data value $\adatum$ at
      variable $\avariable^\achain_{\textit{inc}}$, there must be only
      one future position $j$ where $\avariable^\achain_{dec}$ carries
      the data value $\adatum$ ---since every $\avariable^\achain_{dec}$ has a
      different value, by \eqref{it:counters:2}. For this reason, both
      positions (the $\achain$-increment and the $\achain$-decrement)
      operate on the same counter, and thus the formula faithfully
      expresses  property \eqref{it:counters:6}.  \qedhere \end{itemize}
    \end{proof}
\noindent  As a corollary of Claims~\ref{cl:props:suff:nec} and
  \ref{cl:props:exp:LRV} we obtain a polynomial-time reduction from
  Gainy(1) into the satisfiability problem for
  $\mainlogic(\mynext,\sometimes)$.
\end{proof}
\else\begin{proof}[Proof sketch]
  Let $\aautomaton =
  \triple{\locations,\amap,1}{\locations_0,\locations_F}{\transitions}$
  be a chain system of level $1$ with $f :
  \interval{1}{n} \to \Nat$, having thus $n$ chains of counters, of
  respective size $2^{f(1)}, \dotsc, 2^{f(n)}$.

We encode a word $\rho \in \transitions^*$ that
  represents an accepting run. For this, we use the alphabet
  $\transitions$ of transitions. We can  simulate the labels $\delta=\set{t_1, \dotsc, t_m}$ with variables $\mathtt{t}_0, \dotsc, \mathtt{t}_m$, where a node has an encoding of the label $t_i$  if{f} the formula $\tup{t_i} = \mathtt{t}_0 \approx \mathtt{t}_i$ holds true.
We build an
  $\mainlogic$ formula $\varphi$ so that there is an accepting gainy run $\rho \in \transitions^*$ of $\aautomaton$ if, and only if, there is a model $\sigma$ so that $\sigma \models \varphi$ and $\sigma$ encodes $\rho$.

The counter-blind conditions to check are: (a) Every position
      satisfies $\tup{t}$ for some $t \in \transitions$; (b) the first position satisfies
      $\tup{(q_0,\ainstruction,q)}$ for some $q_0 \in \locations_0$,
      $\ainstruction \in \instructions$, $q\in \locations$; (c) the last position satisfies
      $\tup{(q,\ainstruction,q')}$ for some $q \in \locations$,
      $\ainstruction \in \instructions$, $q' \in \locations_F$; (d) no two consecutive positions $i$ and
      $i+1$ satisfy $\tup{(q,u,q')}$ and $\tup{(p,u',p')}$
      respectively, with $q' \neq p$. The difficulty is then in checking that the values of the counters encode indeed a correct gainy run.

We say that a position $i$ is \defstyle{$\achain$-incrementing}
  [resp.\ \defstyle{$\achain$-decrementing}] if it satisfies
  $\tup{(q,u,q')}$ for some $q,q' \in \locations$ and $u =
  \inc{\achain}$ [resp.\ $u = \dec{\achain}$].
We use a label $(\achain,i)$ and variables $\avariable^\achain_{\textit{inc}}$, $\avariable^\achain_{\textit{dec}}$ for every $\achain \in [1,n]$ and  $i \in [1,f(\achain)]$.
We say
  that a position $i$ operates on the $\achain$-counter $\acounter$, if $\tup{(\achain,j)}$ holds (i.e., position $i$ encodes the label $(\achain,j)$) for every position $j$
 of the representation of $\acounter$ in base
 $2$ containing a `$1$', and $\lnot \tup{(\achain,j)}$ for every position $j$ containing a `$0$'. Note that we can encode every value $0 \leq    \acounter < 2^{f(\achain)}$.

  For every chain $\achain$, let us
  consider the following properties:
  \begin{itemize}
    \item  Every two
      positions of $\sigma$ have different values of
      $\avariable^\achain_{inc}$ [resp.\ of $\avariable^\achain_{dec}$].
    \item  For every position $i$ of $\sigma$ 
      operating on an $\achain$-counter $\acounter$ with an
      instruction `$\cisfirst{\achain}$' [resp.\
      `$\cisnotfirst{\achain}$', `$\cislast{\achain}$',
      `$\cisnotlast{\achain}$'] , we have $\acounter =0$ [resp.\
      $\acounter \neq 0$, $ \acounter =2^{f(\achain)}-1$, $ \acounter
      \neq 2^{f(\achain)}-1$].  
\item For every
      position $i$ of $\sigma$ operating on an $\achain$-counter
      $\acounter$,  if the position contains an
	  instruction `$\cnext{\achain}$' [resp.\
	  `$\cprevious{\achain}$'], then the next position $i+1$
	  operates on the $\achain$-counter $\acounter+1$ [resp.\
	  $\acounter-1$]; otherwise, the position $i+1$ operates
 	  on the $\achain$-counter $\acounter$.  
\item For every $\achain$-incrementing
	  position $i$ of $\sigma$ operating on an $\achain$-counter
	  $\acounter$ there is a future $\achain$-decrementing position
	  $j > i$ on the same $\achain$-counter,
          	  so  that
	  $\sigma(i)( \avariable^\achain_{\textit{inc}})=\sigma(j)(\avariable^\achain_{\textit{dec}})$.
\end{itemize}
In fact, these properties together with (a)--(d) are sufficient and necessary to encode a gainy and accepting run of $\aautomaton$, and they can be all expressed in $\mainlogic$. Then, we obtain a polynomial-time reduction from
  Gainy(1) into the satisfiability problem for
  $\mainlogic(\mynext,\sometimes)$.
\end{proof}
\fi

\makeatletter{}\section{A Robust Equivalence}
\label{section-robustness}\enlargethispage{2\baselineskip}
We have seen that the satisfiability problem for  $\mainlogic$ is equivalent to
the control state reachability problem in an exponentially larger
VASS. In this section we evaluate how robust is this result with
$\mainlogic$ variants or fragments.
We consider infinite data words (instead of finite data words),
finite sets of MSO-definable temporal operators (instead of $\mynext$, $\previous$, $\since$, $\until$)
and also repetitions of pairs of values (instead of repetitions of single values). 
\makeatletter{}\subsection{Infinite Words with Multiple Data}
\label{section-variants}
So far, we have considered only finite
words with multiple data. 
\ifLONG
It is also natural to consider the variant with infinite words but it is known that this may lead 
to 
undecidability: for instance, 
freeze LTL with a single register is decidable over finite data words whereas
it is undecidable over infinite data 
words~\cite{DL-tocl08}. By contrast, FO$^2$ over finite data
words or over infinite data words is decidable~\cite{BDMSS06}.  
\else
It is also natural to consider the variant with infinite words, but it is known that this sometimes leads 
to 
undecidability. However, in this case the decidability and complexity results are preserved.
\fi
\newcommand{\SATo}[1]{\ensuremath{\textup{SAT}_\omega(#1)}}Let $\SATo{~}$ be the variant of the satisfiability problem $\SAT{~}$ in which 
infinite models of length $\omega$ are taken
into account instead of finite ones. 
The satisfaction relation for $\mainlogic$, $\pmainlogic$, $\mainlogic^{\top}$,
etc.\ on  $\omega$-models is defined accordingly.\newpage

\begin{prop} \label{proposition-mainlogic-infinite} \
\begin{description}
\itemsep 0 cm 
\item[(I)] \SATo{$\pmainlogic$} is decidable.
\item[(II)] \SATo{$\mainlogic$} is \twoexpspace-complete.
\end{description}
\end{prop}

\makeatletter{}\begin{proof} 
(I) The developments of Section~\ref{section-elimination} apply to the infinite case,
and since  $\SATo{\pmainlogic^{\top}}$ is shown decidable in~\cite{Demri&DSouza&Gascon12},
we get decidability of  $\SATo{\pmainlogic}$. \\ 
(II) First, note that  $\SATo{\mainlogic}$ is \twoexpspace-hard.
Indeed, there is a simple logarithmic-space reduction from $\SAT\mainlogic$ into 
$\SATo\mainlogic$, which can be performed as for standard LTL. 
Indeed, it is sufficient to introduce two new variables
$\avariable_{new}$ and $\avariablebis_{new}$, to state that
 $\avariable_{new} \oblieqlocal \avariablebis_{new}$ is true at a finite prefix of the model (herein
 $\avariable_{new} \oblieqlocal \avariablebis_{new}$ plays the role of a new propositional variable)
and to relativize all the temporal operators and obligations to positions on which 
 $\avariable_{new} \oblieqlocal \avariablebis_{new}$ holds true. 

Concerning the complexity upper bound, from 
Section~\ref{section-elimination}, we can conclude that 
there is a polynomial-time reduction from $\SATo\mainlogic$ into 
$\SATo{\mainlogic^{\top}}$. The satisfiability problem $\SATo{\mainlogic^{\top}}$
is shown decidable in~\cite{Demri&DSouza&Gascon12} and we can adapt 
developments from Section~\ref{section-reduction-from-top} to get also a \twoexpspace \ upper bound
for  $\SATo{\mainlogic^{\top}}$. This is the purpose of the rest of the proof.

By analyzing the constructions from~\cite[Section 7]{Demri&DSouza&Gascon12}, one
can show that $\aformula$ of $\mainlogic^{\top}$ built over the variables
$\set{\avariable_1, \ldots, \avariable_k}$ is  
$\omega$-satisfiable iff there 
are  $\asetter \subseteq \nepowerset{\set{\avariable_1, \ldots, \avariable_k}}$,
a VASS $\Aphi^{\asetter} = \triple{\states}{\asetter}{\transitions}$
along with sets $\states_{0}, \states_{f} \subseteq \states$ of 
\defstyle{initial} and \defstyle{final} states respectively
and a B\"uchi automaton $\aautomatonbis^{\asetter}$ such that:
\begin{enumerate}
\itemsep 0 cm
\item  $\Aphi^{\asetter}$ is the restriction of $\Aphi$ defined in 
       Section~\ref{section-reduction-from-top} and  
       $\langle \ainitst, \vect{0} \rangle
       \step{*} \langle \afinst, \vect{0} \rangle$ in $\Aphi^{\asetter}$  
       for some $\ainitst \in 
       \states_{0}$ and $\afinst \in \states_{f}$.
\item $\aautomatonbis^{\asetter}$ accepts a non-empty language.
\end{enumerate}

By arguments similar to those from Section~\ref{section-reduction-from-top}, existence
of a run $\langle \ainitst, \vect{0} \rangle
       \step{*} \langle \afinst, \vect{0} \rangle$ can be checked in
\twoexpspace.  Observe that in the construction 
in~\cite[Section 7]{Demri&DSouza&Gascon12}, counters in 
$\nepowerset{\set{\avariable_1, \ldots, \avariable_k}} \setminus \asetter$
are also updated in $\Aphi^{\asetter}$ (providing the VASS $\Aphi$)
but one can show that this is not needed because of {\em optional
decrement} condition and because $\aautomatonbis^{\asetter}$
is precisely designed to take care of the future obligations in the infinite 
related to the
counters in $\nepowerset{\set{\avariable_1, \ldots, \avariable_k}} \setminus \asetter$.
$\aautomatonbis^{\asetter}$ can be built in exponential time
and it is of exponential size in the size of $\aformula$. Hence,
non-emptiness of  $\aautomatonbis^{\asetter}$ can be checked in 
\expspace. Finally, the number of possible subsets $\asetter$ is only at most double
exponential in the size of $\aformula$, which allows to get a nondeterministic
algorithm in \twoexpspace \ and provides a \twoexpspace \ upper bound
by Savitch's Theorem~\cite{Savitch70}.
\end{proof}

\subsection{Adding MSO-Definable Temporal Operators}
It is standard to extend  \ifLONG the  linear-time temporal logic \fi   LTL  with MSO-definable temporal  operators (see 
e.g.~\cite{Wolper83,Gastin&Kuske03}), and the same can be done with $\mainlogic$. 
For this, one considers MSO formulas over a linear order $<$; that is, a structure $\mathbb{A} = (A,<)$ contains only one (binary) relational symbol $<$ and it is determined (modulo isomorphism) by the size of $A$. Let $\mathbb{A}_n = (\set{0, \dotsc, n-1}, <)$ [resp.\ $\mathbb{A}_\omega = (\N, <)$]  be the linear order $0 < \dotsb < n-1$ [resp.\ $0<1< \dotsb$]. A temporal operator $\oplus$ of arity $n$ is \defstyle{MSO-definable} whenever there is an MSO($<$) formula $\aformula(\avariable, \mathtt{P_1},\ldots,\mathtt{P_n})$ with a unique free position variable $\avariable$ and with $n$ free unary predicates $\mathtt{P_1},\ldots,\mathtt{P_n}$ such that
\[\sigma, i \models \oplus(\aformulabis_1, \ldots, \aformulabis_n) \text{ iff }
\mathbb{A}_{|\sigma|} \models \aformula(i,\aset_1, \ldots,\aset_n),\] where $\aformulabis_j$ is a formula of LRV extended with $\oplus$ and  $\aset_j = \set{ i \mid \sigma, i \models \aformulabis_j}$ for every $j$, see e.g.~\cite{Gastin&Kuske03}. The \twoexpspace \ upper bound is preserved with a fixed finite set of MSO-definable
\ifLONG temporal \fi  operators. 

\begin{thm} \label{theorem-mso}
Let $\set{\oplus_1, \ldots, \oplus_N}$ be a finite set of MSO-definable temporal operators.
Satisfiability problem for $\mainlogic$ extended with $\set{\oplus_1, \ldots, \oplus_N}$ is
\twoexpspace-complete. 
\end{thm}

\makeatletter{}\begin{proof} \twoexpspace-hardness is inherited from $\mainlogic$. 
In order to establish the complexity upper bound, first note that  
LTL extended with a fixed finite set of MSO-definable temporal operators 
preserves the nice properties of LTL (see e.g.~\cite{Gastin&Kuske03}):
\begin{itemize}
\item Model-checking and satisfiability problems are \pspace-complete.
\item Given a formula $\aformula$ from such an extension, one can build a B\"uchi automaton
      $\aautomaton_{\aformula}$ accepting exactly the models for $\aformula$
      and the size of $\aautomaton_{\aformula}$ is in $\mathcal{O}(2^{p(\length{\aformula})})$
      for some polynomial $p(\cdot)$ (depending on the finite set of MSO-definable operators). 
\end{itemize}
All these results hold because the set of MSO-definable operators is finite and fixed, 
otherwise the complexity for satisfiability and the size of the B\"uchi automata are of non-elementary 
magnitude in the worst case. Moreover, this holds for finite and infinite models.

In order to obtain the \twoexpspace, the following properties are now sufficient: 
\begin{enumerate}
\itemsep 0 cm 
\item Following developments from Section~\ref{section-elimination}, it is straightforward to show that
there is a logarithmic-space reduction from the satisfiability problem 
for $\mainlogic + \set{\oplus_1, \ldots, \oplus_N}$ into the satisfiability problem 
for $\mainlogic^{\top} + \set{\oplus_1, \ldots, \oplus_N}$.
\item By combining~\cite{Gastin&Kuske03} and~\cite[Theorem 4]{Demri&DSouza&Gascon12} (see also 
Appendix~\ref{section-appendix-aphifrompreviouspaper}),
a formula $\aformula$ in $\mainlogic^{\top} + \set{\oplus_1, \ldots, \oplus_N}$ is satisfiable 
iff $\pair{\ainitst}{\vect{0}} 
\step{*} \pair{\afinst}{\vect{0}}$ for some $\ainitst \in
\states_{0}$ and $\afinst \in \states_{f}$ in some VASS $\Aphi$ such that the number of states
in exponential in $\length{\aformula}$. The relatively small size for $\Aphi$ is due to the fact that
$\Aphi$ is built as the product of a VASS checking obligations (of exponential size in the
number of variables and in the size of the local equalities) 
and of a finite-state automaton accepting the symbolic models
of $\aformula$ of exponential size thanks to~\cite{Gastin&Kuske03} (see also more details in 
Section~\ref{section-reduction-from-top}) . 
\end{enumerate}
By using arguments from Section~\ref{section-reduction-from-top}, we can then reduce existence of a run 
 $\pair{\ainitst}{\vect{0}} 
\step{*} \pair{\afinst}{\vect{0}}$ to an instance of the control state reachability problem in some
VASS of linear size  in the size of $\Aphi$, whence the \twoexpspace \ upper bound. 
\end{proof} 

Note that $\pmainlogic$ augmented with MSO-definable temporal operators is decidable too. 
\ifLONG
Indeed, it can be
translated into $\pmainlogic^{\top}$ augmented with MSO-definable temporal operators by adapting
the developments from Section~\ref{section-elimination}. Then, the satisfiability problem
of this latter logic can be translated in the reachability problem for VASS as done 
in~\cite[Theorem 4]{Demri&DSouza&Gascon12} except that finite-state automata are built according 
to~\cite{Gastin&Kuske03}. 
\fi 

\subsection{The $\mathtt{Now}$ Operator or the Effects of Moving the Origin}
The
satisfiability problem for  Past LTL with the temporal operator $\mathtt{Now}$ is known to be
\expspace-complete~\cite{Laroussinie&Markey&Schnoebelen02}.  
The satisfaction relation is parameterised by
the current position of the origin and past-time temporal operators
use that position. For instance, given $i,o \in \Nat$ with $o \leq i$ (`$o$' is the position of the origin),
\begin{center}
\begin{tabular}{lcl}
$\sigma, i \models_{o} \mathtt{Now} \ \aformula$ & $\equivdef$ & 
$\sigma, i \models_{i} \aformula$ \\ 
$\sigma, i \models_{o} \aformula_1 \since \aformula_2$
& $\equivdef$ & there is $j \in \interval{o}{i}$ such that
$\sigma, j \models_{o} \aformula_2$ and \\
& & 
for all $j' \in \interval{j-1}{i}$, we have $\sigma, j' \models_{o} \aformula_1$.
\end{tabular}
\end{center}
The powerful operator $\mathtt{Now}$ can be
obviously defined in MSO  but not with the above definition since 
it requires  {\em two}
free position variables, one of which refers to the current position of the origin
and 
past-time operators are interpreted 
relatively to that
position. 

\begin{thm} 
\label{theorem-now}
\SAT{$\mainlogic \ + \ \mathtt{Now}$} is
\twoexpspace-complete. 
\end{thm}

\makeatletter{}
\begin{proof} Again,  \twoexpspace-hardness is inherited from $\mainlogic$. 
In order to get the \twoexpspace,  we use arguments similar to those from forthcoming Proposition~\ref{proposition-pspace}. 
Indeed, the decidability proof from~\cite[Theorem 4]{Demri&DSouza&Gascon12} 
(see also Appendix~\ref{section-appendix-aphifrompreviouspaper}) can be adapted to 
$\mainlogic + \mathtt{Now}$. The only difference is that the finite-state automaton  
is of double  exponential size, see details below.
Despite this exponential blow-up, the 
\twoexpspace \ upper bound can be preserved. 

Let $\aformula$ be a formula with $k$ variables. There exist a VASS 
$\Adec = \triple{\locations}{\counters}{\transitions}$ 
and $\locations_0,\locations_f \subseteq \locations$
such that
$\aformula$ is satisfiable iff there are $\afinst \in \locations_f$
and  $\ainitst \in \locations_0$ such that  
$\pair{\afinst}{\vect{0}}  \step{*}
  \pair{\ainitst}{\acountval}$ for some counter valuation
  $\acountval$.
Note that the number of counters in $\Adec$  is bounded by 
$2^k$, $\card{\locations}$ is double exponential in $\length{\aformula}$
and  the maximal value for an update
in a transition of $\Adec$ is $k$.
Indeed,  a formula from Past LTL+$\mathtt{Now}$ is equivalent
to a B\"uchi automaton of double exponential size in its size~\cite{Laroussinie&Markey&Schnoebelen02}. 
Moreover, deciding whether a state in $\locations$ belongs to
$\locations_0$ [resp. $\locations_f$] can be checked in exponential space
in  $\length{\aformula}$ and $\transitions$ can be decided
in exponential space too. 
By using~\cite{Rackoff78} (see also~\cite{Demrietal09}),
we can easily conclude 
that $\pair{\afinst}{\vect{0}}  \step{*}
  \pair{\ainitst}{\acountval}$ for some counter valuation
  $\acountval$ iff  $\pair{\afinst}{\vect{0}}  \step{*}
  \pair{\ainitst}{\acountval'}$ for some
  $\acountval'$ such that the length of the run is bounded
by $p(\length{\Adec} + \mathtt{max}(\Adec))^{\mathfrak{f}(k)}$ where 
$p(\cdot)$ is a polynomial, $\mathfrak{f}$ is a map of double exponential growth and $\mathtt{max}(\Adec)$
denotes the maximal absolute value in an update (bounded by $k$ presently).
In order to take advantage of the results on VAS (vector addition system --without states--), we use the translation
from VASS to VAS introduced in~\cite{Hopcroft&Pansiot79}: if a VASS has 
$N_1$ control states, the maximal absolute value in an update is $N_2$ and it has $N_3$ counters,
then we can build a VAS (being able to preserve coverability properties)
such that it has $N_3+3$ counters, the maximal absolute value in an update is $\max(N_2, N_1^2)$.
From $\Adec$, we can indeed build an equivalent VAS with a number of counters bounded by
$2^k+3$ and with a maximal absolute value in an an update at most double exponential in 
the size of $\length{\aformula}$. 
So, the length of the run is at most triple exponential in
$\length{\aformula}$. Consequently, 
the satisfiability problem for $\mainlogic + \mathtt{Now}$ is in
\twoexpspace. 
\end{proof}

This contrasts with the undecidability results 
from~\cite[Theorem 5]{Kara&Schwentick&Zeume10} in presence of the operator  $\mathtt{Now}$.
In~\cite[Theorem 5]{Kara&Schwentick&Zeume10}, the logic has two sorts of formulas: position formulae from class formulae
(a class being a sequence of positions with the same data value). It is a more expressive formalism that can navigate both the word in the usual way, as well as its data classes.

\makeatletter{}
\subsection{Bounding the Number of Variables}
\label{section-bounding-nb-variables}

Given the relationship between the number of variables in a
formula and the number of counters needed in the corresponding VASS, we investigate the consequences of fixing the number of
variables. Interestingly, this classical restriction has an
effect only for $\mainlogic^{\top}$, i.e., when test formulas
$\aformula$ are restricted to $\top$ in
$\oblieq{\avariable}{\avariablebis}{\aformula}$. Let $\mainlogic_k$ [resp.\ $\mainlogic_k^\top$, $\pmainlogic_k^\top$] be the restriction to formulas with at most $k$ variables.
In~\cite[Theorem 5]{Demri&DSouza&Gascon12}, it is shown that  
\SAT{$\mainlogic^{\top}_{1}$} is \pspace-complete by establishing a reduction
into the reachability problem for VASS when counter values are 
linearly bounded. Below, we generalize this result for any $k \geq 1$ by
using the proof of Theorem~\ref{thm:2ExpspUpBound} and the fact that
the control state reachability problem  for VASS with at most
$k$ counters (where $k$ is a constant) is in \pspace.

\begin{prop} 
\label{proposition-pspace}
 For every $k \geq 1$,
\SAT{$\mainlogic_k^{\top}$}  is \pspace-complete.
\end{prop}
\makeatletter{}\ifLONG
\begin{proof}
\pspace-hardness  is due to the fact that LTL with a 
single propositional
variable is \pspace-hard~\cite{Demri&Schnoebelen02}, which can be easily simulated with the atomic formula
$\avariable \oblieqlocal \mynext \avariable$. 
Let $k \geq 1$ be some fixed value and $\aformula \in 
\mainlogic^{\top}_k$. In the proof of Theorem~\ref{thm:2ExpspUpBound},
we have seen that there exist a VASS 
$\Adec = \triple{\locations}{\counters}{\transitions}$ 
and $\locations_0,\locations_f \subseteq \locations$
such that
$\aformula$ is satisfiable iff there are $\afinst \in \locations_f$
and  $\ainitst \in \locations_0$ such that  
$\pair{\afinst}{\vect{0}}  \step{*}
  \pair{\ainitst}{\acountval}$ for some counter valuation
  $\acountval$.
Note that the number of counters in $\Adec$  is bounded by 
$2^k$, $\card{\locations}$ is exponential in $\length{\aformula}$
and  the maximal value for an update
in a transition of $\Adec$ is $k$. 
Moreover, deciding whether a state in $\locations$ belongs to
$\locations_0$ [resp. $\locations_f$] can be checked in polynomial space
in  $\length{\aformula}$ and $\transitions$ can be decided
in polynomial space too. 
By using~\cite{Rackoff78} (see also~\cite{Demrietal09}),
we can easily conclude 
that $\pair{\afinst}{\vect{0}}  \step{*}
  \pair{\ainitst}{\acountval}$ for some counter valuation
  $\acountval$ iff  $\pair{\afinst}{\vect{0}}  \step{*}
  \pair{\ainitst}{\acountval'}$ for some
  $\acountval'$ such that the length of the run is bounded
by $p(\length{\Adec} + \mathtt{max}(\Adec))^{\mathfrak{f}(k)}$ where 
$p(\cdot)$ is a polynomial, $\mathfrak{f}$ is a map of double exponential growth and $\mathtt{max}(\Adec)$
denotes the maximal absolute value in an update (bounded by $k$ presently).
Since $k$ is fixed, the length of the run is at most exponential in
$\length{\aformula}$. Consequently, the following polynomial-space 
nondeterministic algorithm allows to 
check whether $\aformula$
is satisfiable. Guess $\afinst \in \locations_f$, $\ainitst \in \locations_0$
and guess on-the-fly a run of length at most  $p(\length{\Adec}  + \mathtt{max}(\Adec))^{\mathfrak{f}(k)}$
from $\pair{\afinst}{\vect{0}}$ to some 
$\pair{\ainitst}{\acountval'}$. Counter valuations can be
represented in polynomial space too. 
By Savitch's Theorem~\cite{Savitch70}, we conclude that 
the satisfiability problem for $\mainlogic^{\top}_k$ is in \pspace.
\end{proof} 
\else
\fi

This does not imply that $\mainlogic_k$ is in \pspace, since the reduction from $\mainlogic$ into 
$\mainlogic^\top$ in Section~\ref{section-elimination} introduces new variables. In fact, it introduces a number of 
variables that depends on the size of the formula. It turns out that this is unavoidable, and 
that its satisfiability problem is $2\expspace$-hard, by the following reduction.
\begin{lem}\label{lem:LRV-to-LRV1}
  There is a polynomial-time reduction from 
\SAT{$\mainlogic^\top$} into \SAT{$\mainlogic_1$} [resp.\  \SAT{$\pmainlogic^\top$} and \SAT{$\pmainlogic_1$}].
\end{lem}
\makeatletter{}\ifLONG
\begin{proof}
The idea of the coding is the following. Suppose we have a formula $\varphi \in \mainlogic^\top$ using 
$k$ variables $\avariable_1, \dotsc, \avariable_k$ so that $\sigma \models \varphi$. 

We will encode $\sigma$ in a model $\sigma_\varphi$ that encodes in only one variable, say $\avariable$ the whole model of $\sigma$ restricted to $\avariable_1, \dotsc, \avariable_k$. To this end, $\sigma_\varphi$ is divided into $N$ segments $s_1 \dotsb s_N$ of equal length, where $N=|\sigma|$. A special fresh data value is used as a special constant. 
Suppose that $\adatum$ is a data value that is not in $\sigma$. Then, each  segment $s_i$ has 
length $k' = 2k+1$, and is 
defined as the data values ``$\adatum ~ \adatum_1 ~ \adatum~  \adatum_2~ \dotsc \adatum~ \adatum_k ~\adatum$'', 
where $\adatum_j = \sigma(i)(\avariable_j)$. 
Figure~\ref{fig:lrv-to-lrv1} contains an example for $k=3$ and $N=3$. 
\begin{figure}
  \centering
  \includegraphics[scale=.46]{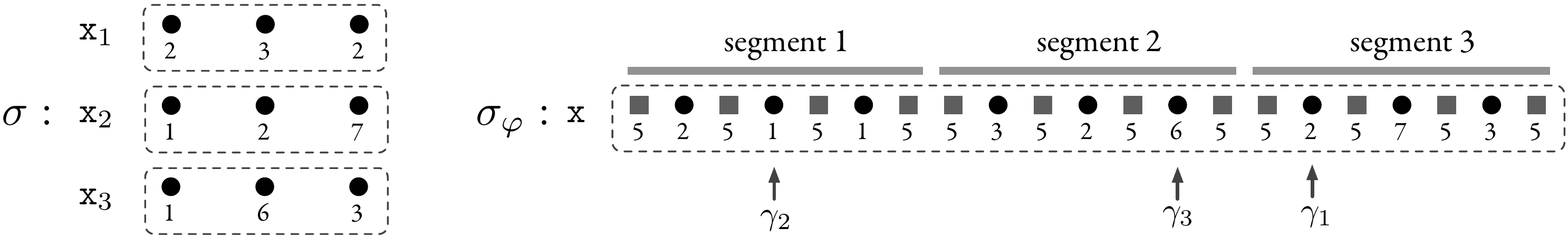}
  \caption{Example of the reduction from \SAT{$\mainlogic^\top$} into \SAT{$\mainlogic_1$}, for $k=3$, $N=3$ and $\adatum=5$}
    \label{fig:lrv-to-lrv1}
\end{figure}
In fact, we can force that the model has this shape with $\mainlogic_1$.
Note that with this coding, we can tell that we are between two segments if there are two consecutive equal data values. In fact, we are at a position corresponding to $\avariable_i$ (for $i \in \interval{1}{k}$) 
inside a segment  if we are standing at the $2i$-th element of a segment, and we can test this with the formula 
\[
\gamma_i ~~~=~~~ \mynext^{k' - 2i} \avariable \approx \mynext^{k' - 2i+1} \avariable ~~\lor~~
(\mynext^{k' - 2i} \top \land \lnot \mynext^{k' - 2i+1} \top) .
\]
Using these formulas $\gamma_i$, we can translate any test $\oblieq{\avariable_i}{\avariable_j}{\top}$ into a formula that
 \begin{enumerate}
 \itemsep 0 cm 
 \item moves to the position $2i$ of the segment (the one corresponding to the $\avariable_i$ data value),
 \item tests $\oblieq{\avariable}{\avariable}{\gamma_j}$.
 \end{enumerate}
We can do this similarly with all formulas.

\smallskip

Let us now explain in more detail how to do the translation, and why it works.

Consider the following property of a given position $i$ multiple of $k'$ of a model $\sigma_\varphi$:
\begin{itemize}
\itemsep 0 cm 
\item for every $j \in \interval{0}{k'-1}$,
  $\sigma_\varphi(i+j)(\avariable)=\sigma_\varphi(i)(\avariable)$ if, and only if,
  $j$ is even, and
\item either $i+k' \geq |\sigma_\varphi|$ or $\sigma_\varphi(i+k') = \sigma_\varphi(i)$.
\end{itemize}
This property can be easily expressed with a formula
\newcommand{\segk}{\textit{segment-}k'}
\[
 \segk  ~~~=~~~ \bigwedge_{\substack{0 \leq j \leq k'-1,\\ j \text{ is even}}} \avariable \approx \mynext^j\avariable ~~\land~~ \bigwedge_{\substack{0 \leq j \leq k'-1,\\ j \text{ is odd}}} \avariable \not\approx \mynext^j\avariable ~~\land~~ (\lnot\mynext^{k'}\top \lor \avariable \approx \mynext^{k'}\avariable) .
\]
Note that this formula tests that we are standing at the beginning of a segment in particular. It is now straightforward to produce a $\mainlogic_1$ formula that tests
\begin{itemize}
\itemsep 0 cm 
\item $\sigma_\varphi,0 \models \segk$
\item for every position $i$ so that $\sigma_\varphi(i)(\avariable) = \sigma_\varphi(i+1)(\avariable)$, we have $\sigma, i+1 \models \segk$.
\end{itemize}
\newcommand{\manyseg}{\textit{many-segments}}
Let us call \manyseg \  the formula expressing such a property. Note that the property implies that $\sigma_\varphi$ is a succession of segments, as the one of Figure~\ref{fig:lrv-to-lrv1}.

We give now the translation of a formula $\varphi$ of $\mainlogic^\top$ with $k$ variables into a formula $\varphi'$ of $\mainlogic_1$ with $1$ variable.
\begin{align*}
  tr(\mynext \psi) &~=~ \overbrace{\mynext \cdots \mynext}^{k' \ {\rm times}} tr(\psi)\\
  tr(\mynext^{-1} \psi) &~=~ \overbrace{\mynext^{-1} \cdots \mynext^{-1}}^{k' \ {\rm times}} tr(\psi)\\
  tr(\sometimes \psi) &~=~ \sometimes(\segk \land tr(\psi))\\
  tr(\psi \until \gamma) &~=~ (\segk \Rightarrow tr(\psi) ~\until~ (\segk \wedge tr(\gamma))\\
  tr(\psi \since \gamma) &~=~ (\segk \Rightarrow tr(\psi) ~\since~ (\segk \wedge tr(\gamma))\\
  tr(\avariable_i \approx\mynext^\ell\avariable_j) &~=~\mynext^{2i-1} \avariable \approx \mynext^{\ell \cdot k'+2j-1} \avariable \qquad \text{(and similarly for $\not \approx$)}\medskip
\end{align*}
\noindent Now we have to translate $\oblieq{\avariable_i}{\avariable_j}{\top}$. This would be translated in our encoding by saying that the $i$-th position of the current segment is equal to the $j$-th position of a future segment. Note that the following formula
\[
    \xi_{i,j} ~=~ \mynext^{2i-1} ( \oblieq{\avariable}{\avariable}{\gamma_j})
\]
does not exactly encode this property. For example, consider  that we would like to test $\oblieq{\avariable_2}{\avariable_3}{\top}$ at the first element of the model $\sigma_\varphi$ depicted in Figure~\ref{fig:lrv-to-lrv1}. Although $\xi_{2,3}$ holds, 
the property is not true, there is no \emph{future} segment with the data value $1$ in the position encoding $\avariable_3$. 
In fact, the formula $\xi$ encodes correctly the property only when $\avariable_i \not\approx \avariable_j$. 
However, this is not a problem since when $\avariable_i \approx \avariable_j$ the formula 
$\oblieq{\avariable_i}{\avariable_j}{\top}$ is equivalent to $\oblieq{\avariable_j}{\avariable_j}{\top}$. 
We can then translate the formula as follows.
\begin{align*}
  tr(\oblieq{\avariable_i}{\avariable_j}{\top}) &~=~ (tr(\avariable_i \approx \avariable_j) \land \xi_{j,j}) ~~\lor~~ (tr(\avariable_i \not\approx \avariable_j) \land \xi_{i,j}) 
\end{align*}
Recall that $\oblineq{\avariable_i}{\avariable_j}{\top}$ is not translated since such formulae can be eliminated.
We then define $\varphi' = \manyseg \land tr(\varphi)$.

\begin{clm}
  $\varphi$ is satisfiable if, and only if, $\varphi'$ is satisfiable.
\end{clm}
\begin{proof}
In fact, if $\sigma \models \varphi$, then by the discussion above, $\sigma_\varphi \models \varphi'$. If, on the other hand, $\sigma' \models \varphi'$ for some $\sigma'$, then since $\sigma' \models \manyseg$ it has to be a succession of segments of size $2k+1$, and we can recover a model $\sigma$ of size $|\sigma'| / 2k+1$ where $\sigma(i)(\avariable_j)$ is the data value of the $2j$-th position of the $i$-th segment of $\sigma'$. In this model, we have that $\sigma \models \varphi$.
\end{proof}
This coding can also be extended with past obligations in a straightforward way, 
\begin{align*}
  tr(\poblieq{\avariable_i}{\avariable_j}{\aformula}) &~=~ (tr(\avariable_i \approx \avariable_j) \land \xi^{-1}_{j,j}) ~~\lor~~ (tr(\avariable_i \not\approx \avariable_j) \land \xi^{-1}_{i,j})\quad \text{where}\\
    \xi^{-1}_{i,j} &~=~ \mynext^{2i-1} ( \poblieq{\avariable}{\avariable}{\gamma_j}) .
\end{align*}
Therefore, there is also a reduction from $\pmainlogic^\top$ into $\pmainlogic_1$.
\end{proof}
\else\begin{proof}[Proof sketch] 
 Let 
  $\varphi \in \mainlogic^\top$ using $k$ variables $\avariable_1,
  \dotsc, \avariable_k$ so that $\sigma \models \varphi$.
  We will encode $\sigma$ restricted to $\avariable_1, \dotsc, \avariable_k$ inside a model $\sigma_\varphi$ with only one variable, say $\avariable$. To this end,
  $\sigma_\varphi$ is divided into $N$ segments $s_1 \dotsb s_N$ of
  equal length, where $N=|\sigma|$. A special data value $d$ not used
  in $\sigma$ plays the role of \emph{delimiter} between segments and
  between positions that code values from $\sigma$. Suppose that
  $d$ is a data value that is not
  in $\sigma$. Then, each segment $s_i$ has length $k' =
  2k+1$, and is defined as the data values ``$d ~ d_1 ~ d~ d_2~ \dotsc
  d~ d_k ~d$'', where $d_j =
  \sigma(i)(\avariable_j)$. Figure~\ref{fig:lrv2lrv1} contains an example.
  \begin{figure*}
    \centering
    \includegraphics[scale=.46]{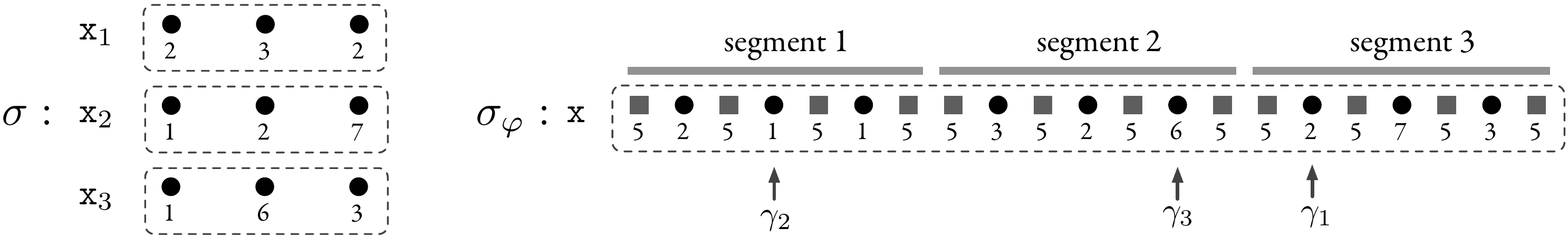}
    \caption{Example of the reduction from \SAT{$\mainlogic^\top$} into \SAT{$\mainlogic_1$}, for $k=3$, $N=3$ and $d=5$.}
\label{fig:lrv2lrv1}
\end{figure*}
              We can force that the model has this shape with
  $\mainlogic_1$.  With this coding, we can tell that we are
  between two segments if there are two consecutive equal data
  values. In fact, we are at a position corresponding to
  $\avariable_i$ (for $i \in \interval{1}{k}$) inside a segment if we are
  standing at the $2i$-th element of a segment, and we can test this
  with the formula
$
  \gamma_i ~=~ \mynext^{k' - 2i} \avariable \approx \mynext^{k' -
    2i+1} \avariable ~\lor~ (\mynext^{k' - 2i} \top \land \lnot
  \mynext^{k' - 2i+1} \top)
$.
 Using the $\gamma_i$'s, we translate 
  $\oblieq{\avariable_i}{\avariable_j}{\top}$ in $\varphi$ into a formula that: moves to the position $2i$ of the segment (the one
    corresponding to $\avariable_i$), and tests $\oblieq{\avariable}{\avariable}{\gamma_j}$.
  We can do this similarly with all formulas.
\end{proof}
\fi

\begin{cor}
For all $k \geq 1$, \SAT{$\mainlogic_k$} is $2\expspace$-complete.
\end{cor}
\begin{cor}\label{cor:PLRV1-equiv-VASS}
For all  $k \geq 1$, \SAT{$\pmainlogic_k$} is as hard as Reach(VASS).
\end{cor}
\makeatletter{}
\subsection{The Power of Pairs of Repeating Values}
\ifLONG
Let us consider the last variant of $\mainlogic$ in which 
the repetition of tuples of values is possible. 
Such an extension amounts to introducing additional variables in a first-order setting.
This may lead to undecidability since 3 variables are enough for undecidability, see 
e.g.~\cite{BDMSS06:journal}. However,  $\mainlogic$  makes a restricted
use of variables, leading to \twoexpspace-completeness.
There might be hope that $\mainlogic$ augmented with repetitions of pairs of values
have a reasonable computational cost. 
Consider an extension to $\mainlogic$ with  
 atomic formulas of the form
$\oblieq{(\avariable_1, \dotsc, \avariable_k)}{(\avariablebis_1, \dotsc, \avariablebis_k)}{\aformula}$
where $\avariable_1, \dotsc, \avariable_k, \avariablebis_1, \dotsc, \avariablebis_k \in \Var$. 
\else
\noindent
Let us consider an \ifLONG last \fi  $\mainlogic$ variant  so that 
repetition of tuples of values is possible: we add \ifLONG atomic \fi formulas of the form
$\oblieq{(\avariable_1, \dotsc, \avariable_k)}{(\avariablebis_1, \dotsc, \avariablebis_k)}{\aformula}$
where $\avariable_1, \dotsc, \avariable_k, \avariablebis_1, \dotsc, \avariablebis_k \in \Var$. 
\fi 
This extends  $\oblieq{\avariable}{\avariablebis}{\varphi}$ by testing whether the vector of data values 
from the variables $(\avariable_1, \dotsc, \avariable_k)$ of the current position coincides with 
that of $(\avariablebis_1, \dotsc, \avariablebis_k)$ in a future position. 

\ifLONG
The semantics are extended 
accordingly:
\begin{center}
\begin{tabular}{lcl}
$\sigma, i \models    \oblieq{(\avariable_1, \dotsc, \avariable_k)}{(\avariablebis_1, \dotsc, \avariablebis_k)}{\varphi}$ & iff & 
  there exists $j$ such that $i < j < \length{\sigma}$, $\sigma, j \models \varphi$, \\
& &  and $\sigma(i)(\avariable_l) = \sigma(j)(\avariablebis_l)$ for every $l \in \interval{1}{k}$.
\end{tabular}
\end{center}
\fi 

We call this extension $\mainlogic_{vec}$. 
Unfortunately, we can show that \SAT{$\mainlogic_{vec}$} is undecidable, 
even when only tuples of dimension 2 are allowed. This is
proved by reduction from a variant of Post's Correspondence Problem
(PCP), see below. In order to code solutions of PCP instances, we adapt a proof
technique used in~\cite{BDMSS06:journal} for first-order logic with
two variables and two equivalence relations on words. However, 
our proof uses only \emph{future}  modalities (unlike the proof of~\cite[Proposition 27]{BDMSS06:journal})
and no past obligations (unlike  the proof of~\cite[Theorem 4]{Kara&Schwentick&Zeume10}). 
To prove this result, we work with a variant of the
PCP problem in which solutions $u_{i_1} \dotsb u_{i_n} = v_{i_1}
\dotsb v_{i_n}$ have to satisfy $|u_{i_1} \dotsb u_{i_j}| \leq |v_{i_1} \dotsb v_{i_j}|$ for every $j$. 

\begin{thm} \label{theorem-undecidability-with-pairs}
  \SAT{$\mainlogic_{vec}(\mynext,\until)$} is undecidable.
\end{thm}

\makeatletter{}Let us introduce the problem below as a variant of PCP:
\begin{center}
  \begin{tabular}{|rl|}
    \hline
\textsc{Problem:}& Modified Directed Post's Correspondence Problem (\pcpdir)
\\
    \hline
    \textsc{Input:} & A finite alphabet $\Sigma$, $n \in \Nat$,
    $u_1, \dotsc, u_n, v_1, \dotsc, v_n \in \Sigma^+$, $\length{u_1},
    \length{u_2} \leq 2$ and\\& $\length{v_{1}}, \length{v_{2}} \ge 3$.
\\
    \textsc{Question:} & Are there indices $1 \leq i_1, \dotsc, i_m \leq n$ \\&
    so that $u_{i_1} \dotsb u_{i_m} = v_{i_1} \dotsb v_{i_m}$,
    $i_{1} \in \set{1,2}$, $\length{u_{i_1} \dotsb u_{i_m}}$ is even and\\
    & for every $\length{u_{i_{1}}} < j < \length{u_{i_1} \dotsb
    u_{i_m}}$, if the $j$\textsuperscript{th} position of $v_{i_1}
    \dotsb v_{i_m}$ occurs in\\ & $v_{i_{k}}$ for some $k$, then the
    $j$\textsuperscript{th} position of $u_{i_1} \dotsb u_{i_m}$
    occurs in $u_{i_{k'}}$ for some\\ & $k' > k$?\\
    \hline
  \end{tabular}
\end{center}

\begin{lem} \label{lemma-undecidability-pcpdir}
  The \pcpdir{} problem is undecidable.
\end{lem}

This is a corollary of the undecidability proof for PCP as 
 done in~\cite{hopcroft:ullman}.
By reducing \pcpdir{} to satisfiability for $\mainlogic_{vec}$, we get undecidability. 

\makeatletter{}
\begin{proof}  
In~\cite{hopcroft:ullman}, 
the halting problem for Turing machines is first
  reduced to a variation of PCP called modified PCP. In
  modified PCP, a requirement is that the solution should begin with
  the pair $u_{1}, v_{1}$ (this requirement can be easily eliminated
  by encoding it into the standard PCP, but here we find it convenient
  to work in the presence of a generalisation of this requirement). To
  ensure that there is a solution of even length whenever there is a
  solution, we let $u_{2} = \$u_{1}$ and $v_{2} = \$v_{1}$, where $\$$
  is a new symbol. Now $u_{1} u_{i_{2}} \dotsb u_{i_{m}} = v_{1}
  v_{i_{2}} \dotsb v_{i_{m}}$ is a solution iff $u_{2} u_{i_{2}}
  \dotsb u_{i_{m}} = v_{2} v_{i_{2}} \dotsb v_{i_{m}}$ is a solution.
  In modified PCP, $\length{u_{1}} = 1$ and $\length{v_{1}} \ge 3$.
  Hence, $\length{ u_{2}} = 2$ and $\length{ v_{2}} \ge 3$. The
  resulting set of strings $u_{1}, \ldots, u_{n}, v_{1}, \ldots,
  v_{n}$ make up our instance of \pcpdir{}.

  In the encoding of the halting problem from the cited 
  work,
            $\length{u_{i}} \le 2$ 
  for any $i$ between $1$ and $n$ (this
  continues to hold even after we add $\$$ to the first pair as
  above). A close examination of the proof in the cited 
    work reveals
  that if the modified PCP instance  $u_1, \dotsc, u_n, v_1, \dotsc,
  v_n \in \Sigma^*$ has a solution $u_{i_1} \dotsb u_{i_m} = v_{i_1}
  \dotsb v_{i_m}$, then for every $\length{u_{i_{1}}} < j <
  \length{u_{i_1} \dotsb u_{i_m}}$, if the $j$\textsuperscript{th}
  position of $v_{i_1} \dotsb v_{i_m}$ occurs in $v_{i_{k}}$ for some
  $k$, then the $j$\textsuperscript{th} position of $u_{i_1} \dotsb
  u_{i_m}$ occurs in $u_{i_{k'}}$ for some $k' > k$ (we call this the
  directedness property). In short, the reason for this is that for
  any $k \in \Nat$, if $v_{i_{1}} \dotsb v_{i_{k}}$ encodes the first
  $\ell+1$ consecutive configurations of a Turing machine, then
  $u_{i_{1}} \dotsb u_{i_{k}}$ encodes the first $\ell$
  configurations. Hence, for any letter in $v_{i_{k+1}}$ (which starts
  encoding $(\ell + 2)$\textsuperscript{th} configuration), the
  corresponding letter cannot occur in $u_{i_1} \dotsb u_{i_k+1}$
  (unless the single string $u_{i_{k+1}}$ encodes the entire
  $(\ell+1)$\textsuperscript{st} configuration and starts encoding
  $(\ell + 2)$\textsuperscript{th} configuration; this however is not
  possible since there are at most $3$ letters in $u_{i_{k+1}}$ and
  encoding a configuration requires at least $4$ letters). After the
  last configuration has been encoded, the length of $u_{i_{1}} \dotsb
  u_{i_{k}}$ starts catching up with the length of $v_{i_{1}} \dotsb
  v_{i_{k}}$ with the help of pairs $(u,v)$ where $\length{u} = 2$ and
  $\length{v} = 1$. However, as long as the lengths do not catch up,
  the directedness property continues to hold. As soon as the lengths
  do catch up, it is a solution to the PCP. Only the positions of
  $u_{i_{1}}$ and the last position of the solution violate the
  directedness property.
\end{proof}

Given an instance $\apcpinst$ of \pcpdir{}, we construct an $\dcltltup(\mynext,\until)$ formula $\aformula_{\apcpinst}$ such that $\apcpinst$ has
a solution iff $\aformula_{\apcpinst}$ is satisfiable.  To do so, we adapt a proof technique
from~\cite{BDMSS06:journal,Kara&Schwentick&Zeume10} but we need to provide substantial
changes in order to fit our logic. Moreover, none of the results in~\cite{BDMSS06:journal,Kara&Schwentick&Zeume10}
allow to derive our main undecidability result since we use neither
past-time temporal operators nor past obligations of the form 
$\poblieq{\pair{\avariable}{\avariable'}}{\pair{\avariablebis}{\avariablebis'}}{\aformulater}$.

Let $\bar
\aalphabet = \set{\bar \aletter \mid \aletter \in \aalphabet}$ be a
disjoint copy of $\aalphabet$. For convenience, we assume that each
position of a $\mainlogic$ model is labelled by a letter from
$\aalphabet \cup \bar \aalphabet$ (these labels can be easily encoded
as equivalence classes of some extra variables). For such a model
$\amodel$, let $\amodel_{\aalphabet}$ (resp.  $\amodel_{\bar
\aalphabet}$) be the model obtained by restricting $\amodel$ to
positions labelled by $\aalphabet$ (resp. $\bar \aalphabet$). If
$u_{i_1} \dotsb u_{i_m} = v_{i_1} \dotsb v_{i_m}$ is a solution to
$\apcpinst$, the idea is to construct a $\mainlogic$ model whose
projection to $\aalphabet \cup \bar\aalphabet$ is $u_{i_1}
\overline{v_{i_1}} \dotsb u_{i_m} \overline{v_{i_m}}$. To check that
such a model $\amodel$ actually represents a solution, we will write
$\dcltltup(\mynext,\until)$ formulas to say that ``for all $j$, if the
$j$\textsuperscript{th} position of $\amodel_{\aalphabet}$ is labelled
by $\aletter$, then the $j$\textsuperscript{th} position of
$\amodel_{\bar \aalphabet}$ is labelled by $\bar \aletter$''.

The main difficulty is to get a handle on the $j$\textsuperscript{th}
position of $\amodel_{\aalphabet}$. The difficulty arises since the
$\mainlogic$ model $\amodel$ is an interleaving of
$\amodel_{\aalphabet}$ and $\amodel_{\bar \aalphabet}$. This is
handled by using two variables $\avariable$ and $\avariablebis$.
Positions $1$ and $2$ of $\amodel_{\aalphabet}$ will have the same
value of $\avariable$, positions $2$ and $3$ will have the same value
of $\avariablebis$, positions $3$ and $4$ will have the same value of
$\avariable$ and so on.  Generalising, odd positions of
$\amodel_{\aalphabet}$ will have the same value of $\avariable$ as in
the next position of $\amodel_{\aalphabet}$ and the same value of
$\avariablebis$ as in the previous position of $\amodel_{\aalphabet}$.
Even positions of $\amodel_{\aalphabet}$ will have the same value of
$\avariable$ as in the previous position of $\amodel_{\aalphabet}$ and
the same value of $\avariablebis$ as in the next position of
$\amodel_{\aalphabet}$. 
These constraints allow us to chain the positions of $\amodel$
corresponding to $\amodel_{\aalphabet}$. 
To easily identify odd and even positions, an
additional label is introduced at each position, which can be $\odd$
or $\even$. 
These labels encoding the parity status of the positions in $\amodel_{\aalphabet}$
will be helpful to write formulae.
The sequence $\amodel_{\aalphabet}$ looks as follows.
\begin{align}\label{eq:dv-chaind}
  \left[
  \begin{array}{c}
    \avariable \\  \text{Odd/Even} \\ \text{Letter} \\\avariablebis
  \end{array}
  \right]:
  \vertuple{\adatum_{1}}{\oddin}{\aletter_{1}}{\adatum_{1}'}
  \vertuple{\adatum_{1}}{\even}{\aletter_{2}}{\adatum_{2}'}
  \vertuple{\adatum_{3}}{\odd}{\aletter_{3}}{\adatum_{2}'}
  \vertuple{\adatum_{3}}{\even}{\aletter_{4}}{\adatum_{4}'}
  \cdots
  \vertuple{\adatum_{m-1}}{\odd}{\aletter_{m-1}}{\adatum_{m-2}'}
  \vertuple{\adatum_{m-1}}{\evenen}{\aletter_{m}}{\adatum_{m}'} \tag{$\star$}
\end{align}
The $\adatum_i$'s and $\adatum'_i$'s are data values for the variables $\avariable$ and $\avariablebis$ respectively.
Each label is actually a pair in $\aalphabet \times \set{\odd, \oddin, \even, \evenen}$
to record the letter from $\aalphabet$ and the parity status of the position.
The letter $\odd$ identifies an odd position (a sequence starts here from the first position)
and the first position is identified by the special letter $\oddin$. Similarly, 
the letter $\even$ identifies an even position
and the last position is identified by the special letter $\evenen$. 
Similarly, in order to define the sequence $\amodel_{\bar{\aalphabet}}$, we consider
the letters $\bar{\odd}$, $\bar{\even}$, $\bar{\oddin}$ and $\bar{\evenen}$ with analogous intentions. 
So, by way of example, each position of $\amodel_{\aalphabet}$ is labelled by two letters: a letter in $\aalphabet$
and a letter specifying the parity of the position. 

We assume that the atomic formula $\aletter$  ($\aletter\in \aalphabet$)
is true at some position $j$ of a model $\amodel$ iff the position $j$
is labelled by the letter $\aletter$. We denote $\bigvee_{\aletter \in
\aalphabet} \aletter$ by $\aalphabet$. For a word $\aword \in
\aalphabet^{*}$, we denote by $\aformula_{\aword}^{i}$ the
\mainlogic{} formula that ensures that starting from $i$ positions to
the right of the current position, the next $\length{\aword}$
positions are labelled by the respective letters of $\aword$ (e.g.,
$\aformula_{\aletter \aletterbis \aletterbis}^{3} \eqdef
\mynext^{3} \aletter \land \mynext^{4} \aletterbis \land \mynext^{5}
\aletterbis$). We enforce a model of the form 
($\star$)
above through the following formulas.
\begin{enumerate}
  \item \label{cond:instanced} The projection of the labeling of
    $\sigma$ on to $\Sigma \cup \bar \Sigma$ is in $(u_{1}
    \overline{v_{1}} + u_{2}\overline{v_{2}}) \set{ u_i \bar v_i
    \mid i \in \interval{1}{n}}^*$: to facilitate writing this condition in
    \mainlogic{}, we introduce two new variables $\avariablebp,
    \avariableep$ such that at any position, they have the same value
    only if that position is not the starting point of some pair $u_{i}
    \overline{v_{i}}$.
    \begin{align*}
      &\avariablebp \noblieqlocal \avariableep \land \bigvee_{i\in
      \set{1,2}}
      (\aformula_{u_{i}}^{0}\land
      \aformula_{\overline{v_{i}}}^{\length{u_{i}} + 1} \land
      (\mynext^{\length{u_{i} v_{i}} + 1} \top \Rightarrow
      \mynext^{\length{u_{i} v_{i}} + 1} \avariablebp \noblieqlocal
      \avariableep)) \\ & \land \always
      \left( \avariablebp \oblieqlocal \avariableep \lor \bigvee_{i
      \in [1,n]}(\aformula_{u_{i}}^{0}\land
      \aformula_{\overline{v_{i}}}^{\length{u_{i}} + 1}) \land
      (\mynext^{\length{u_{i} v_{i}} + 1} \top \Rightarrow
      \mynext^{\length{u_{i} v_{i}} + 1} \avariablebp \noblieqlocal
      \avariableep)\right).
    \end{align*}
  \item\label{cond:chaind}
    \begin{enumerate}
      \item \label{cond:atmost2d} For every data value $\adatum$, there are
	at most two positions $i,j$ in $\sigma_\Sigma$  with
	$\sigma_\Sigma(i)(\avariable) =
	\sigma_\Sigma(j)(\avariable)= \adatum$ or $\sigma_\Sigma(i)
	(\avariablebis) = \sigma_\Sigma(j)(\avariablebis)= \adatum$.
	\begin{align*}
	  &\always \left( \aalphabet
	  \Rightarrow
	  (\neg \oblieq{\avariable}{\avariable}{\aalphabet}) \lor
	  (\oblieq{\avariable}{\avariable}{\aalphabet \land
	  (\neg \oblieq{\avariable}{\avariable}{\aalphabet}}))
	  \right)\\ \land & \always \left( \aalphabet
	  \Rightarrow
	  (\neg \oblieq{\avariablebis}{\avariablebis}{\aalphabet}) \lor
	  (\oblieq{\avariablebis}{\avariablebis}{\aalphabet \land
	  (\neg \oblieq{\avariablebis}{\avariablebis}{\aalphabet}}))
	  \right).
	\end{align*}
	The same condition for $\sigma_{\bar\Sigma}$, enforced with a
	formula similar to the one above.
      \item\label{cond:proj:LRd}  $\amodel_{\aalphabet}$ projected
	onto $\set{\odd,\oddin,	\even, \evenen}$ is in $\oddin
	(\even\odd)^{*} \evenen$.
	\begin{align*}
	  &\oddin \land \always \left( (\odd \lor \oddin) \Rightarrow
	  \bar \aalphabet \until (\aalphabet \land (\even \lor
	  \evenen)) \right) \land \always \left( \even \Rightarrow
	  \bar \aalphabet \until (\aalphabet \land \odd) \right)\\ &
	  \land \always (\evenen \Rightarrow \lnot \mynext \sometimes
	  \aalphabet).
	\end{align*}
        Note that $\oddin$ [resp. $\evenen$] is useful to identify the first odd 
        position [resp. last even position]. 
	$\amodel_{\bar \aalphabet}$ projected onto $\set{\bar \odd,
	\overline{\oddin}, \bar \even, \overline{\evenen}}$ is in
	$\overline{\oddin}(\bar \even \bar \odd)^{*}
	\overline{\evenen}$.
	\begin{align*}
	  &\aalphabet \until  \overline{\oddin} \land \always \left(
	  (\bar \odd \lor \overline{\oddin}) \Rightarrow
	  \aalphabet \until (\bar \aalphabet \land (\bar \even \lor
	  \overline{\evenen}))
	  \right) \land \always \left( \bar \even
	  \Rightarrow \aalphabet \until (\bar \aalphabet \land \bar \odd)
	  \right)\\ & \land \always (\overline{\evenen} \Rightarrow
	  \lnot \mynext \top). 
	\end{align*}
      \item \label{cond:LtoRd} For every position $i$ of
	$\amodel_{\aalphabet}$ labelled by $\odd$ or $\oddin$, there
	exists a future position $j > i$ of $\amodel_{\aalphabet}$
	labelled by $\even$ or $\evenen$
	so that $\sigma_{\aalphabet}(i) (\avariable) =
	\sigma_{\aalphabet}(j) (\avariable)$.
	\begin{align*}
	  \always \left( (\odd \lor \oddin)
	  \Rightarrow \oblieq{\avariable}{\avariable}{\aalphabet \land
	  (\even \lor \evenen)} \right).
	\end{align*}
        For every position $i$ of
	$\amodel_{\bar \aalphabet}$ labelled by $\bar \odd$ or
	$\overline{\oddin}$, there
	exists a future position $j > i$ of $\amodel_{\bar \aalphabet}$
	labelled by $\bar \even$ or $\overline {\evenen}$
	so that $\sigma_{\bar \aalphabet}(i) (\avariable) =
	\sigma_{\bar \aalphabet}(j) (\avariable)$ (enforced with a
	formula similar to the one above).
      \item \label{cond:RtoLd} For every position $i$  of
	$\amodel_{\aalphabet}$  labelled by $\even$, there exists a
	future position $j > i$ of $\amodel_{\aalphabet}$ labelled by
	$\odd$ so that $\sigma_{\aalphabet}(i) (\avariablebis) =
	\sigma_{\aalphabet}(j) (\avariablebis)$.
	\begin{align*}
	  \always \left( \even \Rightarrow
	  \oblieq{\avariablebis}{\avariablebis}{\aalphabet \land \odd} \right).
	\end{align*}
        For every position $i$  of
	$\amodel_{\bar \aalphabet}$  labelled by $\bar\even$, there exists a
	future position $j > i$ of $\amodel_{\bar \aalphabet}$ labelled by
	$\bar \odd$ so that $\sigma_{\bar \aalphabet}(i) (\avariablebis) =
	\sigma_{\bar \aalphabet}(j) (\avariablebis)$ (enforced with a
	formula similar to the one above).
    \end{enumerate}
  \item\label{cond:endpointsd} For any position $i$ of $u_{i_{1}}$,
    the corresponding position $j$ in $\overline{v_{i_{1}}} \dotsb
    \overline{v_{i_{m}}}$ (which always happens to be in $v_{i_{1}}$, since $\length{u_{i_{1}}} \le 2$ and
    $\length{v_{i_{1}}} \ge 3$) should satisfy $\sigma(i)(\avariable)
    = \sigma(j)(\avariable)$ and $\sigma(i)(\avariablebis) =
    \sigma(j)(\avariablebis)$. In addition, position $i$ is labelled
    with $\aletter \in \aalphabet$ iff position $j$ is labelled with
    $\bar \aletter \in \bar \aalphabet$.
      \begin{align*}
	& \left(\bigvee_{\aletter \in \aalphabet} (\oddin \land \aletter)
	\Rightarrow
	\oblieq{\langle \avariable, \avariablebis \rangle}{\langle
	\avariable, \avariablebis \rangle}{\overline{\oddin} \land
	\overline{\aletter}} \right)\\ & \land 
        \mynext \aalphabet \Rightarrow
	\bigvee_{\aletter \in\aalphabet}
        \mynext \aletter \wedge \mynext^3 \overline{\aletter} \wedge
        \mynext(\avariable \oblieqlocal \mynext^2 \avariable \wedge 
        \avariablebis \oblieqlocal \mynext^2 \avariablebis).
                				      \end{align*}
    \item\label{cond:allpointsd}
      For any position $i$ of $\amodel$ with
      $2\length{u_{i_{1}}} < i < \length{u_{i_{1}} \dotsb u_{i_{m}}}$,
      if it is labeled with $\bar a \in \bar \Sigma$, there is a
      future position $j > i$ labeled with $a \in \Sigma$ such that
      $\sigma(i)(\avariable) = \sigma(j)(\avariable)$ and
      $\sigma(i)(\avariablebis) = \sigma(j)(\avariablebis)$. The
      following formula assumes that $\length{u_{i_{1}}} \le 2$, as it
      is in \pcpdir{}.
      \begin{align*}
	&\mynext \aalphabet \Rightarrow \aalphabet \until
	\left(\mynext^{2}\always \bigvee_{\bar \aletter \in \bar \aalphabet} \left(
	(\bar \aletter \land \lnot \overline{\evenen}) \Rightarrow
	(\avariable, \avariablebis) \oblieqlocal \langle \aletter?
	\rangle (\avariable, \avariablebis)\right)\right) \quad
	\text{(} \mynext \aalphabet
	\text{ is true when } \length{u_{i_{1}}} = 2 \text{)}\\
	& \land \mynext \bar \aalphabet \Rightarrow \aalphabet \until
	\left(\mynext\always \bigvee_{\bar \aletter \in \bar \aalphabet} \left(
	(\bar \aletter \land \lnot \overline{\evenen}) \Rightarrow
	(\avariable, \avariablebis) \oblieqlocal \langle \aletter?
	\rangle (\avariable, \avariablebis)\right)\right) \quad
	\text{(} \mynext \bar \aalphabet
	\text{ is true when } \length{u_{i_{1}}} = 1 \text{)}
      \end{align*}
    \item\label{cond:endpointsdfi} If $i$ and $j$ are the last
      positions labeled with $\aalphabet$ and $\bar \aalphabet$
      respectively, then $\sigma(i)(\avariable) =
      \sigma(j)(\avariable)$ and $\sigma(i)(\avariablebis) =
      \sigma(j)(\avariablebis)$. In addition, position $i$ is labelled
      with $\aletter \in \aalphabet$ iff position $j$ is labelled with
      $\bar \aletter \in \bar \aalphabet$.
      \begin{align*}
	\always \left(\bigvee_{\aletter \in \aalphabet} (\evenen \land
	\aletter) \Rightarrow \oblieq{\langle \avariable,
	\avariablebis \rangle}{\langle \avariable, \avariablebis
	\rangle}{\overline{\evenen} \land \overline{\aletter}} \right).
      \end{align*}
\end{enumerate}

\noindent Given an instance $\apcpinst$ of \pcpdir{}, the required formula
$\aformula_{\apcpinst}$ is the conjunction of all the formulas above.
\begin{lem}
  \label{lem:PcpSolFormSat}
  Given an instance $\apcpinst$ of \pcpdir{}, $\aformula_{\apcpinst}$
  is satisfiable iff $\apcpinst$ has a solution.
\end{lem}

\begin{proof}
  Suppose $u_{i_1} \dotsb u_{i_m} = v_{i_1} \dotsb v_{i_m}$ is a
  solution of $\apcpinst$. It is routine to check that, with this
  solution, a model satisfying $\aformula_{\apcpinst}$ can be built.

  Now suppose that $\aformula_{\apcpinst}$ has a satisfying model
  $\amodel$. From condition \eqref{cond:instanced}, we get a sequence
  $u_{i_1}\overline{v_{i_{1}}} \dotsb u_{i_m} \overline{v_{i_{m}}}$.
  It is left to prove that $u_{i_{1}} \dotsb u_{i_{m}} = v_{i_{1}}
  \dotsb v_{i_{m}}$.

  Let $\amodel_{\aalphabet}$ (resp. $\amodel_{\bar \aalphabet}$) be
  the model obtained from $\amodel$ by restricting it to positions
  labeled with $\aalphabet$ (resp. $\bar \aalphabet$). Construct a
  directed graph whose vertices are the positions of
  $\amodel_{\aalphabet}$ and there is an edge from $i$ to $j$ iff $i <
  j$ and $\amodel_{\aalphabet} ( i ) ( \avariable) =
  \amodel_{\aalphabet} ( j ) (\avariable)$ or $\amodel_{\aalphabet} (
  i ) ( \avariablebis) = \amodel_{\aalphabet} ( j ) (\avariablebis)$.
  We claim that the set of edges of this directed graph represents the
  successor relation induced on $\amodel_{\aalphabet}$ by $\amodel$.
  To prove this claim, we first show that all positions have indegree
  $1$ (except the first one, which has indegree $0$) and that all
  positions have outdegree $1$ (except the last one, which has outdegree $0$). Indeed, condition~\eqref{cond:atmost2d} 
  ensures that
  for any position, both the indegree and outdegree are at most $1$.
  From conditions~\eqref{cond:proj:LRd},~\eqref{cond:LtoRd} and~\eqref{cond:RtoLd}, each position (except the last one) has 
  outdegree at least $1$. Hence, all positions except the last one have
  outdegree exactly $1$. By the definition of the set of edges, the
  last position has outdegree $0$. If more than one position has indegree $0$, it will force some 
  other position to have indegree more
  than $1$, which is not possible. By the definition of the set of
  edges, the first position has indegree $0$ and hence, all other
  positions have indegree exactly $1$. To finish proving the claim
  (that the set of edges of our directed graph represents the
  successor relation induced on $\amodel_{\aalphabet}$ by $\amodel$),
  we will now prove that at any position of $\amodel_{\aalphabet}$
  except the last one, the outgoing edge goes to the successor
  position in $\amodel_{\aalphabet}$. If this is not the case, let
  $i$ be the last position where this condition is violated. The
  outgoing edge from $i$ then goes to some position $j > i + 1$. Since
  the outgoing edges from each position between $i + 1$ and $j - 1$
  go to the respective successors, position $j$ will have indegree
  $2$, which is not possible.

  Next we will prove that there cannot be two positions of
  $\amodel_{\aalphabet}$ with the same value for variables
  $\avariable$ and $\avariablebis$. Suppose there were two such
  positions and the first one is labeled $\odd$ (the argument for
  $\even$ is similar). If the second position is also labeled
  $\odd$, then by condition \eqref{cond:LtoRd}, there is
  at least one position labeled $\even$ with the same value for
  variable $\avariable$, so there are three positions with the same
  value for variable $\avariable$, violating condition
  \eqref{cond:atmost2d}.  Hence, the second position must be labelled
  $\even$. Then by condition \eqref{cond:RtoLd}, there is a position
  after the second position with the same value for variable
  $\avariablebis$. This implies there are three positions with the
  same value for variable $\avariablebis$, again violating condition
  \eqref{cond:atmost2d}.  Therefore, there cannot be two positions of
  $\amodel_{\aalphabet}$ with the same value for variables
  $\avariable$ and $\avariablebis$.

  Finally, we prove that for every position $i$ of $\amodel_{\bar
  \aalphabet}$, if $\overline{\aletter_{i}}$ is its label, then the
  unique position in $\amodel_{\aalphabet}$ with the same value of
  $\avariable$ and $\avariablebis$ is position $i$ of
  $\amodel_{\aalphabet}$ and carries the label $\aletter_{i}$. For $1
  \le i \le \length{u_{i_{1}}}$, this follows from condition
  \eqref{cond:endpointsd}. For $i = \length{u_{i_{1}} \dotsb
  u_{i_{m}}}$, this follows from condition \eqref{cond:endpointsdfi}.
  The rest of the proof is by induction on $i$. The base case is
  already proved since $\length{u_{i_{1}}} \ge 1$. For the induction
  step, assume the result is true for all positions of $\amodel_{\bar
  \aalphabet}$ up to position $i$. Suppose position $i$ of
  $\amodel_{\bar \aalphabet}$ is labeled by $\bar \odd$ (the case of
  $\bar \even$ is symmetric). Then by the induction hypothesis and~\eqref{cond:proj:LRd}, position $i$ of 
  $\amodel_{\aalphabet}$ is
  labeled by $\odd$ (or $\oddin$, if $i=1$).  By condition~\eqref{cond:LtoRd} and the definition of edges in the 
  directed graph
  that we built, position $i + 1$ of $\amodel_{\bar \aalphabet}$
  (resp. $\amodel_{\aalphabet}$) has same value of $\avariable$ as
  that of position $i$.  We know from the previous paragraph that
  there is exactly one position of $\amodel_{\aalphabet}$ with the
  same value of $\avariable$ and $\avariablebis$ as that of position
  $i + 1$ in $\amodel_{\bar \aalphabet}$. This position in
  $\amodel_{\aalphabet}$ cannot be $i$ or before due to induction
  hypothesis. If it is not $i+1$ either, then there will be three
  positions of $\amodel_{ \aalphabet}$ with the same value of
  $\avariable$, which violates condition \eqref{cond:atmost2d}. Hence,
  the position of $\amodel_{\aalphabet}$ with the same value of
  $\avariable$ and $\avariablebis$ as that of position $i + 1$ in
  $\amodel_{\bar \aalphabet}$ is indeed $i + 1$ and it carries the
  label $\aletter_{i}$ by condition \eqref{cond:allpointsd}.
\end{proof}
As a conclusion,  the satisfiability problem for $\mainlogic_{vec}(\mynext,\until)$ is undecidable.

\ifLONG
The reduction from \SAT{$\mainlogic$} to \SAT{$\mainlogic^{\top}$} can be easily adapted to lead to
a reduction from \SAT{$\mainlogic_{vec}$}   to  \SAT{$\mainlogic_{vec}^{\top}$}, 
whence \SAT{$\mainlogic_{vec}^{\top}$} is undecidable too. 
\else
Furthermore, our proof can be adapted to show that \SAT{$\mainlogic_{vec}^{\top}$}
is also undecidable. 
\fi

\makeatletter{}\section{Implications for Logics on Data Words}
\label{section-implications}
A data word is an element of $(\Sigma \times \D)^*$, where $\Sigma$ is a finite alphabet and $\D$ is an infinite domain. We focus here on first-order logic with two variables, and on a temporal logic.

\subsection{Two-variable Logics}
We study a fragment of \emsot{} on data words, and we show that it has a satisfiability problem in 3\expspace, as a consequence of our results on the satisfiability for \mainlogic. The satisfiability problem for \emsot{} is known to be decidable, equivalent to the reachability problem for VASS \cite{BDMSS06:journal}, with no known primitive-recursive algorithm. Here we show a large fragment with elementary complexity.

Consider the fragment of \emsot{} ---that is, first-order logic with two variables, with a prefix of 
existential quantification over monadic relations--- where all formulas are of the form $\exists X_1, 
\dotsc, X_n ~ \varphi$ with
\begin{align*}
  \varphi &\coloneqq \textit{atom} \mid \lnot \varphi \mid 
                \varphi \land \varphi \mid \varphi \lor \varphi \mid \\
&~~~~~\,\exists x \exists y ~
\varphi \mid \forall x \forall y ~ \varphi \mid \forall x \exists y ~ (x \leq y ~ \land ~
  \varphi), \text{ where}\\
   \textit{atom} &\coloneqq \zeta = \zeta' \mid \zeta \neq \zeta' \mid
   \zeta \sim \zeta' \mid \zeta < \zeta' \mid \zeta \leq \zeta' \mid\\
&~~~~~
   +1(\zeta,\zeta') \mid X_i(\zeta) \mid a(\zeta)
 \end{align*}
for any $a \in \Sigma$, $i \in \interval{1}{n}$, 
and $\zeta, \zeta' \in \set{x,y}$. 
The relation $x < y $ tests that the position $y$ appears after $x$ in
the word; $+1(x,y)$ tests that $y$ is the next position to $x$; and $x
\sim y$ tests that positions $x$ and $y$ have the same data value.
We call this fragment \femsot. In fact, \femsot{} captures \emsotw{} (i.e., all regular languages on the finite labeling of the data word).\footnote{Indeed, note that one can easily
 test whether a word is accepted by a finite automaton with the help of a monadic relation $X_{\textit{fst}}$ that holds only at the first position. Further, the property of $X_{\textit{fst}}$ can be expressed in \femsot{} as $(\exists y ~.~ X_{\textit{fst}}(y)) \land (\forall x \forall y ~.~ x<y \rightarrow \lnot X_{\textit{fst}}(y))$.}
However, it seems to be strictly less expressive than \emsot, since \femsot{} does not appear to be able to express the property \emph{there are exactly two occurrences of every data value}, which can be easily expressed in \emsot{}. Yet, it can express \emph{there are at most two occurrences of every data value} (with $\exists X ~ \forall x \forall y .~ x \sim y \land x < y \rightarrow X(x) \land \lnot X(y)$), and \emph{there is exactly one occurrence of every data value}.
For the same reason, it would neither capture $\emsotraneq$.

By an exponential reduction into 
\SAT{$\mainlogic$},
we obtain that \femsot{} is decidable in elementary time.

\begin{prop}\label{prop:sat-forward-EMSOT-2EXPSPACE}
  The satisfiability problem for \femsot{} is in 3\expspace.
\end{prop}
\makeatletter{}\ifLONG
\begin{proof}
Through a standard translation we can bring any formula of $\femsot$ into a formula of the form
\[
\varphi = \exists X_1, \dotsc, X_n ~ \left( \forall x \forall y ~ \chi ~\land~ \bigwedge_k \forall x \exists y ~ (x \leq y \land \psi_k) \right)
\]
that preserves satisfiability, where $\chi$ and all $\psi_k$'s are 
quantifier-free formulas, and there are no tests for labels. Furthermore, this is a polynomial-time translation. This translation is just the Scott normal form of $\emsot$ \cite{BDMSS06:journal} adapted to $\femsot$, and can be done in the same way.

We now give an exponential-time translation $tr : \femsot \to \mainlogic$. For any formula $\varphi$ of \femsot, $tr(\varphi)$ is an equivalent (in the sense of satisfiability) $\mainlogic$ formula,
 whose satisfiability can be tested in $\twoexpspace$ (Corollary~\ref{cor:mainlogic-twoexpspace}). 
 This yields an upper bound of $3\expspace$ for $\femsot$.

The translation makes use of: a distinguished variable $\avariable$ that encodes the data values of any data word satisfying $\varphi$; variables $\avariable_0, \dotsc, \avariable_n$ that are used to encode the monadic relations $X_1, \dotsc, X_n$; and a variable $\avariable_{\textit{prev}}$ whose purpose will be explained later on.
We give now the translation.
To translate $\forall x \forall y ~ \chi$, we first bring the formula to a form 
\begin{align}\label{eq:normal-form-for-chi}
  \bigwedge_{m \in M} ~ \lnot \exists x \exists y
  ~ \big(~ x \leq y~\land~\chi_{m} ~\land~ \chi_{m}^x ~\land~
  \chi_{m}^y ~\big),
\end{align}
where $M\subseteq \N$, and every $\chi_{m}^x$ (resp.\ $\chi_{m}^y$) is a conjunction of (negations of) atoms of 
monadic relations on $x$ (resp.\ $y$); 
and $\chi_{m} = \mu \land \nu$  where $\mu \in \set{x{=}y, +1(x,y), \lnot (+1(x,y) \lor x {=} y)}$ and 
$\nu \in \set{x \sim y, \lnot (x \sim y)}$. 
\begin{clm}
  $\forall x \forall y ~ \chi$ can be translated into an equivalent formula of the 
   form~\eqref{eq:normal-form-for-chi} 
    in exponential time.
\end{clm}
 \begin{proof}
As an example, if \[\chi = (x \neq y ~\Rightarrow~ \lnot (x \sim y) \lor X(x) \lor X(y)),\] then
the corresponding formula would be
\begin{align*}
\bigwedge_\mu \lnot\exists x \exists y ~ \big(x \leq y \land \mu \land x \sim y \land\lnot X(x) \land \lnot X(y)\big).
\end{align*}
for all $\mu \in \{+1(x,y), \lnot (+1(x,y) \lor x {=} y)\}$.
We can bring the formula into this normal form in exponential time. To
this end, we can first bring $\chi$ into CNF, $\chi = \bigwedge_{i \in
  I} \bigvee_{j \in J_i} \nu_{ij}$, where every $\nu_{ij}$ is an atom
or a negation of an atom. Then,
\begin{align*}
  \forall x \forall y ~ \chi ~~\equiv~~ \forall x \forall y ~\bigwedge_{i
    \in I} \bigvee_{j \in J_i} \nu_{ij} ~~\equiv~~ \bigwedge_{i \in I}
  \forall x \forall y ~ \bigvee_{j \in J_i} \nu_{ij} ~~\equiv~~
  \bigwedge_{i \in I} \lnot \exists x \exists y ~ \bigwedge_{j \in
    J_i} \lnot \nu_{ij}
\end{align*}
Let $\nu_{ij}^{x \leftrightarrow y}$ be $\nu_{ij}$ where $x$ and $y$
are swapped.  Note that $\exists x \exists y ~ \bigwedge_{j \in J_i}
\lnot \nu_{ij}$ is equivalent to $\exists x \exists y ~ \bigwedge_{j
  \in J_i} \lnot \nu_{ij}^{x \leftrightarrow y}$. Now for every $i \in
I$, let
\begin{align*}
  \mu_{i,1} &= x \leq y \land
  \bigwedge_{j \in J_i} \lnot \nu_{ij} &
 \mu_{i,2} &= x \leq y \land \bigwedge_{j \in J_i} \lnot
  \nu_{ij}^{x \leftrightarrow y}
.
\end{align*}
Note that $\exists x \exists y ~ \mu_{i,1} \lor \exists x
\exists y ~ \mu_{i,2}$  is equivalent to $\exists x \exists y ~ \bigwedge_{j
  \in J_i} \lnot \nu_{ij}$. Hence,
\[
\bigwedge_{i \in I} \lnot \bigvee_{j \in \set{1,2}} \exists x \exists
y ~ (x \leq y \land \mu_{i,j}) ~~\equiv~~ \bigwedge_{i \in I} \bigwedge_{j
  \in \set{1,2}} \lnot \exists x \exists y ~ (x \leq y \land
\mu_{i,j})
\]
is equivalent to $\forall x \forall y ~ \chi$. Finally, every $\mu_{i,j}$ can be easily split into a conjunction of three formulas (one of binary relations, one of unary relations on $x$, and one of unary relations on $y$), thus obtaining a formula of the 
 form ($\star$).
This procedure takes polynomial time once the CNF normal form is obtained, and it then takes exponential time in the worst case.
 \end{proof}

We define 
$tr(\chi_{m}^x)$ 
as the conjunction of all the formulas $\avariable_0 \approx \avariable_i$ so that 
$X_i(x)$ is a conjunct of 
$\chi_{m}^x$,
and all the formulas $\lnot (\avariable_0 \approx \avariable_i)$ so that 
$\lnot X_i(x)$ is a conjunct of 
 $\chi_{m}^x$;  
we do similarly for 
$tr(\chi_{m}^y)$.
If $\mu = +1(x,y)$ and $\nu = x \sim y$ we translate 
 \[
 tr\left(\exists y ~ \big(~ x \leq y ~\land~\chi_{m} ~\land~ \chi_{m}^x ~\land~ \chi_{m}^y ~\big)\right) ~=~ 
 \avariable \approx \X\avariable~ \land~ tr(\chi_{m}^x) ~\land~ 
 \X tr(\chi_{m}^y)
 .
 \]
If $\mu = (x = y)$ and $\nu = x \sim y$ we translate 
 \[
 tr\left(\exists y ~ \big(~ x \leq y ~\land~\chi_{m} ~\land~ \chi_{m}^x ~\land~ 
 \chi_{m}^y ~\big)\right) ~=~  tr(\chi_{m}^x) ~\land~ 
 tr(\chi_{m}^y)
 .
 \]
We proceed similarly for $\mu = +1(x,y)$, $\nu = \lnot(x \sim y)$; and the translation is of course $\bot$ (false) if $\mu = (x {=} y)$, $\nu = \lnot(x {\sim} y)$.  The difficult cases are the remaining ones. Suppose 
$\mu = \lnot (+1(x,y) \lor x {=} y)$, $\nu = x {\sim} y$.
In other words, $x$ is at least two positions before $y$, and they have the same data value. 
Observe that the formula 
 $tr(\chi_{m}^x) \land \oblieq{\avariable}{\avariable}{tr(\chi_{m}^y)}$
 does not 
encode precisely this case, as it would correspond to a weaker condition $x < y  \land x \sim y$. In order to 
properly translate this case we make use of the variable $\avariable_{\textit{prev}}$, ensuring that it always has 
the data value of the variable $\avariable$ in the previous position
\[
    \textit{\textit{prev}} ~=~ \G \big( ~\X \top ~\Rightarrow~ \avariable \approx \X \avariable_{\textit{prev}} ~\big)
.
\]
 We then define $tr\left(\exists y ~ \big(~ x \leq y~\land~\chi_{m} ~\land~ \chi_{m}^x ~\land~ 
 \chi_{m}^y ~\big)\right)$ as
 \[
  tr(\chi_{m}^x) ~\land~ \oblieqz{\avariable}{\avariable_{\textit{prev}}} {
 \oblieq{\avariable_{\textit{prev}}}{\avariable}{tr(\chi_{m}^y)} }
 .
 \]
Note that by nesting twice the future obligation we ensure that the target position where $tr(\chi_{m}^y)$ must hold is at a distance of at least two positions.
For $\nu = \lnot(x {\sim}y)$ we produce a similar formula, replacing the innermost appearance of $\approx$ with $\not\approx$ in the formula above.
We then define $tr(\forall x \forall y ~ \chi)$ as 
\[
\textit{\textit{prev}} ~\land~ \bigwedge_{m \in M} ~ \lnot ~\F ~ tr\big(\exists y ~ (x \leq y 
~\land~\chi_{m} \land \chi_{m}^x \land \chi_{m}^y )\big)
.
\]

To translate $\forall x \exists y ~ (x \leq y \land \psi_k)$ we proceed in a similar way. As before, we bring $x \leq y \land \psi_k$ into the form
$
\bigvee_{m \in M} ~ x \leq y \land \chi_{m} \land \chi_{m}^x \land \chi_{m}^y,
$
in exponential time. We then define $tr(\forall x \exists y ~ (x \leq y \land \psi_k))$ as
\[
\textit{\textit{prev}} ~\land~  \G~ \bigvee_{m \in M}  tr\big(\exists y ~ (x \leq y \land \chi_{m} \land \chi_{m}^x \land \chi_{m}^y )\big)
.
\]
Thus, 
\[
tr(\varphi) ~=~ tr\big(\forall x \forall y ~ \chi\big)
~\land~
\bigwedge_k tr\big(\forall x \exists y ~ (x \leq y \land \psi_k)\big).
\]

\medskip

\noindent One can show that the translation $tr$ defined above preserves satisfiability. 
More precisely, it can be seen that:
\begin{enumerate}
\item Any data word whose data values are the
  $\avariable$-projection from a model satisfying $tr(\varphi)$,
  satisfies $\varphi$; and, conversely,
\item for any data word  satisfying $\varphi$ with a given assignment for $X_1, \dotsc, X_n$ to the word positions, and for any model $\sigma$ such that
  \begin{itemize}
  \item $\sigma$ has the same length as the data word,
  \item for every position $i$, $\sigma(i)(\avariable)$ is the data value of position $i$ from the data word, and $\sigma(i)(\avariable_0) = \sigma(i)(\avariable_j)$ iff $X_j$ holds at position $i$ of the data word, and
  \item for every position $i>0$, $\sigma(i)(\avariable_{\textit{prev}}) = \sigma(i-1)(\avariable)$,
  \end{itemize}
we have that $\sigma \models tr(\varphi)$.
\end{enumerate}\smallskip

\noindent By Corollary~\ref{cor:mainlogic-twoexpspace} we can decide the satisfiability of the translation in \twoexpspace, and since the translation is exponential, this gives us a $3\expspace$ upper bound for the satisfiability of $\femsot$.
\end{proof}
\else
\begin{proof}[Proof sketch]
Through a standard translation we can bring any formula of $\femsot$ into a formula of the form
\[
\varphi = \exists X_1, \dotsc, X_n ~ \left( \forall x \forall y ~ \chi ~\land~ \bigwedge_k \forall x \exists y ~ (x \leq y \land \psi_k) \right)
\]
that preserves satisfiability, where $\chi$ and all $\psi_k$'s are 
quantifier-free formulas, and there are no tests for labels. Furthermore, this is a polynomial-time translation. This translation is just the Scott normal form of $\emsot$ \cite{BDMSS06:journal} adapted to $\femsot$, and can be done in the same way. 

We now give an exponential-time translation $tr : \femsot \to \mainlogic$. For any formula $\varphi$ of \femsot, $tr(\varphi)$ is an equivalent (in the sense of satisfiability) $\mainlogic$ formula.

The translation makes use of: a distinguished variable $\avariable$ that encodes the data values of any data word satisfying $\varphi$; variables $\avariable_0, \dotsc, \avariable_n$ that are used to encode the monadic relations $X_1, \dotsc, X_n$; and a variable $\avariable_{\textit{prev}}$ whose purpose will be explained later on.
We give now the translation.
To translate $\forall x \forall y ~ \chi$, we first bring the formula to a form 
\begin{align}
\tag{$\star$}
  \bigwedge_{m \in M} ~ \lnot \exists x \exists y
  ~ \big(~ x \leq y~\land~\chi_{m} ~\land~ \chi_{m}^x ~\land~
  \chi_{m}^y ~\big),
\end{align}
where every $\chi_{m}^x$ (resp.\ $\chi_{m}^y$) is a conjunction of (negations of) atoms of 
monadic relations on $x$ (resp.\ $y$); 
and $\chi_{m} = \mu \land \nu$  where $\mu \in \set{x{=}y, +1(x,y), \lnot (+1(x,y) \lor x {=} y)}$ and 
$\nu \in \set{x \sim y, \lnot (x \sim y)}$.  As an example, if 
$\chi = \lnot (x \sim y) \lor X(x) \lor X(y)$, then
the corresponding formula would be
\begin{align*}
  \bigwedge_{\mu,\chi^x, \chi^y} \big(&
(\lnot \exists x \exists y ~ x \leq y \land \mu \land  \lnot
  X(x) \land \chi^y ) ~ \land ~\\[-.9em]&
(\lnot \exists x \exists y ~ x \leq y \land \mu \land \chi^x \land \lnot X(y)) \big)
\end{align*}
for all $\mu \in \set{x{=}y, +1(x,y), \lnot (+1(x,y) \lor x {=} y)}$, $\chi^x \in \set{X(x), \lnot X(x)}$, $\chi^y \in \set{X(y), \lnot X(y)}$.
We can bring the formula into this normal form in exponential time.
We define 
$tr(\chi_{m}^x)$ 
as the conjunction of all the formulas $\avariable_0 \approx \avariable_i$ so that 
$X_i(x)$ is a conjunct of 
$\chi_{m}^x$,
and all the formulas $\lnot (\avariable_0 \approx \avariable_i)$ so that 
$\lnot X_i(x)$ is a conjunct of 
 $\chi_{m}^x$;  
we do similarly for 
$tr(\chi_{m}^y)$.
If $\mu = +1(x,y)$ and $\nu = x \sim y$ we define
$ tr\left(\exists y \big( x \leq y \land\chi_{m}\land \chi_{m}^x \land \chi_{m}^y \big)\right)$
as
 $
 \avariable \approx \X\avariable \land tr(\chi_{m}^x) \land
 \X tr(\chi_{m}^y)
 .
 $
 
If $\mu = (x = y)$ and $\nu = x \sim y$ we translate 
 \[
 tr\left(\exists y  \big( x \leq y \land\chi_{m} \land \chi_{m}^x \land
 \chi_{m}^y \big)\right) ~=~ tr(\chi_{m}^x) \land
 tr(\chi_{m}^y)
 .
 \]
 
We proceed similarly for $\mu = +1(x,y)$, $\nu = \lnot(x \sim y)$; and the translation is of course $\bot$ (false) if $\mu = (x {=} y)$, $\nu = \lnot(x {\sim} y)$.  The difficult cases are the remaining ones. Suppose 
$\mu = \lnot (+1(x,y) \lor x {=} y)$, $\nu = x {\sim} y$.
In other words, $x$ is at least two positions before $y$, and they have the same data value. 
Observe that the formula 
 $tr(\chi_{m}^x) \land \oblieq{\avariable}{\avariable}{tr(\chi_{m}^y)}$
 does not 
encode precisely this case, as it would correspond to a weaker condition $x < y  \land x \sim y$. In order to 
properly translate this case we make use of the variable $\avariable_{\textit{prev}}$, ensuring that it always has 
the data value of the variable $\avariable$ in the previous position
\[
    \textit{\textit{prev}} ~=~ \G \big( ~\X \top ~\Rightarrow~ \avariable \approx \X \avariable_{\textit{prev}} ~\big)
.
\]
 We then define $tr\left(\exists y ~ \big(~ x \leq y~\land~\chi_{m} ~\land~ \chi_{m}^x ~\land~ 
 \chi_{m}^y ~\big)\right)$ as
 \[
  tr(\chi_{m}^x) ~\land~ \oblieqz{\avariable}{\avariable_{\textit{prev}}} {
 \oblieq{\avariable_{\textit{prev}}}{\avariable}{tr(\chi_{m}^y)} }
 .
 \]
Note that by nesting twice the future obligation we ensure that the target position where $tr(\chi_{m}^y)$ must hold is at a distance of at least two positions.
For $\nu = \lnot(x {\sim}y)$ we produce a similar formula, replacing the innermost appearance of $\approx$ with $\not\approx$ in the formula above.
We then define $tr(\forall x \forall y ~ \chi)$ as 
\[
\textit{\textit{prev}} ~\land~ \bigwedge_{m \in M} ~ \lnot ~\F ~ tr\big(\exists y ~ (x \leq y 
~\land~\chi_{m} \land \chi_{m}^x \land \chi_{m}^y )\big)
.
\]

To translate $\forall x \exists y ~ (x \leq y \land \psi_k)$ we proceed in a similar way. 
One can show that $tr$  preserves satisfiability. 
By Corollary~\ref{cor:mainlogic-twoexpspace} we can decide the satisfiability of the translation in \twoexpspace, and since the translation is exponential, this gives us a $3\expspace$ 
bound for the satisfiability of $\femsot$.
\end{proof}
\fi

\ifLONG
\begin{rem}
  The proof above can be also extended to work with a similar fragment of $\emsotk$, that is, $\emsot$ extended with all binary relations of the kind $+i(x,y)$ for every $i \leq k$, with the semantics that $y$ is $i$ positions after $x$.  Hence, we also obtain the decidability of the satisfiability problem for this logic in $3\expspace$.  We do not know if the upper bounds we give for $\femsot$ and $\femsotk$ can be improved.
\end{rem}

Our result stating that $\pmainlogic_1$ is equivalent to reachability in VASS (Corollary~\ref{cor:PLRV1-equiv-VASS}), can also be seen as a hardness result for $\fo{2}{<,\sim, \set{+k}_{k \in \Nat}}$, that is, first-order logic with two variables on data words, extended with all binary relations $+k(x,y)$ 
denoting that two elements are at distance $k$. It is easy to see that this logic captures $\pmainlogic_1$ and hence that it is equivalent to reachability in VASS, even in the absence of an alphabet.
\begin{cor}
The satisfiability problem for  $\fo{2}{<,\sim, \set{+k}_{k \in \Nat}}$ is as hard as the reachability problem in VASS, even when restricted to having an alphabet $\Sigma=\emptyset$.
\end{cor}
\fi

\ifLONG
\begin{rem}
  The proof above can be also extended to work with a similar fragment of $\emsotk$, that is, $\emsot$ extended with all binary relations of the kind $+i(x,y)$ for every $i \leq k$, with the semantics that $y$ is $i$ positions after $x$.  Hence, we also obtain the decidability of the satisfiability problem for this logic in $3\expspace$.  We do not know if the upper bounds we give for $\femsot$ and $\femsotk$ can be improved.
\end{rem}

Our result stating that $\pmainlogic_1$ is equivalent to reachability in VASS (Corollary~\ref{cor:PLRV1-equiv-VASS}), can also be seen as a hardness result for $\fo{2}{<,\sim, \set{+k}_{k \in \Nat}}$, that is, first-order logic with two variables on data words, extended with all binary relations $+k(x,y)$ 
denoting that two elements are at distance $k$. It is easy to see that this logic captures $\pmainlogic_1$ and hence that it is equivalent to reachability in VASS, even in the absence of an alphabet.
\begin{cor}
The satisfiability problem for  $\fo{2}{<,\sim, \set{+k}_{k \in \Nat}}$ is as hard as the reachability problem in VASS, even when restricted to having an alphabet $\Sigma=\emptyset$.
\end{cor}
\fi

\subsection{Temporal Logics}
\makeatletter{}\ifLONG
Consider a temporal logic with (strict) future operators $\F_=$ and $\F_{\neq}$, so that ${\F_=}\, \varphi$ (resp.\ ${\F_{\neq}}\, \varphi$) holds at some position $i$ of the finite data word if there is some future position $j>i$ where $\varphi$ is true, and the data values of positions $i$ and $j$ are equal  (resp.\ distinct). We also count with ``next'' operators $\X^k_=$ and $\X^k_{\neq}$ for any $k \in \Nat$, where $\X^k_= \varphi$ (resp.\ $\X^k_{\neq} \varphi$) holds at position $i$ if $\varphi$ holds at position $i+k$, and the data values of position $i+k$ and $i$ are equal (resp.\ distinct). Finally, the logic also features a standard until operator $\U$, tests for the labels of positions, and it is closed under Boolean operators.  We call this logic $\ltlrepeat$. There is an efficient satisfiability-preserving translation from $\ltlrepeat$ into $\mainlogic$ and back and hence we have the following.

\begin{prop}\label{prop:sat-LTLrepeat:2expspace}
  The satisfiability problem for $\ltlrepeat$ is $\twoexpspace$-complete.
\end{prop}
\begin{proof}[Proof sketch]
For the $\twoexpspace$-membership, there is a straightforward polynomial-time translation from  $\ltlrepeat$ into $\mainlogic$ that preserves satisfiability, where 
$\F_= \varphi$ is translated as $\oblieq{\avariable}{\avariable}{\varphi'}$
and $\F_{\neq} \varphi$ is translated as $\oblineq{\avariable}{\avariable}{\varphi'}$; $\X_=^k \varphi$ as $\X^k \varphi' ~\land~ \avariable \approx \X^k \avariable$; and any test for label $a_i$ by $\avariable_0 \approx \avariable_i$ ($\varphi'$ is the translation of $\varphi$).

On the other hand, given a formula $\phi$ of $\mainlogic^{\top}$ using $k$ variables, consider the alphabet $\set{a_1, \dotsc, a_k}$ and the formula $\psi_k$ that forces the model to be a succession of `blocks' of data words of length $k$ with labels $a_1 \dotsb a_k$ (hence, forcing its length to be a multiple of $k$). We show how to define a formula $tr(\phi)$ of $\ltlrepeat$ so that $tr(\phi) \land \psi_k$ is satisfiable if and only if $\phi$ is satisfiable. We define $tr(\X \phi')$ as $\X^k tr(\phi')$; and $tr(\F \phi')$ as $\F (a_1 \land tr(\phi'))$. For future obligation formulas, we define $tr(\oblieq{\avariable_i}{\avariable_j}{\top})$  as $\X^{i-1}\F_= a_j$ if $j\leq i$ or 
$\X^{i-1}(\X_=^{j-i}\F_= a_j \lor (\F_= a_j \land \X_{\neq}^{j-i} \top))$ otherwise ---note that in the second case special care must be taken to ensure that the data value is repeated at a strictly future block. Finally, for the local obligation formulas, we define $tr(\avariable_i \oblieqlocal \mynext^{t} \avariable_j)$ as $\X^{i-1}\X_=^{t \cdot k - (i-1) + (j-1)}\top$. Note that this can be easily extended to treat inequalities. Thus, there is a polynomial-time reduction from the satisfiability problem for $\mainlogic^\top$ into that of $\ltlrepeat$.
\end{proof}

In fact, $\ltlrepeat$ corresponds to a fragment of the linear-time 
temporal logic $\textsf{LTL}$ extended with one register for storing and comparing data values. We  denote it by $\ensuremath{{\sf LTL}^\downarrow_1}$, and it was studied in \cite{DL-tocl08}.
This logic contains one operator to store the current datum, one operator to test whether the current datum is equal to 
the one stored. The \emph{freeze} operator $\downarrow \varphi$ permits to \emph{store} the current datum in the 
register and continue the evaluation of the formula $\varphi$. The operator $\uparrow$  \emph{tests} whether the current 
data value is equal to the one stored in the register. When the temporal operators are limited to $\F, \U$ and $\X$, this 
logic is decidable with non-primitive-recursive complexity \cite{DL-tocl08}. 

Indeed $\ltlrepeat$ is the fragment where we only allow $\downarrow$ and $\uparrow$ to appear in the form of $\downarrow \F ({\uparrow} \land \varphi)$ and $\downarrow \X^k ({\uparrow} \land \varphi)$ ---or with $\lnot\,{\uparrow}$ instead of $\uparrow$. Markedly, this restriction allows us to jump from a non-primitive-recursive complexity of the satisfiability problem, to an elementary $\twoexpspace$ complexity.
\else Our results have also implications for a fragment of the 
logic LTL extended with one register for storing and comparing data values (called the \emph{freeze} operator), noted $\textup{LTL}_1^{\downarrow}$.
Its satisfiability problem was shown to be decidable, but with non-primitive-recursive complexity \cite{DL-tocl08}.

 Our results yield a fragment of $\textup{LTL}_1^{\downarrow}$ with 
elementary $2\expspace$ upper bound.

\fi

\makeatletter{}\section{Conclusion}
\label{section-conclusion}
\ifLONG
We introduced  the logic $\mainlogic$ that significantly extends the languages 
in~\cite{Demri&DSouza&Gascon12}, for instance by allowing future obligations of the general form 
$\oblieq{\avariable}{\avariablebis}{\aformula}$. 
We have shown that \SAT{ $\mainlogic$} can be reduced to the
control state reachability problem in VASS, obtaining
a \twoexpspace{} upper bound as a consequence. Since $\mainlogic$ can be also viewed as a fragment of a logic introduced
in~\cite{Kara&Schwentick&Zeume10} whose satisfiability is equivalent
to Reach(VASS), we provide
also an interesting fragment with elementary complexity. 
The reduction into the control state reachability problem involves an
exponential blow-up, which is unavoidable as demonstrated by our
\twoexpspace{} lower bound. To prove this lower bound,
we  introduced  the class of chain systems of level $k$ and we 
proved the
$(k+1)\expspace$-completeness of the control state reachability problem by 
extending the proof from~\cite{Lipton76,Esparza98}. 
This class of systems is interesting for its own sake and could be used to establish other hardness
results thanks to our results. 
We have also shown that the proof technique we used to reduce
$\mainlogic^{\top}$ to the control state reachability problem does not
work in the presence of past obligations. Indeed, the satisfiability
problem for $\pmainlogic^{\top}$ ($\mainlogic^{\top}$ with past
obligations) is as hard as Reach(VASS)
which witnesses that  past obligations  have a computational cost. 
Furthermore, note that none of our lower bound proofs involve the past-time
operators $\previous$ and $\since$. 
Finally, a new correspondence between data logics and decision problems
for VASS is provided, apart from the fact that the complexity of several data 
\ifLONG
logics has been characterised, some of the logics
being defined with first-order features (see Section~\ref{section-implications}). 
\else
logics has been characterized.
\fi 
\begin{figure*}
\center
\cut{
\begin{minipage}{17cm}
$\underbrace{\mainlogic^{\top}_k}_{\text{\normalsize\pspace{\rm -complete}}}$ \ \ \ 
$\underbrace{\mainlogic \equiv \mainlogic^{\top} \equiv \mainlogic_1 \equiv
\mainlogic  +  \set{\oplus_1, \ldots, \oplus_k}}_{\text{\normalsize\twoexpspace{\rm -complete}}}$ \ \ \ 
$\underbrace{\pmainlogic \equiv \pmainlogic^{\top} \equiv \pmainlogic_1}_{\text{\normalsize $\equiv$ Reach(VASS)}}$ 
\ \  \ 
$\underbrace{\mainlogic^{\top}_{vec}}_{\text{\normalsize undecidable}}$ 
\end{minipage}}
\begin{minipage}{4.8cm}
\begin{align*}
  \mainlogic^{\top}_k &:\text{\normalsize\pspace{\rm -complete}}\\
\mainlogic \equiv \mainlogic^{\top} \equiv \mainlogic_1 \equiv
\mainlogic  +  \set{\oplus_1, \ldots, \oplus_k}
&:\text{\normalsize\twoexpspace{\rm -complete}}\\
\pmainlogic \equiv \pmainlogic^{\top} \equiv \pmainlogic_1&\text{\normalsize $\equiv$  Reach(VASS)}\\
\mainlogic^{\top}_{vec}&: \text{\normalsize undecidable}
\end{align*}
\end{minipage}
\caption{Summary of results.}
\label{figure-summary}
\end{figure*}
A summary of the results can be found in Figure~\ref{figure-summary}.

\else
\enlargethispage{-0.4in}
We introduced  the logic $\mainlogic$ and variants by allowing data value repetitions
thanks to formulas of the form $\oblieq{\avariable}{\avariablebis}{\aformula}$. 
$\mainlogic$ extends the main logic from~\cite{Demri&DSouza&Gascon12}
but it is also  a fragment of BD-LTL from~\cite{Kara&Schwentick&Zeume10}
whose satisfiability is equivalent
to Reach(VASS). We showed that \SAT{$\mainlogic$} is \twoexpspace-complete
by reduction into the control-state reachability problem for VASS
(via a detour to  gainy VASS) and
 by introducing a new class of counter machines (the chain systems) in order to
get the complexity lower bound. 
This new class of counter machines is interesting for its own sake and could be used to establish other hardness
results thanks to our results that extend non-trivially the proof from~\cite{Lipton76,Esparza98}. 
Correspondences between extensions of $\mainlogic$ (such as $\pmainlogic$, $\pmainlogic^{\top}$ and 
$\pmainlogic_1$) and reachability problem for VASS are also established, 
strengthening further the relationships between data logics and reachability problems for counter machines.
Other results for  variants are presented in
the paper and a summary  can be found  below. 

\scalebox{.8}{
\begin{minipage}{4.8cm}
\begin{align*}
  \mainlogic^{\top}_k &:\text{\normalsize\pspace{\rm -complete}}\\
\mainlogic \equiv \mainlogic^{\top} \equiv \mainlogic_1 \equiv
\mainlogic  +  \set{\oplus_1, \ldots, \oplus_k}
&:\text{\normalsize\twoexpspace{\rm -complete}}\\
\pmainlogic \equiv \pmainlogic^{\top} \equiv \pmainlogic_1&\text{\normalsize $\equiv$  Reach(VASS)}\\
\mainlogic^{\top}_{vec}&: \text{\normalsize undecidable}
\end{align*}
\end{minipage}
}
\fi 

\newcommand{\etalchar}[1]{$^{#1}$}

\clearpage
\onecolumn
\appendix
\makeatletter{}\section{Elements of the proof of Lemma~\ref{lem:APhiFromEarlierPaper} from~\cite{Demri&DSouza&Gascon12}}
\label{section-appendix-aphifrompreviouspaper}

Below, we recall developments from~\cite{Demri&DSouza&Gascon12} that lead to the proof of 
Lemma~\ref{lem:APhiFromEarlierPaper}. We provide the main definitions but leave the details 
to~\cite{Demri&DSouza&Gascon12}.

Let $\aformula$ be an $\mainlogic^\top$ formula built over variables in
 $\{\avariable_1,\dots, \avariable_k\}$. Let $l$ be the maximal $i$ such
 that a term of the 
form $\mynext^i x$ occurs in $\aformula$ (without any loss of generality, we can assume that $l \geq 1$). 
The value $l$ is called the
$\mynext$-length of $\aformula$. In order to define the set of
atomic formulae used in the symbolic models
we introduce the set of constraints
$\Omega^l_k$ that contains constraints of the form either $\mynext^i \avariable =
\mynext^j \avariablebis$ or $\mynext^i (\oblieq{\avariable}{\avariablebis}{\top})$ with $\avariable, \avariablebis \in \{\avariable_1,\dots, \avariable_k\}$
and $i,j \in \interval{0}{l}$. 

The set of $(l,k)$-frames is denoted $\FFrame^l_k$, and is made up of pairs
comprising a set of atomic constraints, along
with some information 
about the last position of the model. This latter component is an element of
$\interval{0}{l}
\uplus \set{\rm nd}$ where `${\rm nd}$' means that the model does not
end before the end of the frame. We have $\alpha={\rm nd}$  and $\pair{\symbval}{\alpha} \in
\FFrame^l_k$ iff $\symbval$ satisfies the conditions:

\begin{description}

\item[(F1)] For all $i\in \interval{0}{l}$ and $\avariable \in \{\avariable_1,\ldots,
  \avariable_k\}$, $\mynext^i \avariable = \mynext^i \avariable \in \symbval$.
\item[(F2)] For all $i,j \in \set{0,\ldots,l}$ and $\avariable, \avariablebis \in
  \{\avariable_1,\ldots, \avariable_k\}$, $\mynext^i \avariable = \mynext^j \avariablebis \in \symbval$ iff
  $\mynext^j \avariablebis = \mynext^i \avariable \in \symbval$.
\item[(F3)] For all $i,j,j' \in \interval{0}{l}$ and $\avariable, \avariablebis, \avariableter \in \{
  \avariable_1, \ldots, \avariable_k \}$, if $\set{\mynext^i \avariable = \mynext^j \avariablebis, \mynext^j \avariablebis =
    \mynext^{j'} \avariableter} \subseteq \symbval$ then $\mynext^i \avariable = \mynext^{j'} \avariableter
  \in \symbval$.
\item[(F4)] For all $i,j \in \interval{0}{l}$ and $\avariable, \avariablebis \in
  \{\avariable_1,\ldots, \avariable_k\}$ such that $\mynext^i \avariable = \mynext^j \avariablebis \in \symbval$:
  \begin{itemize}
  \itemsep 0 cm
  \item if $i=j$, then for every $\avariableter \in \{ \avariable_1, \dots, \avariable_k \}$ we have
    $\mynext^i (\oblieq{\avariable}{\avariableter}{\top}) \in \symbval$ iff $\mynext^j (\oblieq{\avariablebis}{\avariableter}{\top}) \in
    \symbval$;
  \item if $i<j$ then $\mynext^i (\oblieq{\avariable}{\avariablebis}{\top}) \in \symbval$, and for
    $\avariableter \in \set{\avariable_1,\ldots,\avariable_k}$, $\mynext^i (\oblieq{\avariable}{\avariableter}{\top}) \in \symbval$ iff either
    $\mynext^j (\oblieq{\avariablebis}{\avariableter}{\top}) \in \symbval$ or there exists $i < j' \leq
    j$ such that $\mynext^i \avariable = \mynext^{j'} \avariableter \in \symbval$.
  \end{itemize}
\end{description}

Additionally, we have  $\alpha \in \interval{0}{l}$ and $\pair{\symbval}{\alpha} \in
 \FFrame^l_k$ iff $\symbval \subseteq \Omega_k^{\alpha}$ and 
\begin{description}
\item[(F1$'$)] For all $j \in \interval{0}{\alpha}$ and $\avariable \in \{\avariable_1,\ldots,
  \avariable_k\}$, $\mynext^j \avariable = \mynext^j \avariable \in \symbval$.
\item[(F2$'$)] For all $j,j' \in  \interval{0}{\alpha}$ and $\avariable, \avariablebis \in
  \{\avariable_1,\ldots, \avariable_k\}$, $\mynext^j \avariable = \mynext^{j'} \avariablebis \in \symbval$ iff
  $\mynext^{j'} \avariablebis = \mynext^j \avariable \in \symbval$.
\item[(F3$'$)] For all $j,j',j'' \in \interval{0}{\alpha}$ and $\avariable, \avariablebis, \avariableter \in \{
  \avariable_1, \ldots, \avariable_k \}$, if $\set{\mynext^j \avariable = \mynext^{j'} \avariablebis,\mynext^{j'} \avariablebis
    = \mynext^{j''} \avariableter} \subseteq \symbval$ then $\mynext^j \avariable = \mynext^{j''}
  \avariableter \in \symbval$.
\item[(F4$'$)] For all $j,j' \in \interval{0}{\alpha}$ and $\avariable, \avariablebis \in
  \{\avariable_1,\dots \avariable_k\}$ such that $\mynext^{j} \avariable = \mynext^{j'} \avariablebis \in \symbval$:
  \begin{itemize}
  
  \item if $j=j'$, then for every $\avariableter \in \{ \avariable_1, \dots, \avariable_k \}$ we
    have $\mynext^j (\oblieq{\avariable}{\avariableter}{\top}) \in \symbval$ iff $\mynext^{j'} 
   (\oblieq{\avariablebis}{\avariableter}{\top}) \in \symbval$;
  \item if $j<j'$ then $\mynext^j (\oblieq{\avariable}{\avariableter}{\top}) \in \symbval$, and for
    $\avariableter \in \set{\avariable_1,\ldots,\avariable_k}$, $\mynext^j (\oblieq{\avariable}{\avariableter}{\top}) \in \symbval$
    iff either $\mynext^{j'} (\oblieq{\avariablebis}{\avariableter}{\top}) \in \symbval$ or there exists $j
    < j'' \leq j'$ such that $\mynext^j \avariable = \mynext^{j''} \avariableter \in \symbval$.
  \end{itemize}
\item[(F5$'$)] For every constraint $\mynext^{j} (\oblieq{\avariable}{\avariableter}{\top}) \in
  \symbval$ such that $j \in \interval{0}{\alpha}$ there is $j < j' \leq \alpha$
  such that $\mynext^j \avariable = \mynext^{j'} \avariablebis \in \symbval$.
\end{description}
The last condition (F5$'$) imposes that every future obligation is
satisfied before the end of the model.  
The conditions (F1$'$)--(F4$'$) are variants of (F1)--(F4) in which
the value $l$ is replaced by $\alpha$ (index of the last position in the model).
A finite model $\sigma$
satisfies a frame $\pair{\symbval}{\alpha}
\in \FFrame^l_k$ at position $j$ iff
\begin{itemize}

\item either $\alpha = {\rm nd}$ and $\length{\sigma} - (j+1) \geq l+1$,
or $\alpha = \length{\sigma} - (j+1)$,
\item for every constraint $\aconstraint$ in $\symbval$ we have
  $\sigma,j \models \aconstraint$.
\end{itemize}
In this case, we write $\sigma, j \models \pair{\symbval}{\alpha}$.

 We set
$\tup{\top?}_{\pair{\symbval}{\alpha}} (\avariable,i) \egdef [(\avariable,i)]_{\pair{\symbval}{\alpha}} \egdef
\emptyset$ when $\alpha \neq {\rm nd}$ and $i>\alpha$. If $\alpha = {\rm nd}$ or $i
\leq \alpha$ then the definitions are the following:
\[ 
\tup{\top?}_{\pair{\symbval}{\alpha}}(\avariable,i) \egdef \set{ \avariablebis \mid \mynext^i ( \oblieq{\avariable}{\avariablebis}{\top} ) \in
  \symbval}.
\]
\[
[(\avariable,i)]_{\pair{\symbval}{\alpha}} \egdef \set{ \avariablebis \mid \mynext^i \avariable = \mynext^i \avariablebis \in
  \symbval}.
\]
A pair of $(l,k)$-frames
$\pair{\pair{\symbval}{\alpha}}{\pair{\symbval'}{\alpha'}}$ is one-step
consistent iff
\begin{itemize}
\item $\pair{\alpha}{\alpha'}$ belongs to 
      $\set{\pair{j}{j'} \in \interval{1}{l}  \times \interval{0}{l-1}: j' = j - 1}
       \cup \set{\pair{{\rm nd}}{{\rm nd}},\pair{{\rm nd}}{l}}
     $.
\item $\pair{\symbval}{\symbval'}$ satisfies  the  following conditions 
      when $ \alpha = \alpha' = {\rm nd}$:
\begin{description}
\item[(OSC1)] for all $\mynext^i \avariable = \mynext^j \avariablebis \in \Omega^l_k$ with 
       $0 < i,j$, we have $\mynext^i \avariable = \mynext^j \avariablebis \in \symbval$ iff
       $\mynext^{i-1} \avariable = \mynext^{j-1} \avariablebis \in \symbval'$,
\item[(OSC2)] for all  $\mynext^i (\oblieq{\avariable}{\avariablebis}{\top} ) \in \Omega^l_k$ with $i > 0$, we have
 $\mynext^i (\oblieq{\avariable}{\avariablebis}{\top} ) \in \symbval$ iff $\mynext^{i-1} (\oblieq{\avariable}{\avariablebis}{\top})
 \in \symbval'$.
\end{description}
\item  $\pair{\symbval}{\symbval'}$ satisfies the conditions below when $\alpha \in \interval{1}{l}$:
      \begin{itemize}
       \item for every $\mynext^j \avariable = \mynext^{j'} \avariablebis \in \Omega^{\alpha}_k$ with 
       $0 < j,j'$, we have 
       $\mynext^j \avariable = \mynext^{j'} \avariablebis \in \symbval$
       iff $\mynext^{j-1} \avariable  = \mynext^{j'-1} \avariablebis \in \symbval'$,
       \item for every  $\mynext^j (\oblieq{\avariable}{\avariablebis}{\top}) \in \Omega^{\alpha}_k$ with $j > 0$, we have
        $\mynext^j (\oblieq{\avariable}{\avariablebis}{\top}) \in \symbval$ iff
        $\mynext^{j-1} ( \oblieq{\avariable}{\avariablebis}{\top} ) \in \symbval'$.
        \end{itemize}
\end{itemize}
A symbolic model is a finite one-step consistent sequence of frames
ending with a symbolic valuation of the form
$\pair{\symbval}{0}$. Hence, for every position $i$ in a finite
symbolic model $\rho$, if $\length{\rho} -(i+1) \geq l+1$ then $\rho(i)$ is
of the form $\pair{\symbval}{\rm nd}$, otherwise $\rho(i)$ is of the
form $\pair{\symbval}{\length{\rho} -(i+1)}$.

Given a symbolic model $\rho$, we adopt the following notations:
\[ 
\tup{\top?}_{\rho}(\avariable,i) \egdef 
\tup{\top?}_{\rho(i)}(\avariable,0)
\]
\[
[(\avariable,i)]_{\rho} \egdef [(\avariable,0)]_{\rho(i)}.
\]
The symbolic satisfaction relation $\rho,i \symbmodels
\aformula$ is defined in the obvious way (for instance, $\rho, i \models \mynext^{j} \avariable  = \mynext^{j'} \avariablebis$ iff
$\mynext^{j} \avariable  = \mynext^{j'} \avariablebis$ belongs to the first component of $\rho(i)$). 
The following correspondence between symbolic and concrete models
can be shown similarly.\medskip

\begin{lem}~\cite{Demri&DSouza&Gascon12}
  An $\mainlogic^\top$ formula $\aformula$ of $\mynext$-length $l$ over the
  variables $\{\avariable_1,\ldots, \avariable_k\}$ is satisfiable over finite
  models iff there exists a 
  finite $(l,k)$-symbolic model $\rho$ such that $\rho \symbmodels \aformula$
  and $\rho$ is realizable (i.e., \ there is a finite model whose symbolic model is $\rho$).
\end{lem}

We define $\Aphi$ as an automaton over the alphabet  $\FFrame^l_k$ (only used for synchronisation purposes),
that accepts the intersection of the languages accepted by the
three automata $\Aosc$, 
$\Asymb$ and $\Asat$ described below:
\begin{itemize}

\item $\Aosc$ recognizes the sequences of frames such that every pair of consecutive
  frames is one-step consistent (easy),
\item $\Asymb$ recognizes the set of symbolic
  models satisfying $\aformula$ (as for LTL formulae),
\item $\Asat$ recognizes the set of symbolic
  models that are realizable. 
\end{itemize}

The automaton $\Aosc$ is a finite-state automaton that checks that the
sequence is one-step 
consistent. Formally, we build $\Aosc = \triple{Q}{Q_0,F}{\step{}}$
such that:
\begin{itemize}

\item $Q$ is the set of $(l,k)$-frames and $Q_0=Q$.

\item the transition relation is defined by $\pair{\symbval_1}{\alpha_1}
  \step{\pair{\symbval}{\alpha}} \pair{\symbval_2}{\alpha_2}$ iff $\pair{\symbval}{\alpha} = \pair{\symbval_1}{\alpha_1}$ and the pair
  $\pair{\pair{\symbval_1}{\alpha_1}}{ \pair{\symbval_2}{\alpha_2}}$ is one-step consistent.

\item  The set
of final states $F$ is equal to $Q$.

\end{itemize}

We define the finite-state automaton $\Asymb$ by adapting the
construction  for LTL. We define
$\mathit{cl}(\aformula)$ to be the standard \emph{closure} of
$\aformula$, namely the smallest set of formulas $\aset$ that contains
$\aformula$, is closed under subformulas, and
satisfies the following conditions:
\begin{itemize}
\item If $\aformulabis \in \aset$ and $\aformulabis$
  is not of the form $\neg \aformulabis_1$ for some $\aformulabis_1$,
  then $\neg \aformulabis \in \aset$.
\item If $\aformulabis_1 \until \aformulabis_2 \in
  \aset$ then
  $\mynext(\aformulabis_1 \until \aformulabis_2) \in \aset$.
\item If $\aformulabis_1 \since \aformulabis_2 
  \in \aset$ then 
  $\mynext^{-1}(\aformulabis_1 \since \aformulabis_2) \in \aset$.
\end{itemize} 
An atom of $\aformula$ is a subset
$\mathit{At}$ of $\mathit{cl}(\aformula)$ which is maximally
consistent: it satisfies the following conditions.
\begin{itemize}
\item For every $\neg \aformulabis \in
  \mathit{cl}(\aformula)$, we have $\neg \aformulabis \in \mathit{At}$ iff
  $\aformulabis \not\in \mathit{At}$.

\item For every $\aformulabis_1 \wedge \aformulabis_2 \in 
\mathit{cl}(\aformula)$, we have $\aformulabis_1 \wedge \aformulabis_2 
\in \mathit{At}$ iff  $\aformulabis_1$ and $\aformulabis_2$ are 
in $\mathit{At}$.

\item For every $\aformulabis_1 \vee \aformulabis_2 \in 
\mathit{cl}(\aformula)$, we have
$\aformulabis_1 \vee \aformulabis_2 \in \mathit{At}$ iff 
$\aformulabis_1$ or $\aformulabis_2$ is in $\mathit{At}$.

\item For every $\aformulabis_1 \until \aformulabis_2 \in 
\mathit{cl}(\aformula)$, we have
$\aformulabis_1 \until \aformulabis_2 \in \mathit{At}$ iff 
either $\aformulabis_2 \in \mathit{At}$ or both
$\aformulabis_1$ and $\mynext(\aformulabis_1 \until
\aformulabis_2)$ are in $\mathit{At}$.

\item For every $\aformulabis_1 \since \aformulabis_2\in 
\mathit{cl}(\aformula)$, we have 
$\aformulabis_1 \since \aformulabis_2 \in \mathit{At}$
iff either $\aformulabis_2 \in \mathit{At}$ or both
$\aformulabis_1$ and $\mynext^{-1}(\aformulabis_1 \since
\aformulabis_2)$ are in $\mathit{At}$.
\end{itemize} 
We denote by ${\rm Atom}(\aformula)$ the set of atoms of $\aformula$.
We now define $\Asymb = (Q,Q_0,\rightarrow,F)$ over the alphabet
$\FFrame^l_k$, where:
\begin{itemize}
\item $Q = {\rm Atom}(\aformula) \times (\FFrame^l_k
\cup \{\sharp\})$ and 
$Q_0= \set{\pair{\mathit{At}}{\sharp} \mid \aformula \in
\mathit{At}}$,
\item $\pair{\mathit{At}}{\aset} \step{\pair{\symbval}{\alpha}} \pair{\mathit{At'}}{\pair{\symbval}{\alpha}}$
iff
\begin{description}
\itemsep 0 cm
\item[(atomic)] if $\alpha = {\rm nd}$ then $\mathit{At}' \intersection
  \Omega^l_k = \symbval$,
      otherwise $\mathit{At}' \intersection
  \Omega^{\alpha}_k = \symbval$, 
      \item[(one-step)] \ 
\begin{enumerate}
\item if $\alpha = 0$ then there is no formula of the form $\mynext
  \aformulabis$ in $\mathit{At}$, otherwise for every $\mynext \aformulabis
  \in \mathit{cl}(\aformula)$, we have $\mynext \aformulabis \in
  \mathit{At}$ iff $\aformulabis \in \mathit{At}',$\\[0.2em]
\item for every $\previous \aformulabis \in \mathit{cl}(\aformula)$,
  we have $\aformulabis \in \mathit{At}$ iff $\previous \aformulabis
  \in \mathit{At}'$.
\end{enumerate}
\end{description}
\item The set
of final states is therefore equal to $\set{\pair{\symbval}{0} \mid
\pair{\symbval}{0} \in \FFrame^l_k}$.
\end{itemize}

For each $\aset \in \nepowerset{\set{\avariable_1, \ldots,\avariable_k}}$, we introduce a counter
that keeps track of the number of obligations that need to be
satisfied by $\aset$. We identify the counters with finite subsets of
$\set{\avariable_1, \ldots,\avariable_k}$.  

A counter valuation 
is a map from $\nepowerset{\set{\avariable_1,\ldots,\avariable_k}}$ to $\Nat$.
We define below a canonical sequence of counter valuations along a
symbolic model.  We will need to introduce some additional
definitions first.  For an $(l,k)$-frame $\pair{\symbval}{\alpha}$ and a counter $\aset \in
\nepowerset{\set{\avariable_1,\ldots, \avariable_k}}$, we define a point of
  increment for $\aset$ in  $\pair{\symbval}{\alpha}$  to be an equivalence class of the
form $[(\avariable,0)]_{\pair{\symbval}{\alpha}}$ such that $\tup{\top?}_{\pair{\symbval}{\alpha}} (\avariable,0) = \aset$ and
 there is no edge between $(\avariable,0)$ and some 
$(\avariablebis,j)$ with $j \in \interval{1}{l}$ assuming that $\symbval$ is represented
graphically. 
In a similar way, a
point of decrement for $\aset$ in $\pair{\symbval}{\alpha}$ is defined to be an
equivalence class of the form $[(\avariable,l)]_{\pair{\symbval}{\alpha}}$ such that
$\tup{\top?}_{\pair{\symbval}{\alpha}} (\avariable,l) \union [(\avariable,l)]_{\pair{\symbval}{\alpha}} = \aset$, and 
there is no edge between $(\avariable,l)$ and some $(\avariablebis,j)$ with $j \in \interval{0}{l-1}$. 
So, when $\alpha \neq nd$, there is point of decrement in  $\pair{\symbval}{\alpha}$.

We denote by $u^+_{\pair{\symbval}{\alpha}}$ the counter valuation which records the
number of points of increment for each counter $\aset$ in $\pair{\symbval}{\alpha}$.
Similarly $u^-_{\pair{\symbval}{\alpha}}$ is the counter valuation which records the
number of points of decrement for each counter $\aset$ in $\pair{\symbval}{\alpha}$.
The codomain of both $u^+_{\pair{\symbval}{\alpha}}$ and
$u^-_{\pair{\symbval}{\alpha}}$ is $\interval{0}{k}$.

Let $\rho$ be an $(l,k)$-symbolic model.
For every $\aset \in \nepowerset{\set{\avariable_1,\ldots,\avariable_k}}$, a point of
  increment for $\aset$ in $\rho$ is an equivalence class of the form
$[(\avariable,i)]_{\rho}$ such that $[(x,0)]_{\rho(i)}$ is a point of increment
for $\aset$ in the frame $\rho(i)$.  Similarly, a point of
  decrement for $\aset$ in $\rho$ is an equivalence class of the form
$[(\avariable,i)]_{\rho}$ such that $i \geq l+1$ and $[(\avariable,l)]_{\rho(i-l)}$ is a
point of decrement for $\aset$ in $\rho(i-l)$.

We can now define a canonical counter valuation sequence $\acountseq$
along $\rho$, called the counting sequence along $\rho$, which
counts the number of obligations corresponding to each subset of
variables $\aset$
that remain unsatisfied at each position.  
We define $\acountseq$ inductively: for each $\aset \in
\nepowerset{\set{\avariable_1,\ldots,\avariable_k}}$ we have $\acountseq(0)(\aset) \egdef 0$ and
$
\acountseq(i+1)(\aset) \egdef \max(0,\acountseq(i)(\aset) + (u^+_{\rho(i)}(\aset) -
u^-_{\rho(i+1)}(\aset))),
$
for every $0 \leq i < \length{\rho}$.

The following result shows that the set of realizable symbolic models
can be characterized using their counting sequences.\medskip

\begin{lem}~\cite{Demri&DSouza&Gascon12}
 A symbolic model $\rho$ is realizable iff for every
counter $\aset$ the final value of the counting sequence $\acountseq$ along
$\rho$ is equal to $0$.
\end{lem}

It remains to define the automaton $\Asat$  that recognizes the set of
realizable symbolic modelst.
The automaton $\Asat$ is equal to 
$\triple{\locations,\locations_0,\locations_f}{\aalphabet,C}{\transitions}$ such that
\begin{itemize}

\item $\locations = \FFrame^l_k$, $\locations_0 = \locations$, $\locations_f = \set{\pair{\symbval}{0} \ \mid  \ \pair{\symbval}{0} \in \locations}$.
  
\item $\aalphabet = \FFrame^l_k$.
  
\item $C = \nepowerset{\{\avariable_1, \ldots, \avariable_k\}}$.
  
\item The transition relation $\transitions$ is defined by
\[
\pair{\symbval}{\alpha} \step{\mathit{up},\pair{\symbval}{\alpha}}
\pair{\symbval'}{\alpha'} \in \transitions
\]
for every $\pair{\symbval}{\alpha}, \pair{\symbval'}{\alpha'} \in
\FFrame^l_k$ and $\mathit{up}: C \rightarrow \interval{-k}{+k}$ verifying 
$$u^+_{\pair{\symbval}{\alpha}}(\aset) -
u^-_{\pair{\symbval'}{\alpha'}}(\aset) \leq \mathit{up}(\aset) \leq
u^+_{\pair{\symbval}{\alpha}}(\aset)$$ 
for every $\aset \in C$ (the letter $\pair{\symbval}{\alpha}$ shall be removed after the product construction below
as well as the notions of initial or final states).
\end{itemize}\medskip

\noindent A symbolic model $\rho$ is accepted by $\Asat$ iff $\rho$ is
realizable~\cite{Demri&DSouza&Gascon12}.

The VASS defined as the intersection of the two automata $\Aosc$, 
$\Asymb$ and the VASS $\Asat$, is indeed a VASS satisfying the assumptions in 
Lemma~\ref{lem:APhiFromEarlierPaper} (once the alphabet is removed but used to build
the product). It is worth noting that 
$\Aosc$ and $\Asymb$  are finite-state automata and therefore the counter updates
are those from the transitions in $\Asat$.


\begin{thebibliography}{BDM{\etalchar{+}}11}

\bibitem[AJ96]{Abdulla&Jonsson96}
P.~Abdulla and B.~Jonsson.
\newblock Verifying programs with unreliable channels.
\newblock {\em Information and Computation}, 127(2):91--101, 1996.

\bibitem[AvW12]{Alur&Cerny&Weinstein12}
R.~Alur, P.~\v{C}ern{\'y}, and S.~Weinstein.
\newblock Algorithmic analysis of array-accessing programs.
\newblock {\em ACM Transactions on Computational Logic}, 13(3), 2012.

\bibitem[BCGK12]{Bolligetal12}
B.~Bollig, A.~Cyriac, P.~Gastin, and K.~Narayan Kumar.
\newblock Model checking languages of data words.
\newblock In {\em {FoSSaCS}'12}, volume 7213 of {\em Lecture Notes in Computer
  Science}, pages 391--405. Springer, 2012.

\bibitem[BDM{\etalchar{+}}11]{BDMSS06:journal}
M.~Boja{\'n}czyk, C.~David, A.~Muscholl, Th. Schwentick, and L.~Segoufin.
\newblock Two-variable logic on data words.
\newblock {\em ACM Transactions on Computational Logic}, 12(4):27, 2011.

\bibitem[BL10]{Bojanczyk&Lasota10}
M.~Boja{\'n}czyk and S.~Lasota.
\newblock An extension of data automata that captures {XP}ath.
\newblock In {\em LICS'10}, pages 243--252. IEEE, 2010.

\bibitem[BMS{\etalchar{+}}06]{BDMSS06}
M.~Boja{\'n}czyk, A.~Muscholl, Th. Schwentick, L.~Segoufin, and C.~David.
\newblock Two-variable logic on words with data.
\newblock In {\em LICS'06}, pages 7--16. {IEEE}, 2006.

\bibitem[Bol11]{Bollig11}
B.~Bollig.
\newblock An automaton over data words that captures {EMSO} logic.
\newblock In {\em CONCUR'11}, volume 6901 of {\em Lecture Notes in Computer
  Science}, pages 171--186. Springer, 2011.

\bibitem[Bou02]{Bouyer02}
P.~Bouyer.
\newblock A logical characterization of data languages.
\newblock {\em Information Processing Letters}, 84(2):75--85, 2002.

\bibitem[Dav09]{David09}
C.~David.
\newblock {\em Analyse de {XML} avec donn\'ees non-born\'ees}.
\newblock PhD thesis, LIAFA, 2009.

\bibitem[DDG12]{Demri&DSouza&Gascon12}
S.~Demri, D.~D'Souza, and R.~Gascon.
\newblock Temporal logics of repeating values.
\newblock {\em Journal of Logic and Computation}, 22(5):1059--1096, 2012.

\bibitem[DFP13]{Demri&Figueira&Praveen13}
S.~Demri, D.~Figueira, and M.~Praveen.
\newblock Reasoning about data repetitions with counter systems.
\newblock In {\em LICS'13}, pages 33--42. IEEE, 2013.

\bibitem[DHLT14]{Deckeretal14}
N.~Decker, P.~Habermehl, M.~Leucker, and D.~Thoma.
\newblock Ordered navigation on multi-attributed data words.
\newblock In {\em CONCUR'14}, volume 8704 of {\em Lecture Notes in Computer
  Science}, pages 497--511, 2014.

\bibitem[DJLL09]{Demrietal09}
S.~Demri, M.~Jurdzi{\'n}ski, O.~Lachish, and R.~Lazi{\'c}.
\newblock The covering and boundedness problems for branching vector addition
  systems.
\newblock In {\em FST\&TCS'09}, pages 181--192. LZI, 2009.

\bibitem[DL09]{DL-tocl08}
S.~Demri and R.~Lazi{\'c}.
\newblock {LTL} with the freeze quantifier and register automata.
\newblock {\em ACM Transactions on Computational Logic}, 10(3), 2009.

\bibitem[DS02]{Demri&Schnoebelen02}
S.~Demri and Ph. Schnoebelen.
\newblock The complexity of propositional linear temporal logics in simple
  cases.
\newblock {\em Information and Computation}, 174(1):84--103, 2002.

\bibitem[Esp98]{Esparza98}
J.~Esparza.
\newblock Decidability and complexity of {Petri} net problems --- an
  introduction.
\newblock In {\em Advances in Petri Nets 1998}, volume 1491 of {\em Lecture
  Notes in Computer Science}, pages 374--428. Springer, Berlin, 1998.

\bibitem[Fig10]{Figueira10}
D.~Figueira.
\newblock {\em Reasoning on words and trees with data}.
\newblock PhD thesis, ENS Cachan, 2010.

\bibitem[Fig11]{Figueira11}
D.~Figueira.
\newblock A decidable two-way logic on data words.
\newblock In {\em {LICS}'11}, pages 365--374. {IEEE}, 2011.

\bibitem[Fit02]{Fitting02}
M.~Fitting.
\newblock Modal logic between propositional and first-order.
\newblock {\em Journal of Logic and Computation}, 12(6):1017--1026, 2002.

\bibitem[FS01]{Finkel&Schnoebelen01}
A.~Finkel and Ph. Schnoebelen.
\newblock Well-structured transitions systems everywhere!
\newblock {\em Theoretical Computer Science}, 256(1--2):63--92, 2001.

\bibitem[FS09]{Figueira&Segoufin09}
D.~Figueira and L.~Segoufin.
\newblock Future-looking logics on data words and trees.
\newblock In {\em MFCS'09}, volume 5734 of {\em Lecture Notes in Computer
  Science}, pages 331--343, 2009.

\bibitem[GK03]{Gastin&Kuske03}
P.~Gastin and D.~Kuske.
\newblock Satisfiability and model checking for {MSO}-definable temporal logics
  are in {PSPACE}.
\newblock In {\em CONCUR'03}, volume 2761 of {\em Lecture Notes in Computer
  Science}, pages 222--236. Springer, 2003.

\bibitem[HMU01]{hopcroft:ullman}
J.~Hopcroft, R.~Motwani, and J.~Ullman.
\newblock {\em Introduction to Automata Theory, Languages, and Computation
  (2nd. ed.)}.
\newblock Addison Wesley, 2001.

\bibitem[HP79]{Hopcroft&Pansiot79}
J.~Hopcroft and J.J. Pansiot.
\newblock On the reachability problem for 5-dimensional vector addition
  systems.
\newblock {\em Theoretical Computer Science}, 8:135--159, 1979.

\bibitem[Jan95]{Jancar95}
P.~Jan{\v{c}}ar.
\newblock Undecidability of bisimilarity for {P}etri {N}ets and some related
  problems.
\newblock {\em Theoretical Computer Science}, 148:281--301, 1995.

\bibitem[Kos82]{Kosaraju82}
S.~Rao Kosaraju.
\newblock Decidability of reachability in vector addition systems.
\newblock In {\em STOC'82}, pages 267--281, 1982.

\bibitem[KST12]{Kara&Schwentick&Tan12}
A.~Kara, Th. Schwentick, and T.~Tan.
\newblock Feasible automata for two-variable logic with successor on data
  words.
\newblock In {\em LATA'12}, volume 7183 of {\em Lecture Notes in Computer
  Science}, pages 351--362. Springer, 2012.

\bibitem[KSZ10]{Kara&Schwentick&Zeume10}
A.~Kara, Th. Schwentick, and Th. Zeume.
\newblock Temporal logics on words with multiple data values.
\newblock In {\em FST\&TCS'10}, pages 481--492. LZI, 2010.

\bibitem[KV06]{Kupferman&Vardi06}
O.~Kupferman and M.~Vardi.
\newblock Memoryful {B}ranching-{T}ime {L}ogic.
\newblock In {\em LICS'06}, pages 265--274. IEEE, 2006.

\bibitem[Lam92]{Lambert92}
J.L. Lambert.
\newblock A structure to decide reachability in {P}etri nets.
\newblock {\em Theoretical Computer Science}, 99:79--104, 1992.

\bibitem[Laz06]{Laz06}
R.~Lazi{\'c}.
\newblock Safely freezing {LTL}.
\newblock In {\em {FST\&TCS}'06}, volume 4337 of {\em Lecture Notes in Computer
  Science}, pages 381--392. Springer, 2006.

\bibitem[Ler11]{Leroux11}
J.~Leroux.
\newblock Vector addition system reachability problem: a short self-contained
  proof.
\newblock In {\em POPL'11}, pages 307--316, 2011.

\bibitem[Lip76]{Lipton76}
R.J. Lipton.
\newblock The reachability problem requires exponential space.
\newblock Technical Report~62, Dept. of Computer Science, Yale University,
  1976.

\bibitem[LMS02]{Laroussinie&Markey&Schnoebelen02}
F.~Laroussinie, N.~Markey, and Ph. Schnoebelen.
\newblock Temporal logic with forgettable past.
\newblock In {\em LICS'02}, pages 383--392. IEEE, 2002.

\bibitem[LP05]{Lisitsa&Potapov05}
A.~Lisitsa and I.~Potapov.
\newblock Temporal logic with predicate $\lambda$-abstraction.
\newblock In {\em TIME'05}, pages 147--155. IEEE, 2005.

\bibitem[LS15]{Leroux&Schmitz15}
J.~Leroux and S.~Schmitz.
\newblock Demystifying reachability in vector addition systems.
\newblock In {\em LICS'15}, pages 56--67. {IEEE}, 2015.

\bibitem[May84]{Mayr84}
E.W. Mayr.
\newblock An algorithm for the general {P}etri net reachability problem.
\newblock {\em SIAM Journal of Computing}, 13(3):441--460, 1984.

\bibitem[May03]{Mayr03}
R.~Mayr.
\newblock Undecidable problems in unreliable computations.
\newblock {\em Theoretical Computer Science}, 297(1--3):337--354, 2003.

\bibitem[MR09]{Manuel&Ramanujan09}
A.~Manuel and R.~Ramanujam.
\newblock Counting multiplicity over infinite alphabets.
\newblock In {\em RP'09}, volume 5797 of {\em Lecture Notes in Computer
  Science}, pages 141--153, 2009.

\bibitem[MR12]{Manuel&Ramanujam12}
A.~Manuel and R.~Ramanujan.
\newblock Automata over infinite alphabets.
\newblock In D.~D'Souza and P.~Shankar, editors, {\em Modern applications of
  automata theory}, volume~2 of {\em IISc Research Monographs Series},
  chapter~2, pages 529--554. World Scientific, 2012.

\bibitem[NS11]{Niewerth&Schwentick11}
M.~Niewerth and Th. Schwentick.
\newblock Two-variable logic and key constraints on data words.
\newblock In {\em ICDT'11}, pages 138--149. ACM, 2011.

\bibitem[NSV04]{Neven&Schwentick&Vianu04}
F.~Neven, T.~Schwentick, and V.~Vianu.
\newblock Finite state machines for strings over infinite alphabets.
\newblock {\em ACM Transactions on Computational Logic}, 5(3):403--435, 2004.

\bibitem[Rac78]{Rackoff78}
C.~Rackoff.
\newblock The covering and boundedness problems for vector addition systems.
\newblock {\em Theoretical Computer Science}, 6(2):223--231, 1978.

\bibitem[Sav70]{Savitch70}
W.J. Savitch.
\newblock Relationships between nondeterministic and deterministic tape
  complexities.
\newblock {\em Journal of Computer and System Sciences}, 4(2):177--192, 1970.

\bibitem[Sch10a]{Schnoebelen10}
Ph. Schnoebelen.
\newblock Lossy counter machines undecidability cheat sheet.
\newblock In {\em RP'10}, volume 6227 of {\em Lecture Notes in Computer
  Science}, pages 51--75. Springer, 2010.

\bibitem[Sch10b]{Schnoebelen10b}
Ph. Schnoebelen.
\newblock Revisiting {A}ckermann-hardness for lossy counter machines and reset
  {P}etri nets.
\newblock In {\em MFCS'10}, volume 6281 of {\em Lecture Notes in Computer
  Science}, pages 616--628. Springer, 2010.

\bibitem[Seg06]{Segoufin06}
L.~Segoufin.
\newblock Automata and logics for words and trees over an infinite alphabet.
\newblock In {\em CSL'06}, volume 4207 of {\em Lecture Notes in Computer
  Science}, pages 41--57. Springer, 2006.

\bibitem[Wol83]{Wolper83}
P.~Wolper.
\newblock Temporal logic can be more expressive.
\newblock {\em Information and Computation}, 56:72--99, 1983.

\end{thebibliography}
\end{document}